\pdfoutput=1

\newif\ifdraft\draftfalse   % draft = comments
\newif\ifanon\anonfalse      % anon = light double-blind reviewing
\newif\ifcamera\cameratrue % camera = camera-ready version
\newif\iflongrefs\longrefsfalse % Long references (maybe for journal)
\newif\ifsooner\soonerfalse % maybe still before ICFP deadline
\newif\iflater\laterfalse   % at some point later, definitely not before ICFP deadline
\newif\iffull\fullfalse   % full = includes things that were cut from
\newif\ifneedspace\needspacetrue % even more painful cuts or very shady hacks
\newif\ifhighlightnewtext\highlightnewtextfalse % Highlight new text added/edited in response to reviewer feedback
\newif\ifallcites\allcitesfalse % Can't include all citations in camera-ready version, since we don't have space
\newif\ifbackref\backreffalse % backref option for hyperref; useful for shrinking references
\makeatletter \@input{texdirectives.tex} \makeatother

\ifcamera
  \documentclass[acmsmall]{acmart}
\else
  \ifanon
  \documentclass[acmsmall,review,anonymous,screen,nonacm]{acmart}
  \else
  \documentclass[acmsmall,review=false,screen=true,nonacm]{acmart}
  \fi
\fi

\AtBeginDocument{%
  }

\ifcamera
\setcopyright{rightsretained}   %% They ask for this only to hack things later; stupid
\else
  \ifanon
  \setcopyright{none}             %% For review submission
  \else
  \setcopyright{cc}               %% CH: for anything else Creative Commons
  \fi
\fi

\acmDOI{10.1145/3632916}
\acmYear{2024}
\copyrightyear{2024}
\acmSubmissionID{popl24main-p385-p}
\acmJournal{PACMPL}
\acmVolume{8}
\acmNumber{POPL}
\acmArticle{74}
\acmMonth{1}
\received{2023-07-11}
\received[accepted]{2023-11-07}

\usepackage{booktabs}   %% For formal tables:
\usepackage{subcaption} %% For complex figures with subfigures/subcaptions
\usepackage[normalem]{ulem}

\definecolor{dkblue}{rgb}{0,0.1,0.5}
\definecolor{dkgreen}{rgb}{0,0.4,0}
\definecolor{dkred}{rgb}{0.6,0,0}
\definecolor{dkpurple}{rgb}{0.7,0,1.0}
\definecolor{purple}{rgb}{0.9,0,1.0}
\definecolor{olive}{rgb}{0.4, 0.4, 0.0}
\definecolor{teal}{rgb}{0.0,0.4,0.4}
\definecolor{azure}{rgb}{0.0, 0.5, 1.0}
\definecolor{gray}{rgb}{0.5, 0.5, 0.5}
\definecolor{dkgray}{rgb}{0.3, 0.3, 0.3}

\newcommand{\comm}[3]{\ifdraft{{\color{#1}[#2: #3]}}\fi}
\newcommand{\ch}[1]{\comm{teal}{CH}{#1}}
\newcommand{\tw}[1]{\comm{purple}{TW}{#1}}
\newcommand{\ca}[1]{\comm{dkblue}{CA}{#1}}
\newcommand{\er}[1]{\comm{blue}{ER}{#1}}
\newcommand{\gm}[1]{\comm{azure}{GM}{#1}}
\newcommand{\et}[1]{\comm{dkred}{ET}{#1}}

\newcommand{\meta}[1]{\ifdraft{{\color{dkgray}[#1]}}\fi}

\newcommand{\remove}[1]{\ifdraft\sout{#1}\fi}

\usepackage{xspace}

\newcommand*{\EG}{e.g.,\xspace}
\newcommand*{\IE}{i.e.,\xspace}

\newcommand\fstar{F$^{\star}$\xspace}
\newcommand\sciostar{SCIO$^{\star}$\xspace}

\newcommand{\cmparrow}{\hspace{-0.35em}\downarrow}
\newcommand{\cmp}[1]{#1\cmparrow}
\newcommand{\bakarrow}{\hspace{-0.35em}\uparrow}
\newcommand{\bak}[1]{#1\bakarrow}

\usepackage{listings}
\usepackage[framemethod=tikz]{mdframed}
\newcommand\maybecolor[1]{\color{#1}}

\ifneedspace
 \definecolor{light-gray}{gray}{0.95}
\else
\fi

\newcommand{\ls}[1]{\lstinline{#1}}
\newcommand{\lss}[1]{\lstinline[basicstyle=\footnotesize]{#1}}

\usepackage{combelow} % For romanian letters
\usepackage{tikz}
\usepackage{wrapfig}
\usepackage[inline]{enumitem}
\newlist{inlist}{enumerate*}{1}
\setlist[inlist]{label=(\arabic*)}

\usetikzlibrary{shapes,arrows}
\tikzset{%
  block/.style    = {draw, thick, rectangle, minimum height = 3em,
    minimum width = 3em},
}

\def\Snospace~{\S{}}

\newcommand{\citeFull}[2]{\ifallcites\cite{#1,#2}\else\cite{#1}\fi}

\renewcommand{\paragraph}[1]{\ifneedspace\else\smallskip\fi{\bf #1}\;}
\usepackage[skip=0pt, belowskip=-10pt]{caption}

\newcommand{\newtext}[1]{\ifhighlightnewtext{\color{dkgreen}#1}\else#1\fi}
\newcommand{\newremove}[1]{\ifhighlightnewtext{\color{dkred}\sout{#1}}\fi}

\begin{document}

\hypersetup{
pdftitle={Securing Verified IO Programs Against Unverified Code in F*},
          % Somehow the * still doesn't show up
}

\title{Securing Verified IO Programs Against Unverified Code in \fstar}

\ifanon
\author{}
\else
\author{Cezar-Constantin Andrici}
  \authornote{First author.} % since Cezar would also come first in fully alphabetical order
  \affiliation{\institution{MPI-SP}\city{Bochum}\country{Germany}}
  \email{cezar.andrici@mpi-sp.org}
  \orcid{0009-0002-7525-2440}
\author{\cb{S}tefan Ciob\^{a}c\u{a}}
  \affiliation{\institution{Alexandru Ioan Cuza University}\city{Ia\cb{s}i}\country{Romania}}
  \email{stefan.ciobaca@gmail.com}
  \orcid{0009-0000-4082-570X}
\author{C\u{a}t\u{a}lin Hri\cb{t}cu}
  \affiliation{\institution{MPI-SP}\city{Bochum}\country{Germany}}
  \email{catalin.hritcu@mpi-sp.org}
  \orcid{0000-0001-8919-8081}
\author{Guido Mart\'{i}nez}
  \affiliation{\institution{Microsoft Research}\city{Redmond}\state{WA}\country{USA}}
  \email{guimartinez@microsoft.com}
  \orcid{0009-0005-5831-9991}
\author{Exequiel Rivas}
  \affiliation{\institution{Tallinn University of Technology}\ifcamera\city{Tallinn}\fi\country{Estonia}}
  \email{exequiel.rivas@ttu.ee}
  \orcid{0000-0002-2114-624X}
\author{\'{E}ric Tanter}
  \affiliation{\institution{University of Chile}
  \department{Computer Science Department}
  \city{Santiago}\country{Chile}}
  \email{etanter@dcc.uchile.cl}
  \orcid{0000-0002-7359-890X}
\author{Théo Winterhalter}
  \affiliation{\institution{Inria Saclay}\ifcamera\city{Saclay}\fi\country{France}}
  \email{theo.winterhalter@inria.fr}
  \orcid{0000-0002-9881-3696}
\fi

\ifanon\else
\renewcommand{\shortauthors}{Andrici et al.}
\fi

\begin{abstract}
We introduce \sciostar{}, a formally secure compilation framework for statically
verified programs performing input-output (IO).
The source language is an \fstar{} subset in which a verified
program interacts with its IO-performing context via a higher-order interface
that includes refinement types as well as pre- and post-conditions about past IO events.
The target language is a smaller \fstar{} subset
in which the compiled program is linked with an adversarial context that has an interface
without refinement types, pre-conditions, or concrete post-conditions.
To bridge this interface gap and make compilation and linking secure we propose
a formally verified combination of higher-order contracts and reference
monitoring for recording and controlling IO operations.
Compilation uses contracts to convert the logical assumptions the
program makes about the context into dynamic checks on each context-program
boundary crossing.
These boundary checks can depend on information about past IO events stored in
the state of the monitor. But these checks cannot stop the
adversarial target context {\em before} it performs dangerous IO operations.
Therefore linking in \sciostar{} additionally forces the context to perform all
IO actions via a secure IO library, which uses reference monitoring to dynamically
enforce an access control policy before each IO operation.
We prove in \fstar{} that \sciostar{} soundly enforces a global trace property
for the compiled verified program linked with the untrusted context.
Moreover, we prove in \fstar{} that \sciostar{}
satisfies by construction Robust Relational Hyperproperty Preservation,
a very strong secure compilation criterion.
Finally, we illustrate \sciostar{} at work on a simple web server example.

\iflater
\meta{WIP: we use parametricity to show formally in \fstar{} that this gives us a form of noninterference}.
\fi
\end{abstract}

  


%\received{11 July 2023}
%\received[revised]{26 October 2023}

\maketitle

\section{Introduction}
\label{sec:intro}

Increasingly realistic programs have been written and verified in
proof-oriented programming languages like Coq\ifallcites~\cite{coq}\fi,
Isabelle/HOL\ifallcites~\cite{NipkowPW02}\fi,
Dafny\ifallcites~\cite{leino10dafny}\fi,
and \fstar{}\ifallcites~\cite{mumon}\fi{}
for obtaining strong static guarantees of correctness and
security~\citeFull{compcert,certikos,record,haclstar,evercrypt,everparse,DBLP:conf/cpp/ArasuRRSFHPR23,
  linear-dafny,veribetrkv,BhargavanB0HKSW21, HoPBB22,
  ZakowskiBYZZZ21, Appel16, MurrayMBGBSLGK13, KleinAEHCDEEKNSTW10}{steel,vale,galapagos}.
% \ca{I cited CompCert, CertiKOS, EverCrypt (it also includes HACL*), miTLS, EverParse, FastVer2, what else?}\ch{For
%   a start some realistic programs verified in Coq and Dafny}
%
One way to speed up the development process in such languages is to use existing
unverified libraries written in more mainstream languages.
% to which they compile, such as
%   OCaml~\cite{Let2008}, C\#~\cite{dafny-doc} or C~\cite{lowstar}.\ch{In general one can use
%     libraries written in any other languages, even if one doesn't compile to
%     those languages; so as written now this seems overly specific?
%     If I just remove ``to which they compile'' though we do lose some information
%     that's still useful below. This is the specific setting in which
%     we are working, but maybe we don't need to commit to that {\bf so early}?}%
%\ch{To what language(s) does Dafny
%  extract by the way? Any citation we can give for that?}\ca{Dafny extracts to
%  Java, C\#, Go, JavaScript, Python, C++. My impression was that the extraction
%  to C\# was the original one and the most solid. Not sure what to cite,
%  the documentation? \url{https://dafny.org/latest/DafnyRef/DafnyRef}}\ch{Yes}
%
In such cases, the program is usually verified under some {\em assumptions}
about the libraries,
% \ch{Yet below we say it's possible to make no
%   assumptions, which contradicts this. So I added the ``usually'', but may need
%   better solution, since it's unclear what's the most usual case.}%
  % \ca{I want to explain here that one has to write down new assumptions about the
  % behavior of the library}\ch{This is fine, but not all assumptions are expressed explicitly,
  %   and we may even be explicit about that: See tentative start of next phrase,
  %   which would bringing back one of the babies that were thrown with the bathwater.}\ca{We
  %   do nothing related to implicit assumptions, isn't it?
  %   In this paper, we avoid having to deal with implicit assumptions
  %   because our source
  %   and target language have the same semantics. I'm not sure how much it is
  %   worth to mention implicit assumptions and I think it would be best
  %   to stick to the explicit ones.}\ca{maybe we can argue that
  %   converting the explicit ones is challenging enough}%
  %   \ch{I thought this paragraph is about the general problem now though,
  %     and I think this is an important part of the general problem. In our specific setting,
  %     we make no implicit assumptions, which is a good way to solve this kind of problem.
  %     I anyway don't think that a reader will worry what we do about implicit assumptions
  %     from the first paragraphs. This is introduced as general context,
  %     not as a specific problem we solve in this work.}\ca{+1}
%
and then the verified program is compiled and linked with the unverified libraries.
The assumptions about the linked libraries can be expressed
explicitly, as specifications~\cite{CokL20}, or implicitly---\IE assuming that the
libraries respect the abstractions of the language of the verified
program and \EG do not have undesirable side effects. In either case,
these assumptions are necessary for the verification of any global
property characterizing the behavior of the program together with the
libraries.
%
%\ch{One new thing in what you started writing is that you try to stay generic in
%  the proof assistant (although Dafny hardly fits this denomination), while
%  below we commit to F* from the first sentence, but then go a bit back and
%  forth and mention Coq too. Another one is that you try to list successes of
%  formal verification, which we didn't do explicitly below, but would indeed be
%  good to do at least for \fstar{}. It may be the case that turning this into a
%  more general list, including \EG CompCert, is a good idea and it's not getting
%  too clichéd; at least if we believe that
%  what we do has some connection to CompCert -- for instance could the
%  translation validation that CompCert does for unverified passes like register
%  allocation (written in OCaml) at least in principle benefit from the kind
%  of technology we build here?}
%
However, current compilers from proof-oriented languages (\EG from Coq and \fstar{}
to OCaml~\cite{Let2008,SozeauBFTW20} or C~\cite{lowstar,certicoq17,Paraskevopoulou21},
or from Dafny to C\#\iffull~\cite{dafny-doc}\fi,
or from \fstar{} and Dafny to Assembly~\cite{vale, valefstar}) 
% naively assume -- CH: Eric complained (found it offensive), also the word assume/assumption is already taken here
take for granted that the linked unverified code cannot break the
assumptions on which the verified program relies. This is unsound and thus provides no
global guarantees after compilation and linking.
More importantly, this is insecure if---as we assume in this paper---the linked
code is {\em untrusted} (\EG vulnerable, compromised, or malicious).

% \ca{``compromised'' means that it is ``vulnerable'', no? seems like a repetition}
% CH: compromised is worse than vulnerable, they are increasingly bad, no repetition

\iflater
\ch{[Note to self, not actionable yet]
  I feel that a few babies were thrown out with the bathwater in the text
  commented out below: we were illustrating that in our higher-order setting it
  is more difficult to minimize/eliminate assumptions about the context compared
  to the first-order setting, especially when the program is the library. Now
  that's gone, which simplifies indeed the story, but also weaken the problem. Could we
  maybe return to this point when we introduce the program is library setting below?}
\ch{That got replaced with a focus on global properties; and it's indeed the case
  that even first-order libraries can't obtain global guarantees. So probably fine?}
\fi

To ensure security, one could statically verify the linked code (\EG library) to
prove that it cannot be compromised and all the assumptions we make about it
  hold~\citeFull{Melocoton, SammlerSSDKGD23, IrisWasm, AhmedMWA22}{Bowman21,
  Bowman18, KoronkevichRAB22, SchererNRA18},
% \ifsooner\ch{Melocoton, DimSum, etc are fine,
%   but can we find any multi-language program verification paper where one of the
%   languages is that of a proof assistant, which would be even closer to what we do.}\fi{}
% CH: added papers by Amal and William
but this would likely require much effort,
as static verification most often involves user interaction and expertise,
thus taking away the simplicity of just using the unverified code.
% to speed up development.
%
% While one cannot avoid static verification to prove the implicit assumptions
% about the unverified code\ca{I have no idea if this is a founded claim,
%       first try in focusing only on explicit assumptions},
% for the explicit assumptions
The alternative is to enforce the assumptions at runtime by converting them into dynamic
checks~\citeFull{evercrypt, haclstar}{libcrux}.
In this paper we assume that one is willing to (re)write the linked code in the
proof-oriented language, so that all assumptions are explicit. However, to
speed up development, one wants to enforce the explicit assumptions
about linked unverified code dynamically instead of proving them statically.

In this paper we introduce \sciostar, which provides
a {\em formally verified} way to systematically convert into
dynamic checks all the assumptions made by verified \fstar{}~\ifallcites\else\cite{mumon} \fi programs with IO
about linked unverified \fstar{} code with IO---\EG reading and writing files and
network sockets.\iflater\ch{could add (here and/or in contributions or in key ideas?) that this could also be
  instrumentation (output) events, but first we would need an example showing that}\fi{}
We look at this problem through the lens of secure
compilation~\cite{AbateBGHPT19, MarcosSurvey}: {\bf \sciostar{} is
a formally secure compilation framework for IO programs}
consisting of a compiler and a linker
between two languages shallowly embedded in \fstar{}.

% CH: Eric's version before I reverted, and manually applied some of the changes
% \ch{the transition/flow here seems worse to me now; also introducing secure compilation framework
%   before saying that we look at this through the lens of secure compilation seems like the wrong order
%   (TODO: might need to revert of rephrase)}%
% In this paper we introduce \sciostar, a formally secure compilation framework
% for statically verified partial programs performing input-output (IO)---\EG reading and writing files and
% network sockets\ifsooner\ch{could add (here and/or in contributions or in key ideas?) that this could also be
%   instrumentation (output) events, but first we would need an example showing that}\fi.\ch{Also
%   our contexts can do IO, but now that got lost.}
% We look at this problem through the lens of secure
% compilation~\cite{AbateBGHPT19, MarcosSurvey}: \sciostar{} consists of a compiler and a linker
% between two languages shallowly embedded in \fstar{}.

%\( \begin{array}{ll}\textit{add higher-}\\\textit{order contracts}\end{array} \)

\begin{figure}
\begin{tikzpicture}[auto, thick, node distance=5.5cm]
\draw
    node at (0, 0) {}
    node [block] (vps) {\( \small \begin{array}{cc}\textit{verified}\\\textit{program}\end{array} \)}
    node [below of=vps, node distance=0.8cm, align=center] (vps1) {\tiny strong\\[-1.5ex]\tiny interface}
    node [block, right of =vps] (vpi) {\( \small \begin{array}{cc}\textit{compiled}\\\textit{program}\end{array} \)}
    node [below of=vpi, node distance=0.8cm, align=center] (vpi1) {\tiny intermediate\\[-1.5ex]\tiny interface}
    node [block, right of =vpi] (upt) {\( \small \begin{array}{cc}\textit{unverified}\\\textit{context}\end{array} \)}
    node [below of=upt, node distance=0.8cm, align=center] (upt1) {\tiny weak\\[-1.5ex]\tiny interface}
    ;
\draw[-] [thin] (1.5, 0.75) -- (5.5, 0.75);
\draw[-] [thin] (7.75, 0.75) -- (11.75, 0.75);
\draw node at (6.625, 0.75) {\tiny target language};
\draw node at (0.0, 0.75) {\tiny source language};
\draw[->](vps) -- node [above] { \small {\bf compile} } node [below,align=center] { \small add higher-\\[-1ex]\small order contracts } (vpi);
\draw[-](vpi) -- node [above] { \small {\bf link} } node [below,align=center] { \small add reference\\[-1ex]\small monitor } (upt);
% Boxing and labelling noise shapers
\draw[-] [color=gray,thick, dashed] (1.25,-0.75) -- %% node [color=black, left, at end] {\tiny source language}
%% node[color=black, right, at end] {target language}
(1.25,0.75);
\end{tikzpicture}
\caption{\label{fig:overview}An overview of \sciostar.\vspace{-1em}}
\end{figure}

The {\em source language} is an \fstar{} subset in which a verified
{\em partial source program} interacts with 
a {\em source context} via a {\em strong higher-order interface}, consisting of
specifications expressed as refinement types and pre- and post-conditions.
Refinement types are used to constrain the values of a base type
with a logical formula (\EG the value of an integer is larger than zero).
%
%\ch{don't think a reader will understand that by pure you mean it's not stateful;
%   even more important would be to explain that those predicates are used to constrain
%   the values of the argument and return types, not that they are pure.
%   Could also mention / illustrate on an example
%   that a refinement on a function result can depend on the function arguments (dependent function types).
%   Just that it's not currently supported in the implementation.}
% \ca{Mention that we do not support dependent refinements.
% Like, we do not support \ls{x:int -> y:int{x < y} -> int}}\ch{No beating up out own work here!
%   You can mention it some technical section if needed.}
%
% \ch{The IO specific part only starts here:}
%\ca{Can one have refinement types without dependent types?}
%\fstar{} is also a dependently-typed language in which
The pre- and post-conditions of a
function can depend on its arguments, can specify its result, and
can also specify its IO behavior by considering
the trace of past IO events---each
time an IO operation is performed, an event
containing the arguments and the result is appended to the trace.
A pre-condition can constrain the trace of past IO events at the time the
function is called (\EG it could require that a file is currently open by
looking back in the trace for an open event for the corresponding file
descriptor and making sure that no close event happened afterwards for this descriptor).
A post-condition can additionally take into account the IO events produced by
the function itself when it returns (\EG the function closed all file
descriptors that it opened) and can also specify the relation between the result
value and the return-time trace (\EG the returned value was read from a file).

The {\em target language} is a smaller \fstar{} subset
% \stefan{Why smaller?}\ch{The sentence
%   continues and explains what's removed / not allowed compared to the source.}
in which the compiled program
% \remove{, which has an {\em intermediate interface},}\ch{too early?}
is linked with an adversarial {\em target context}
that has a {\em weak interface} without refinement types, pre-conditions, or
concrete post-conditions
% was: without refinement types or pre- and post-conditions
% \ch{This explanation of weak interface is easy to understand, but
%   the part about post-conditions is technically not true (the context is
%   parametric in $\Sigma$, which acts as a post-condition). The same formulation is
%   used in a few more places (abstract, next sentence). Still wondering if we can
%   rephrase this in a way that's technically more true, but still easy to understand.}
(see \autoref{fig:overview}).
The partial source program and the target context can only interoperate
securely when their interfaces match, which is not the case in our setting because
the strong interface of the verified program contains concrete specifications expressed
as refinement types and pre- and post-conditions,
while the weak interface of the context does not.
% \ch{I'm not yet fully convinced
%   this last sentence is ideal here, but at least now the whole para is technically
%   correct, while before it was not:
%   \ca{confusing sentence. the compiled program does not have a strong interface,
%      the compiled program has an intermediate interface}}
%
% \ch{was this better phrased before? -- not really
%   also need to be careful here, since our actual target language is quite surprising
%   (parametric in an access control policy spec expressed as a post-condition!)}
%  CH: dropped *any* before any context specs; the rest is not worse than before

\sciostar{} enables a verified partial source program to be securely compiled
and linked against an arbitrary target context.
% \ch{Not completely sure about the bold here.
%   By putting so much emphasis here we may actually be deemphasizing the fact that our
%   framework is itself verified. It was said before, but only in passing, and not in bold.
%   Any chance we can still add that here? Or emphasize it elsewhere?}
% CH: tried going for this by using some bold and some italics in the para introducing SCIO*
% CH: in the end, dropped bold from this paragraph too
\sciostar{} supports both the case in which the untrusted context is a library used by the program and the case in which the program is a library used by the untrusted context.
To achieve this secure interoperability \sciostar{} uses a combination of
\begin{inlist}
  \item\label{itm:ref-mon} \emph{reference monitoring}~\citeFull{Anderson73}{UNIX, Ames81}, introduced 
  by the linker, and
  \item\label{itm:hoc} \emph{higher-order contracts}~\citeFull{FindlerF02}{DisneyFM11, ScholliersTM15, MooreDFFC16, TovP10}, introduced by the compiler.
\end{inlist}
These techniques have, as far as we know, not been applied so far to our setting,
where we are
protecting statically verified code from unverified code, and where we aim for
strong formal security guarantees.
We use reference monitoring and higher-order contracts to bridge the gap between the strong interface of the
source program and the weak interface of the untrusted target context by giving the
compiled program an {\em intermediate interface} in between (see \autoref{fig:overview}).

\ref{itm:ref-mon} The reference monitor records in its state information about the
trace of IO events that happened so far during the execution.
The monitor uses this information to enforce an {\em access control policy}
by performing a dynamic check before each IO operation of the context
% to decide whether it can be allowed or not
(\EG preventing access to a
file descriptor associated with a password file or network socket).
\sciostar{} implements this by linking the target context with a secure
IO library that dynamically enforces the policy.
This policy enforcement on each IO operation of the context is necessary, because
checking some post-conditions when the context returns would be too late to
prevent bad IO events from happening.
For example, a post-condition of the context could state that it has neither
accessed the passwords file, nor the network. If we only detect a violation of
this policy after the context returns, the damage has already been done, with
passwords leaked over the network. \newtext{Instead, we use the monitor to prevent
illegal IO events from actually happening, and the trace is then unaffected
as we do not record such prevented events.}
However, reference monitoring is not enough on its own,
since the strong interface contains refinement types and
some pre- and post-conditions that cannot be enforced at the level of IO
operations, but have to be enforced on the higher-order boundary between the program
and the context (\EG the pre-condition of a callback sent to the context
can require that its argument is an open file descriptor).

\ref{itm:hoc} Higher-order contracts are used by \sciostar{} during
compilation to convert
the assumptions on the boundary between the partial program and the context into
dynamic checks.
The \sciostar{} compiler wraps the program in new functions with weaker types
and adds dynamic checks
before each function call and after each function return crossing the boundary
between program and context.
This dynamically enforces refinement types on arguments and results,
% \ch{Has anyone done this before?}\ca{at least Eric,added citation}
as well as pre- and some post-conditions of functions.
To enforce pre- and post-conditions related to IO behavior,
the higher-order contracts access the state of the monitor to get information about
the IO events that happened before the function was called and during the
execution of the function.

\iflater
\ch{We say ``need'' here, but an interesting thought experiment is what exactly
  would break in our proofs if we allowed the untrusted target context
  unrestricted access. For the IO operations linking would break, since we
  wouldn't be able to strengthen the interface of the context?  And at a less
  syntactic level, soundness would break, right? What about GetMState in the
  target though? We didn't yet prove noninterference, and even that was IIRC for
  source, not target contexts?}\ch{Maybe it's too early here (unless we find a
  very intuitive explanation), but this is the kind of insights the key ideas
  section could/should convey!}
\fi

We are the first to secure verified IO programs against unverified code,
  and moreover we provide strong machine-checked security guarantees.
For this we make the following {\bf contributions}:
% This article % \tw{Or `We'?} -- CH: I like the current less personal / more professional phrasing
% describes the following {\bf contributions}:
% SC: It is not the paper making the contributions :), I think "we" is better.
% \paragraph{Contributions:}

% \ch{For each item below, please explicitly call out what's novel, interesting,
%   useful, challenging, etc. Just because we did something it doesn't mean it's a
%   contribution to science, so please explain why each item is listed here.}

\begin{itemize}[leftmargin=12pt,nosep,label=$\blacktriangleright$]

\item We introduce \sciostar{}, an \fstar{} framework for securely compiling a
verified partial IO program and linking it against an IO context with a
weak higher-order interface.
We bridge the interface gap between the verified program and the untrusted context
by using a combination of higher-order contracts and a reference monitor,
which share state that records information about prior IO events.
% \ch{As also mentioned in the abstract, the related work section will need to justify
%   this novel combination claim}
%
%One novel idea for this combination is
%that we require the post-conditions of the context to be split so that each
%of them is dynamically checked either by the reference monitor on each IO
%operation of the context, or by the contracts when the context returns to the
%program. \ch{The way we phrased this we're giving a lot of weight to this one
% idea, and last reviewer wasn't even convinced (look at our answer for
% inspiration). In any case, if we want to keep this here (do we? my vote would
% currently go to keeping it for key ideas)  we need to
% rephrase it to not get so much emphasis, even compared to everything else we
% say in this paragraph. And hopefully we can find some more (convincing) ideas?}
% \ca{I think the intuition that the word 'split' gives is kind of wrong,
%    and indeed, it does not seem important enough to brag about it here}
%
Moreover, \sciostar{} takes advantage of the program being statically
verified and performs no dynamic checks when the program performs IO or when it
passes control to the context.

\item
We statically verify that \sciostar{} adds enough dynamic
checks to guarantee that the program and the context can interoperate securely.
This verification includes the implementation of higher-order contracts, the
wrapping of the context's IO library with the monitor's access control checks,
% of the reference monitor, -- no space for this
as well as their combination.
%
% the correctness of user-provided
% % \ch{{\bf the ``user-provided'' part no longer explained}, which creates
% %   bad flow; there was a full phrase about this that was not duplicate}\ch{Or maybe
% %   we can just rely on the people's common sense and only explain this in key ideas?}
% access control policy and contract checks,
% and the correct combination of these parts.
% \ch{There is one extra part
%   below that just got commented out: {\bf we do prove that the policy
%   is correctly {\em enforced} by the checks we add}, no?}
%, our enforcement of the access control
%policy, and the fact that all these parts work well\ca{``well'' is vague. ``correct'' or ``soundly'' would fit better} together.
%
% We stated this intrinsic\ch{no clue what ``intrinsic'' is supposed to mean here; can we drop it?}
%   property as a soundness theorem,\ch{\bf no longer explaining what soundness is
%   (global trace property; see abstract and commented out text); this seems wrong!}
%which states that the
%global trace property that holds in the source language,
%also formally holds in the target language because of the added dynamic checks.
Moreover, we prove in \fstar{} the soundness of the entire \sciostar{} framework,
showing that a global trace property holds for the compiled
program linked with the untrusted context.
% access control policy on the
% context in combination with static verification of the program soundly enforces\ca{i'm not
%  sure that the ``combination'' ``enforces'' something here.
%  one proves statically that the global trace property holds in the source,
%  and the compilation framework enforces/preserves the property in the target by using dynamic checks.}
% a global trace property.
This proof is done modularly by typing and also heavily benefits from \fstar{}'s SMT
automation\ifallcites~\cite{dm4all,mumon,metafstar}\fi.
%which is modular and which in \fstar{}
% is a form of program verification that
% \meta{It sounds like we're explaining what typing is. Feels odd.} CH: fixed

\item 
% \ch{Wouldn't it be better to put this contribution {\bf before} the previous one?
%   And the same for the following one? Or was there a good reason to move soundness to the top?}%
% We define our source and target language by using shallow
% embeddings.\stefan{This does not sound like much of a novel contribution.}\ch{+1, It
%   no longer starts well after Cezar's changes, but this may be fixable? For instance,
%   it could be said before the contributions, unless we can claim it as a contribution?
%   We define two new languages? :)}
We represent computations
% \ca{we don't explain what a computation is. what would a reader understand by
% it? maybe we can avoid this keyword here and explain it later}\tw{I would
% expect POPL audience to have at least a vague intuition. For me this looks
% fine, what do others think?} CH: looks fine to me too
in our shallowly embedded source and target languages by using a
new monadic effect \ls{MIO}---\IE Monitored IO---that
is at its core a way to statically verify terminating IO programs,
% (based on the DM4All idea \tw{We cannot say it like that I think.})
engineered to take advantage of SMT automation in \fstar{}.
In addition to usual IO operations
(\EG reading and writing files and network sockets\iflater\ch{maybe also instrumentation
  (output) events?  see comment from paragraph 3; especially if we add that there an
  alternative is to drop the paren here}\fi),
\ls{MIO} contains a \ls{get_mstate} operation,
which models in a simple abstract way access to the reference monitor state,
% provide a simple model for recording ...
% \tw{It provides an abstraction but not a model right?}\ch{
%   Facepalm. I don't think normal readers will think of model in terms of model theory.
%   Dictionary definition of model: a simplified description,
%   especially a mathematical one, of a system or process to assist calculations and predictions.}
% \ch{I mean, this is POPL, not normal readers, and also we don't actually have a
%   model to justify our claims so I wouldn't want to draw unnecessary
%   attention. If you still disagree then remove my comment.}
% \ch{Replaced ``provides a simple model'' with ``models in a simple way in the hope
%   that no one is confused by our standard use of the word model/models in English
%   and science in general.}
%
% that hides the implementation details,\ch{
%   These implementation details are not part of our formally secure compilation framework,
%   so this may be confusing the reader (TODO 03.A, 03.B). For a local fix we
%   could simply drop the ``which provides a simple abstract model [...] that hides
%   the implementation details''?}
% (\EG initialization and state updates), -- CH: doesn't sound super complex or impressive
and which allows us to implement the dynamic checks done by
the monitor and higher-order contracts.
In addition, at the specification level we distinguish between events produced
by the program and those produced by the context, which enables enforcing a
stronger specification on the untrusted context.

\item To model in a simple way that our \newtext{shallowly embedded}
contexts cannot directly access the IO
operations and the monitor's internal state we propose a novel use of
flag-based effect polymorphism.
The flag is an index of the \ls{MIO} monad that controls which operations a
computation can access.
%To prevent the untrusted context from directly calling the IO operations and \ls{GetMState}
%we index the \ls{MIO} monad by an extra flag that controls
%which operations a computation can access and we propose a novel use of flag-based
%effect polymorphism to model the partial program and the context.
Flag-based effect polymorphism is necessary since \fstar{} gives up on more general kinds of effect
polymorphism % \ch{citations needed}
in order to gain from SMT automation.
\newtext{In addition to this shallow embedding of contexts, we introduce a
syntactic representation of target contexts in a small deeply embedded
language that can be translated into the shallow embedding.}
\iflater
\meta{WIP: To formalize that a flag-based effect polymorphic context can't directly
call the IO actions + GetMState we use parametricity(?) to prove in \fstar{} a
noninterference property showing that the context cannot obtain more information
about the trace than what can be obtained by calling IO operations that the
access control policy allows. This noninterference property moreover takes into
account that the success or failure of the dynamic checks of our reference
monitor and higher-order contracts necessarily reveal some information about the
trace to the context.}
\fi
% \ch{... also what the external inputs reveal to it. This is not very obvious, so
%   not sure whether we want to mention it {\em here}, but we also assume that
%   whatever the context can \EG read from files is public, even if the program
%   just wrote secrets to those same files. It seems that a stricter policy can
%   usually mitigate this problem.}
% CH: tried to include above, but may need to return to this
%     when discussing noninterference in more detail

\item We show that \sciostar{} is secure by providing a machine-checked proof in \fstar{}
that it satisfies Robust Relational Hyperproperty
Preservation (RrHP), which is the strongest secure compilation criterion of
\citet{AbateBGHPT19}, and in particular stronger than
full abstraction\iffull~\cite{MarcosSurvey,AbateBGHPT19}\fi.
Intuitively this ensures that \sciostar{} provides enough protection to the
compiled program so that linked target contexts do not enjoy more attack power
than a source context would have against the original source program.
% \tw{Why is it RrHP and not RRHP? Usually smaller letters are used when they
% are conjunctions or actually part of the word whose first letter is already part
% of the acronym.}\ch{It's a complicated naming scheme that we introduced in
%   \citet{AbateBGHPT19}, and it's not like we can take that back now.}
%
While proofs of such % strong secure compilation
criteria are generally challenging~\cite{MarcosSurvey, DevriesePPK17, AbateBGHPT19, NewBA16, JacobsDT22},
we have carefully set things up in \sciostar{}
so that our proof is simple by construction, in particular because:
% Such a very simple proof is possible in our setting since:
\begin{inlist}
  \item our languages are shallowly embedded in \fstar{};
  \item we used flag-based effect polymorphism to model the context;
  \item we designed our higher-order contracts mechanism so that we can define
  both compilation and back-translation in a way that they
  satisfy a syntactic inversion law that
  immediately implies RrHP (\IE compiling the program and linking it with the
  context is syntactically equal to back-translating the context and linking it
  with the program).
\end{inlist}
% \ch{From our call at the beginning of April, it seems that ``cancellation law''
%   is not the right word here, but better for what Cezar started doing for
%   compiler correctness. inversion law?}
%
We prove RrHP % this secure compilation theorem
for both the case in which the
context is a library used by the program and the case in which the
program is a library used by the context.
%
% \iflater
% \ch{Most likely too late for ICFP, but our case would be even stronger if this
%   simplicity would still be true even if we were to also consider the setting in
%   which we swap roles between program and context, by having the context take
%   the program as argument. We discarded that (more common for secure
%   compilation!) setting early on because we couldn't do effect polymorphism, but
%   now we can, and that dual setting would probably change things at least a bit
%   for point 2 above (back-translation). Maybe points 1 and 3 will still help us
%   though? Could we speculate a bit about this in future work?  Otherwise we risk
%   someone who knows secure compilation to bash us badly that the only reason our
%   proof is simple is because we are in a non-standard setting (see PriSC reviewer
%   A's first bullet puzzling exactly about this).}\ca{I thought about this in the past and things
%   should still be simple. I can give this a quick go in \fstar to make sure.}
% \ch{Would be nice at some point, but there are more urgent issues for ICFP.}
% \fi

\item We illustrate \sciostar{}
% in practice -- CH: Seems like a stretch to me, and anyway is already at odds with "simple".
%                    I think that, in such cases it's better to underclaim and then maybe pleasantly surprise,
%                    than to overclaim and then disappoint reviewers.
on a main case study in which we statically verify a simple web server in \fstar{},
compile it to our target language, % using our secure compiler
and % then
link it against some adversarial request handlers and
a non-adversarial request handler serving files.
This illustrates that our languages are expressive enough to write
interesting code and that the \ls{MIO} monadic effect enables effective verification.
 % simplifies verification.\ch{Unsure if simplifies
 %  is the right word here (unclear with respect to what it's simpler). Enables effective verification maybe?}
 % of the web server by taking advantage of SMT automation.

% \ch{Are we watering down our formal contribution by talking too much about
%   extraction to OCaml? It's not clear how much we should be talking about
%   extraction in this contribution, given how defensive we need to be about it.
%   Tried to limit the damage by rephrasing, and let's see if that's already enough.}

% \ca{\bf what would the reviewers expect the case study to do?}
% \tw{They will expect it to evidence that it can be used in a rather realistic  setting.}
% \ch{Remains to be seen whether using phrasing like ``works in practice''
%   is actually justified by the realism of the web server. I can tell
%   you after I read the section about it :)}
% \ca{I'm not sure what realistic means. I went for ``these examples of adversarial request handlers
%   are blocked by our mechanism''. Is that realistic?}

\end{itemize}

\paragraph{Outline.}
We start by illustrating the key ideas of
\sciostar{} on the verified web server case study
(\autoref{sec:key-ideas}).
% \tw{Weird segue here.}\ch{no clue what you're trying to say}\tw{I meant the 
% transition between the two sentences was abrupt I guess? Not sure any more.}
%
% The following two sections present % the main technical pieces of
%   % our work:\ch{This text needs rewriting after removing section 5.}
We then present the \ls{MIO} effect (\autoref{sec:mio}) and
the implementation of higher-order contracts (\autoref{sec:contracts}).
%\ifgottheorem\else,
%and the precise relation between the monitor state and the IO event trace
%(\autoref{sec:abstract-state})\fi.
%
We put these pieces together to define % the
\sciostar{} % framework
and prove it soundly enforces a global safety property and
it satisfies RrHP (\autoref{sec:secure-compilation}).
\newtext{
Next, we explain
% an experiment with syntactic contexts before explaining -- CH: duplicating reference below
how we execute the web server in OCaml (\autoref{sec:running-case-study})
and illustrate a few more examples (\autoref{sec:other-examples}).
}%
Finally, we discuss related
(\autoref{sec:related-work}) and future work (\autoref{sec:conclusion}).
%\autoref{sec:mio-examples} presents an additional example for \ls{MIO}

\newtext{
\paragraph{Longer-term goal.}
Achieving strong secure compilation criteria such as RrHP is extremely
difficult, and no realistic compilation chain that achieves such criteria has
ever been built~\cite{MarcosSurvey,AbateBGHPT19}.
We see our work as an important step towards building a formally secure
compilation chain from \fstar{} to a safe subset of OCaml.
This is a formidable research challenge though, so we carefully limited the
scope of the current paper by focusing on IO as the single side effect and by
assuming for now that the unverified context is also written in a shallowly
embedded subset of \fstar{}, but only respects a weak interface.
% (in addition we also give a simply-typed syntactic representation to contexts in
% \autoref{sec:syntactic-contexts}).
%
In \autoref{sec:conclusion} we discuss how these assumptions could be lifted in
the future to achieve formally secure compilation from \fstar{} to OCaml.
}

\section{\sciostar{} in Action}
\label{sec:key-ideas}

\begin{figure}[h]
\vspace{-1em}
\begin{mdframed}[backgroundcolor=black!5,hidealllines=true]
\begin{lstlisting}[numbers=left,escapechar=\$]
type req_handler =
  (client:file_descr) ->
  (req:buffer{valid_http_request req}) -> $\label{line:req_ref}$
  (send:(res:buffer{valid_http_response res} -> MIO (either unit err) $\label{line:send_start}$
                                            (requires (fun h -> did_not_respond h))
                                            (ensures (fun _ _ lt -> exists r. lt = [EWrite _ client r])))) -> $\label{line:send_stop}$
  MIO (either unit err) (requires (fun h -> did_not_respond h)) $\label{line:hndlr_pre}$
                      (ensures (fun h r lt -> (wrote_to client lt \/ Inr? r) /\ $\label{line:hndlr_post_wrote}$
                                         handler_only_opens_and_reads_files_from_folder lt /\ $\label{line:hndlr_post_only}$
                                         web_server_only_writes lt) $\label{line:hndlr_post_web}$
let web_server (handler:req_handler) :
    MIO unit (requires (fun h -> True)) (ensures (fun _ _ lt -> every_request_gets_a_response lt) = $\label{line:srvr_spec}$
  let s = socket () in setsockopt s SO_REUSEADDR true; bind s "0.0.0.0" 3000; listen s 5; ... $\label{line:srvr_init}$
  let client = select_client s in $\label{line:srvr_select}$
  let req = get_req client in $\label{line:srvr_get_req}$
  if Inr? (handler client req (write client)) then sendError 400 client; $\label{line:srvr_handle}$
  close client; ...
\end{lstlisting}
\end{mdframed}
\caption{\sciostar{} case study: web server that takes a request handler as argument. The types are simplified.}
\label{fig:running_example}
\end{figure}

\ifsooner

\ch{Also I think it would be good to more clearly identify in the text what the
  presented key ideas are. What's exciting, novel, cool about what we're
  presenting here. Sometimes the text seems to focus on boring standard stuff
  (2.5) or gory technical details (2.6), without explaining first what the main
  point is. Why should one pay attention in the first place?}

% \ch{In Sec 4, the fact
%   that the source context is also effect polymorphic is new information.
%   Maybe fine given the reference to \autoref{sec:rrhp}, but if this is a key idea,
%   we may not wait until \autoref{sec:rrhp} and could maybe bring it up briefly in 2 already?
%   This is the second time in section 4 (first time it was in 4.1 when indexing
%   interim by a flag) where we expose that effect polymorphism
%   is more pervasive than section 2 currently mentions.}

% CA: paper got accepted at POPL, so I think it is fine for now
% \ch{Finally, at the very least the text in this section should be understandable
%   to the broad POPL audience (think conference talk). While some improvements were
%   made, there is still quite a bit of work left (starting with 2.2).}
\fi

\noindent
We illustrate % \ch{ some of?}
the key ideas of \sciostar{} using our main case study as a running example:\footnote{% \paragraph{\fstar syntax primer.}
\fstar syntax is similar to OCaml (\lss{val}, \lss{let}, \lss{match}, etc).
Binding occurrences \lss{b} take the form \lss{x:t}
% , declaring
% a variable \lss{x} at type \lss{t}; 
or \lss{#x:t} for an implicit argument. We omit the type in a binding when it can be inferred. 
Lambda abstractions are written
\lss{fun b_1 ... b}$_n$\lss{ -> t}
(where \lss{t} ranges over both types and terms), whereas
\lss{b_1 -> ... -> b}$_n$\lss{ -> C} denotes a curried function
type with result \lss{C}, a computation type describing the effect,
result, and specification of the function.
Contiguous binders of the same type may be written
\lss{(v_1 ... v}$_n$\lss{ : t)}.
Refinement types are written \lss{b\{t\}} 
(\EG \lss{x:int\{x>=0\}} represents natural numbers).
Dependent pairs are written as \lss{x:t_1 & t2}.
The \lss{squash t} type is defined as the refinement \lss{_:unit\{t\}},
and can be seen as the type of computationally-irrelevant proofs of
\lss{t}.
%
% and write \EG \lss{fun x -> x+1} for the function incrementing an integer \lss{x}.
%
For non-dependent function types, we omit the name in the argument binding, \EG type
\lss{#a:Type -> (#m #n : nat) -> vec a m -> vec a n -> vec a (m+n)}
represents the type of the
append function on vectors,
where both unnamed explicit arguments and the return type depend on the
implicit arguments. % marked with `\lss{#}'.
%
% We mostly omit implicit binders, except when needed for clarity,
% treating all unbound variables in types as prenex quantified, writing
% the type of append as just \lss{vec a m -> vec a n -> vec a (m + n)}.
%
A type-class constraint \lss{\{| d : c t1 .. tn |\}} is a special
kind of implicit argument, solved by a tactic during elaboration.
%
%They can also be provided explicitly.
%
\lss{Type0} is the lowest universe of \fstar{}; we also use it to write propositions,
including \lss{True} (True) and \lss{False} (False).
% \tw{Is it really needed to say that \(\top\)is True? It's very standard notation.}
% CH: what's standard is very subjective; it's still not standard to me for instance
%
We generally omit universe annotations.
The type \lss{either t_1$~$t_2} has two constructors, \lss{Inl:(#t_1 #t_2:Type) -> t_1 -> either t_1$~$t_2}
and \lss{Inr:(#t_1 #t_2:Type) -> t_2 -> either t_1$~$t_2}.
Expression \lss{Inr? x} tests whether \lss{x} is of the shape \lss{Inr y}.
Binary functions can be made infix by using backticks:
\lss{x `op`$~$y} stands for \lss{op x y}.
}
the partial program is a simple web server that is verified
in \fstar{} to respond to every request, while the unverified context is a request handler
represented as a higher-order function.
The web server has the initial control and it gets the request handler as
an argument.
We carefully crafted the case study to be higher-order and to contain the
interesting types of specifications \sciostar{} supports.
\newtext{The case study is still moderately realistic, since we can run our web server
  and it can handle HTTP requests from a real browser (\autoref{sec:running-case-study}).}
\newtext{We also applied \sciostar{} to a few other examples we discuss in
  \autoref{sec:other-examples}.}
% -- more aspects -- CH: way too vague, so removed for now

We start with the strong interface of the web server, how it is verified (\autoref{sec:key-verified}),
and the assumptions that % the web server
it makes about the handler, \newtext{which is our adversarial target context
(\autoref{sec:assumptions})}.
We then explain how the intermediate interface bridges the gap between the
web server and the handler (\autoref{sec:framework-overview}),
before presenting the weak interface of the handler (\autoref{sec:key-context}).
We then introduce the executable checks and access control policy the
\sciostar{} requires to enforce the specification (\autoref{sec:key-checks}).
% \ch{This is already a first place that gives the impression that there is
%   {\em zero} automation about the checks, even in the case the monitor state is
%   the trace. Reviewer A bashed us badly, and we rebutted about this. Let's at
%   least future proof these kind of claims; we will not always be silly
%   about conceptually obvious automation I hope.}\ca{replaced ``user provided''
%     with ``requires''. not sure what to say about the rest of the comment}
% CA: I could not find any other places that made similar claims
%
We use the checks to weaken the assumptions the web server makes about the
untrusted handler by using higher-order contracts (\autoref{sec:key-contracts}).
Finally, we strengthen the type of a target handler using reference monitoring
to enforce the policy (\autoref{sec:key-monitor}).

\subsection{The verified partial program (web server)}
\label{sec:key-verified}
\autoref{fig:running_example} illustrates the most important part of the web server's implementation.
It starts by opening a TCP socket (line \ref{line:srvr_init}) and then,
inside an elided terminating loop,\footnote{\newtext{Our framework supports recursion, as
  long as \fstar{} can prove termination. General recursion is future work though (\autoref{sec:conclusion}).}}
%\ca{can we say instead of "loop", "recursion"?}
%\ch{+(elided)? \bf otherwise where is the terminating loop in the code?}
it accepts clients and waits for incoming data in a non-blocking way.
The web server waits for requests from multiple clients at the same time.
%
%\ca{is this enough for people to understand
%  that the web server can handle multiple clients in the same time?}\ch{\bf no;
%  it would be better to make it explicit}
For each client that sends data (line \ref{line:srvr_select}), the web server reads it and validates the HTTP request (line \ref{line:srvr_get_req}),
and then it passes the request to the handler (line \ref{line:srvr_handle}).
If the handler failed to respond to the request---\IE it returned an error value (tagged with \ls{Inr})---then
the web server responds with an error (\ls{400}) to the client.
%
%The handler is a higher-order function that takes the request as argument
%together with the file descriptor of the client and a callback.

The web server takes as argument a handler of type \ls{req_handler} and its
result type (line \ref{line:srvr_spec}) is that of an \ls{MIO} computation that
returns a unit and has a trivial pre-condition (indicated by \ls{requires})
and a post-condition (indicated by \ls{ensures}).
The pre- and the post-condition are part of the type,
as indices of the \ls{MIO} monadic effect.
Similarly, the type \ls{req_handler} also contains many interesting refinement
types and pre- and post-conditions.
We highlight the following specifications that appear on the types of the web
server and handler:

\begin{enumerate}[leftmargin=*,nosep]
\item\label{spec:web_server_post}
The post-condition of the \ls{web_server} (line \ref{line:srvr_spec}) ensures that it
responds to all accepted clients.
% It does this either via the \ls{handler}
% or if that fails by responding itself with an error (\ls{400}).
% CH: this was already said above

\item\label{spec:handler_pre}
The first two arguments of the handler are a file descriptor \ls{client}
and a buffer \ls{req} that is guaranteed to contain a valid HTTP request
thanks to the refinement type (line \ref{line:req_ref}).
The \ls{handler} also has as pre-condition that no response has been sent yet to the
latest request (line \ref{line:hndlr_pre}).
\iflater
\ca{There was the proposal to point out that expensive redundant checks can be
    avoided because of this predicate. Not sure if it is that worth since the handler is
    untrusted.}\ch{Didn't understand what you try to say here, but look at
    valid HTTP request, not at the response}
    \ch{Could have a look at LowParse's ``zero-copy'' parsers
    (\url{https://www.microsoft.com/en-us/research/publication/everparse/}),
    where checking that a message is well-formed was indeed expensive,
    and then what one would get in the end was ways to access the
    original buffer.}
\fi

\item\label{spec:send}
The third argument of the \ls{handler} is a callback \ls{send} (lines \ref{line:send_start}-\ref{line:send_stop})
that expects as argument a buffer that contains a valid HTTP response
and requires that no response has been sent yet to the latest
request---\IE the same pre-condition as the \ls{handler}.
The concrete callback for \ls{send} is passed by the web server and it
simply writes the response to the client.

\item\label{spec:handler_post_ctrs}
The type of the \ls{handler} ensures that either it wrote to the file
descriptor it got as argument or it otherwise returns an error value 
(tagged with \ls{Inr}, line \autoref{line:hndlr_post_wrote}).

\item\label{spec:handler_post_acp}
The type of the \ls{handler} also ensures that it
only opens files from a specific folder, reads only from its own opened files, and closes only them (line \ref{line:hndlr_post_only}).
It also ensures that during its execution
only the trusted \ls{web_server} can write (so only when the \ls{handler} calls the \ls{send} callback, line \ref{line:hndlr_post_web}).
Any other IO operation of the \ls{handler}
or the \ls{web_server} is not permitted by the post-condition.
\end{enumerate}
These specifications describe how the web server behaves
and also what assumptions the web server makes about the request handler.
%\ca{not
%  all of them are assumptions about the handler, some of them are
%  ``promises''?}\ch{That seems already covered by the ``how the web server behaves''
%  part. I see no problem here.}
%
% Some of\ch{Why only some?}
The assumptions about the request handler are necessary to verify
that the web server satisfies its post-condition.
In particular, without the post-condition of the handler ensuring that
it writes to the client when it returns a success value (tagged with \ls{Inl}),
one would not be able to verify that the web server responds to every request.
\newtext{The post-condition of the web server is defined using the following predicate
that checks whether every request (read from a file descriptor) is followed at some point by a
response (write to the same file descriptor).}
\begin{lstlisting}
let every_request_gets_a_response (lt:trace) : Type0 =
  let rec aux = (fun lt read_fds -> 
    match lt with
    | [] -> read_fds == []
    | ERead (fd, _) (Inl _) :: tl -> aux tl (fd :: read_fds)
    | EWrite (fd, _) _ :: tl -> aux tl (filter (fun fd' -> fd <> fd') read_fds)
    | _ :: tl -> aux tl read_fds) in
  aux lt []
\end{lstlisting}
The proof that the web server satisfies its post-condition
is done mostly automatically with the help of the SMT solver.
The implementation is split into several functions, which
makes verification more modular by feeding smaller verification
conditions to the SMT solver.
We needed to state and prove a few lemmas (some using interactive proofs~\cite{metafstar}),
% which are essentially properties
about the 
\ls{every_request_gets_a_response} predicate, which had to be proven by
induction on the trace, something that the SMT solver cannot do on its own.
However, once these lemmas were proven, the SMT solver was able to exploit them to
prove the specifications of the various parts of our web server without further
manual intervention. This is thanks to the SMT automation \fstar{} provides
for monadic effects and to our careful design of the \ls{MIO} monadic effect
(\autoref{sec:mio}).

\newtext{
\subsection{The assumptions about the context (handler)}
\label{sec:assumptions}
%\ca{what would be a better name?}\ch{A less terribly vague one,
%  like ``[...] assumptions the source program makes about the context''}

The strong interface of the web server includes the type \ls{req_handler},
which defines the expected specification of the handler}
%the type of the handler which contains the
%assumptions made about it
(the other components of
a strong interface are given in \autoref{sec:compilation-framework}
but they are not yet relevant here).
% CH: very lengthy way to say you don't want to show something
% complete definition of a
% strong interface is given in \autoref{sec:secure-compilation}, but until then it is enough to
% imagine that the strong interface is the higher-order type \ls{req_handler}.
%
A traditional compiler that just erases specifications\footnote{As mentioned in
  \autoref{sec:intro}, current extraction mechanisms from proof-oriented
  languages~\cite{Let2008, lowstar, dafny-doc, vale, valefstar} just
  erase specifications, even the extraction mechanisms that
  are formally verified to be correct~\cite{SozeauBFTW20, certicoq17, Paraskevopoulou21}.}
would for instance convert \ls{req_handler}
into a type without refinement types and without pre- and post-conditions:
% \ch{too technical addition to this sentence follows, needed at all?}process
% during which \ls{MIO} also loses its non-computable-relevant indices:
\begin{lstlisting}
type too_weak_handler_type =
  file_descr -> buffer -> (buffer -> MIO (either unit err)) -> MIO (either unit err)
\end{lstlisting}
However, it would be unsound and insecure to directly link an arbitrary inhabitant of this
type to the partial program.
The simplest example of an adversarial handler that breaks soundness
is the one that immediately returns a success value.
This breaks specification \ref{spec:handler_post_ctrs} of the web server that expects the handler
to write to the client at least once when it returns a success value.
\begin{lstlisting}
let adversarial_handler1 client req send = Inl ()
\end{lstlisting}

\newtext{
  To securely compile the web server, 
  it is required to add dynamic checks to enforce the assumptions about the context---\IE
  specifications \ref{spec:send}, \ref{spec:handler_post_ctrs} and \ref{spec:handler_post_acp}.}
\sciostar{}
enforces specification \ref{spec:send} and \ref{spec:handler_post_ctrs} using higher-order contracts and
enforces specification \ref{spec:handler_post_acp} using an access control policy.
The post-condition of the handler (\ref{spec:handler_post_ctrs}-\ref{spec:handler_post_acp}) is the most interesting 
because it is enforced in part by higher-order contracts and in part by reference monitoring.
The first part of the post-condition---\IE checking whether the
handler wrote to the client (\ref{spec:handler_post_ctrs})---can be soundly enforced by a contract
when the handler returns, but not
using reference monitoring at the level of the IO operations.
The second part of the post-condition \ref{spec:handler_post_acp} specifies the IO behavior of the handler,
which cannot be securely checked by a contract when the handler returns.
% \tw{
%   To me this last sentence reads weird. "it cannot be checked only when."}\ca{dropped the "only". is it better?}
% TW: Ah! Now I understand what the only wanted to convey. It might be better
% without to avoid confusion?
%
The post-condition \ref{spec:handler_post_acp} for instance specifies that the handler does not open files outside of a specific folder: if
the untrusted handler opens a password file from a different folder we could detect that when it gives
back control to the web server;
however, this still breaks our post-condition since the
bad event has already happened and it is too late to do anything about it.
%\ca{This
%  was explained also in the introduction, in the paragraph about reference
%  monitoring (marked with (1)). Can we avoid repeating this?}
%
Therefore, we prevent the violation from happening
by using reference monitoring---\IE using an access control policy that
enforces the two predicates of specification \ref{spec:handler_post_acp} (more in \autoref{sec:key-monitor}).
%\ch{But what is 
%  the access control policy for our concrete example? Without at least
%  some intuition, the connection to the example is lost.}\ch{Now one has to wait
%  to \autoref{sec:key-monitor} for any information on this.}
%
%Thus, we require that post-conditions can be split
%into a dynamic check to be enforced at the end of the function using contracts
%and the access control policy enforced by the reference monitor on each
%IO operation of the context.
%\ch{Returning
%  to this it's not clear where what happens to specs 3-5 is {\em actually} explained.
%  Should it be in this subsection? Or in some future subsection, in which case
%  some forward references would definitely help.}

% \ca{this paragraph has to be integrated better, but I don't know how for now}%
% CH: looks fine to me
Because the web server is statically verified,
we do not have to dynamically enforce
anything when the server passes control to the handler%
%because these specifications are statically verified
---\IE the pre-condition of the handler, or that the web server
passes a valid HTTP request to the handler (\ref{spec:handler_pre}).
This also includes when control \emph{returns} from
the web server into the handler, such as when the \ls{send} callback
returns after being invoked by the handler---no dynamic enforcement is needed because
\ls{send} is verified to satisfy its post-condition.
% \ch{Missing a number now; is this 3?\ca{no, by 3
%   we refer to the pre-condition and refinements of \ls{send}}}\ch{
%   Returning to this, could it be that we didn't yet talk about the post-condition of send?\ca{that's
%     a good catch, but not sure where to talk about it because I don't think we ever return to it}}

% \ca{silent switch between ``intermediate interfaces'' and ``intermediate types''}%

\subsection{Bridging the gap with the intermediate interface}
\label{sec:framework-overview}

\sciostar{} combines reference monitoring for recording and controlling IO events
together with higher-order contracts that have access to the monitor state.
Together they
% The two techniques
bridge the gap between the strong interface of the verified program and the
weak interface of the untrusted context---\EG % intuitively% , they bridge the gap
between the \ls{req_handler} type and something like the \ls{too_weak_handler_type} with erased
specifications above.
% \ch{Not clear at all
%   this point what is the type of the
%   untrusted context, even if it's crucial for understanding anything here. A good
%   intuition is the type from the last paragraph of 2.1, even if it's technically
%   not fully accurate.}
%we use % higher-order
%contracts during compilation to weaken the interface of the program
%and we use monitoring during linking to strengthen the interface of the target
%context,
The bridging produces a middle point % between the two
we call the {\em
  intermediate interface} (\autoref{fig:overview}), which both the compiled
partial program and the monitored target context share.

% is done by bringing the strong interface of the verified program
% and the weak interface of the unverified context in the middle at a point
% we call an

The meeting point between what our higher-order contracts do not enforce
and what the reference monitor enforces is the access control policy.
Therefore, an intermediate interface contains types % that are
annotated with no specifications except that each function has a
post-condition stating that it satisfies the access control policy.
We denote this by using the short notation \ls{MIO a True}~$\Sigma$,
where \ls{True} says that there is no pre-condition and
$\Sigma$ is the specification of the access control policy encoded
as a post-condition.
For instance, the intermediate interface version of \ls{req_handler} looks as follows:
\begin{lstlisting}
type intermediate_handler_type =
  file_descr->buffer->(buffer->MIO (either unit err) $\top$ $\Sigma$)->MIO (either unit err) $\top$ $\Sigma$
\end{lstlisting}
We define compilation as converting a program with a strong interface
into a program with this kind of intermediate interface. Therefore, the compilation
of the web server produces a function that expects a request handler
of \ls{intermediate_handler_type}. The % concrete
$\Sigma$ in this type
is presented in \autoref{sec:key-monitor}---\EG it prevents
the handler from opening files outside a specific folder.

% The target context has a weak interface described in \autoref{sec:key-context}. % below.
% During linking we add the reference monitor
% to the context, which gives it an intermediate interface,
% thus the monitored handler also has \ls{intermediate_handler_type}.
% This is because we can give the monitored context the specification that it satisfies
% the access control policy. % ---\EG the concrete access control policy $\Sigma$ above
\newtext{
  The target context gets an intermediate interface during target linking,
  when the reference monitor is added, because the monitor enforces that the context
  satisfies the access control policy.}
Since monitoring is done at the level of IO operations, to monitor the target
context it is enough to link it with
% to replace the standard IO library with -- CH: this is also formally wrong
a secure IO library
that dynamically enforces the access control policy.
% So our target linking passes
% this secure IO library only to the target context, which makes the context monitored.
% CH: doesn't add anything, running around your own tail 
%
On the other hand, the compiled program can call directly the default IO operations that only
record the IO events in the monitor state, but performs no dynamic checks.

\subsection{The weak interface of the unverified context (handler)}
\label{sec:key-context}

% \ch{Should we maybe call any abstract policy $\Alpha$ instead of any concrete
%   $\Sigma$, so as to distinguish the two?}
\newtext{
  The context gets its specification from the reference monitor,
  which enforces the access control policy for the IO operations.}
Therefore, we make explicit the dependency of the unverified context on the IO library by
modeling the unverified context as being {\em parametric} over the library and also
over the abstract access control policy enforced by the library.
This definition captures the intuition that ``the IO library gives
specification to the context'' because the abstract specification of the
context is given by what IO operations the context uses,
and if it uses only operations that enforce an abstract access
control policy, then we can show that the context satisfies the abstract policy.
%
% For a second intuition,
% \ch{flow: second now? now 2.2 uses \ls{too_weak_handler_type} for a first intuition}
So intuitively, the type of the unverified request handler looks more like this in \sciostar{}:

\begin{lstlisting}
type still_too_weak_handler_type = $\Sigma$:erased _ -> io_lib $\Sigma$ ->
                file_descr -> buffer -> (buffer -> MIO (either unit err) $\top$ $\Sigma$)) -> MIO (either unit err) $\top$ $\Sigma$
\end{lstlisting}
where
\ls{io_lib}~$\Sigma$ is the type of the IO libraries that enforce a policy that has abstract specification $\Sigma$.

% \ifsooner
% \ch{Quite a bit of back-and-forth between the $\Sigma$ and the flag.
%   Could we maybe reorder reduce context switches?}
% \fi

The weak interface clarifies that the abstract specification of the
context depends on the % \ca{+concrete} CH: no, it's not concrete in the weak type on the context!
specification enforced by the secure IO library.
The fact that an abstract specification is still part of the weak interface does not invalidate the intuition
that the context is unverified because the weak interface is parametric in this specification.
One cannot verify any concrete specification when type-checking the context's code,
but it is trivial to show that it inherits the abstract specification of the secure IO library.
We also have further evidence that our weak interfaces are
expressive enough to represent unverified contexts:
\begin{inlist}
  \item the implementation of our handlers, including the larger non-adversarial
  one that serves files;
  \item a syntactic representation of simply-typed target contexts and a %total 
  translation
  from any syntactic expression to our shallow embedding % that uses weak interfaces
  (see \autoref{sec:syntactic-contexts}).
\end{inlist}

The benefit of defining weak interfaces like this in \sciostar{} is that after
instantiation the monitored context has a specification in its type, which
greatly simplifies the soundness proof of our secure compilation chain
(\autoref{sec:soundness}).

This type would, however, not prevent the handler from directly calling the IO operations
provided by the \ls{MIO} monadic effect. In proof-oriented languages
like Coq one could simply quantify over the monadic effect~\cite{dm4all}---\IE a computation
monad indexed by a specification monad (see \autoref{sec:mio-essence}
for details)---and use such effect polymorphism to prevent access to IO operations.
However, \fstar{} doesn't allow quantifying over effects
because it needs the concrete implementation of the specification monad
to generate verification conditions for the SMT solver.
Therefore, we were inspired from work by \citet{BrachthauserSO20} on {\em
  parametric effect polymorphism} and combined their idea with the idea of
indexing a monadic effect with a flag~\cite{indexedeffects}.

So, for a start, we added a flag index to the \ls{MIO} effect.
A flag value is simply an inhabitant of a variant type \ls{tflag} with four
constructors: \ls{NoOps}, \ls{GetMStateOps}, \ls{IOOps}
and \ls{AllOps}.
By indexing a computation with an element of this type we can
restrict the operations that the computation can perform:
\begin{inlist}
  \item \ls{NoOps} means that no operation can be used in the computation, % (thus so only lift pure values),
  \item \ls{GetMStateOps} means that only \ls{get_mstate} operation can be used,
  \item \ls{IOOps} means that only IO operations
  (open, read, etc.) can be used, but not \ls{get_mstate}
  \item \ls{AllOps} means that the computation can access any of
  the operations.
\end{inlist}
%
%A computation parametric in the flag would not be able to call any
%\ca{better flow needed}%
Because the target context is parametric in the flag, it prevents calls to
any of the operations as the flag could
take the value \ls{NoOps}. We refer to this form of effect polymorphism as
\emph{flag-based effect polymorphism}.
% \tw{make making making, probably can be improved}\er{gave it a try} TW: Thanks!
%\er{Should we make explicit the fact that this polymorphism is kind
%  of \emph{parametric}?}\ch{Isn't the use of parametric below enough?}%
%
%Because the target context is parametric in the flag it cannot directly call
%the IO operations, and instead to perform any IO it has to call
%our secure IO operations, which we pass to it as argument. \tw{< Feels like a repetition to the above.}
%
So the actual weak type of an unverified handler looks like
this:
%\ch{``more like this'' still sounds like we're still cheating, but isn't
%  this it?  What's missing now? Can we just replace ``weak interface'' with
%  ``weak type'' and no longer water down this claim?}
\begin{lstlisting}
type weak_handler_type = fl:erased tflag -> $\Sigma$:erased _ -> io_lib $\Sigma$ fl ->
          file_descr -> buffer -> (buffer -> MIO (either unit err) fl $\top$ $\Sigma$)) -> MIO (either unit err) fl $\top$ $\Sigma$
\end{lstlisting}
where \ls{fl} is the flag (marked as \ls{erased} so only usable
at the specification level)
% unusable in the
% computation\tw{Should we point out that this is the reason why the parametricity
%   argument works? Because we cannot branch on it.}\ca{can you try to add
%   this?}\tw{I wonder now. Since we won't do it,\ch{what don't we do? we do mark
%   the flag as erased; I guess you mean we don't prove it?}\tw{yes}
%  maybe we should instead talk about it in future work?})
%  CH: okay, rephrased it positively (only usable at the spec level)
%
which is also an index of \ls{io_lib}.
%
%This is a key idea of our paper where\remove{ both the problem
%and} the solution is specific to \fstar{}.
%
The instantiation of the target context happens during target linking,
where the flag- and specification-polymorphic context is passed
the \ls{AllOps} flag, the concrete specification of an access control policy,
and the secure IO operations that enforce the concrete access control policy.

\ifsooner
\ch{Maybe also quickly mention that we don't only do \ls{fl:erased _ -> ... -> MIO _ fl ...}
  but also \ls{fl:erased _ -> ... -> MIO _ (fl + IOOps)  ...} for instance.
  So flag-based effect polymorphism is more flexible than just what we use for contexts,
  and we will use that flexibility in other places.}
\fi

The use of flag-based effect polymorphism is key to correctly
model the attack capabilities of the source and target contexts,
enabling us to achieve and prove RrHP (\autoref{sec:rrhp}).

% The use of flag-based effect polymorphism is key to correctly modeling
% attack capabilities in source and target contexts for achieving and proving RrHP.

\subsection{Providing the dynamic checks and the access control policy}
\label{sec:key-checks}

\newtext{
  Before discussing about how the dynamic checks are enforced, we first note that
}%
since some \fstar{} specifications are not directly executable,
\sciostar{} requires the dynamic checks to be used inside the
higher-order contracts and the access control policy to be enforced by the reference monitor.
The \fstar{} type checker
% \tw{Is it a [typing] discipline? For me discipline is more
% about `good practices' and not about a compiler / type checker actually 
% complaining. Might be wrong though.}
asks 
for \remove{all (and only) }the necessary dynamic checks,
%\ch{I would also expect that our type class
%  instances would already be enough for some simple/generic cases, but we didn't even
%  mention the word type class yet. Are we afraid of it for some reason? 
%  It's not like the \fstar{} type system does magic without any help from
%  the \sciostar{} type classes / instances :)
%  Also we shouldn't give the impression that our framework has 0 automation
%  and it's all terribly manual.}
and it also verifies that they indeed imply the specification.
Since this verification is done in \fstar{} it takes advantage of its support
  for SMT automation and interactive proofs~\cite{metafstar}.
\sciostar{} also leverages type classes in our implementation of higher-order contracts which
gives more automation in figuring out the dynamic checks, which works great on
simpler cases.
While one can make arbitrary assumptions about the context, not all of them have
sound dynamic checks that are both \emph{precise enough} and \emph{efficiently enforceable},
thus it becomes a design decision on what assumptions are made and what dynamic checks are picked.
For our example, \sciostar{} requires dynamic checks that imply the pre-condition of \ls{send} (\ref{spec:send})
and the first part of the post-condition of the handler (\ref{spec:handler_post_ctrs}),
as well as an access control policy that implies the other part (\ref{spec:handler_post_acp}).
%\ch{Spoiler alert! That
%  post-conditions are split is explained in the next paragraph. Is this last phrase
%  coming too soon? Wouldn't it better belong at the end of the next paragraph?}

While our specifications use traces, these can be impractical
% when doing reference monitoring or higher-order contracts
for dynamic enforcement because they grow indefinitely.
Instead, \sciostar{} gives the ability to pick a different state for the
monitor as long as it abstracts the trace, and efficient custom checks
over this abstract state that don't have to traverse a linear trace.
% A user would likely benefit from the ability to define their own notion 
% of state that they could optimize to be more efficient than constantly
% traversing the trace.
%
%\remove{For our example, to abstract the trace the monitor state has to contain enough
%information to enforce the two dynamic checks (for 3 and 4) and
%the access control policy (5).}\ch{This ``for example'' is currently
%  completely useless, since it just repeats what was already said. Let's be
%  more concrete about what some the monitor state is in the example, since
%  otherwise I'll drop this completely.}%
%
We discuss more about how we do this in \autoref{sec:abstract-state}, but until then,
in what follows we assume that the monitor stores the trace of events
to simplify the technical details.
%
%\ch{That should probably already be introduced in 2, but in any case, it would
%  be good to explain what the monitor state is for our web server, and then use
%  that here to motivate this section.}\ch{Only when reading on I discovered that
%  5.1 is exactly about this. But we need good motivation for this section,
%  and also a good outline (see next comment).}
%
%\ch{Naturally the intuition of this should come much earlier, in Section 2,
%  since it is a key idea, and also as motivation for this section (5)}%

\subsection{Compilation using higher-order contracts}
\label{sec:key-contracts}
The \sciostar{} compiler converts a program with a strong interface
into a program with an intermediate interface. For example,
compilation of the web server wraps it
into a new function that takes as argument an intermediately-typed
handler.
The intermediately-typed handler then is made strongly-typed (of type \ls{req_handler})
by enforcing higher-order contracts on it.
\sciostar implements higher-order contracts by using two dual functions:
\ls{export} converts strongly-typed values into intermediately-typed values, and
\ls{import} converts intermediately-typed values into strongly-typed  values.
The \ls{import} and \ls{export} functions are based on the types of the values
and defined in terms of each other, which is needed for supporting higher-order functions~\cite{FindlerF02}.

%taken as argument is imported to the strong type \ls{req_handler}.
Import of the intermediately-typed handler wraps it inside of a new function
of type \ls{req_handler},
in which first the strongly-typed arguments are exported.
The most interesting is the export of \ls{send}:
export wraps it into a new intermediately-typed function that
enforces the refinement type and \ls{send}'s
pre-condition (\ref{spec:send}). % ---the only pre-condition we have to enforce dynamically in the example.
Not enforcing the refinement or the pre-condition of \ls{send} would be unsound, situation in
which \fstar{} would not even let us call \ls{send}, since it would
enable the handler to respond to the same client more than once or to call it with an invalid HTTP response:
%\ch{Unclear what the response below is and why it's invalid. As is it only adds
%  confusion.}\tw{I think it calls send with a too big argument (it's supposed to
%  be <500)}\ch{I find this unconvincing (not just confusing). Does it have
%  anything to do with HTTP validity?}\tw{I see what you mean. I think what we call valid HTTP request in the case study is very simple (<500). Cezar would know better.}
\begin{lstlisting}
let adversarial_handler2 client req send = send (Bytes.of_string "hello") // invalid HTTP response
\end{lstlisting}
In case the dynamic checks of the refinement and of the pre-condition
%\ch{what dynamic checks? far reference, should be more precise (currently bad flow)}
succeed, the \ls{send} function is called and its result is exported and returned,
otherwise if one of the check fails then
%\ch{not clear enough what this ``otherwise'' is about (bad flow)}
\ls{Inr Contract_failure} is returned.
This is why the return type of \ls{send} is \ls{either unit err}, meaning that it 
can fail.
%
%\ch{Are you explaining the key ideas here, or only boring details and pointing
%  out all limitations of our work so that any reviewer will know why to reject
%  such a boring and full of limitations paper. So far the text doesn't seem to
%  call out any reason a reviewer would get (positively) excited about this
%  work. Or did I miss any enthusiastic statement of some great idea?}
%
Thus, when a contract fails (or and when the monitor prevents a bad IO operation),
an error
%\ch{if you want to call this error than the type name should change accordingly (err not err)}
is returned, which allows the program and the context to handle
the errors.\footnote{\newtext{We allow for error recovery, since it would be
    unacceptable for our web server to crash whenever a dynamic check fails.
    Our secure compilation theorem allows the information flow entailed by
    error recovery, as discussed in \autoref{sec:rrhp}.}}

After exporting the strongly-typed arguments, 
the intermediately-typed handler is called with them and an intermediately-typed result is obtained.
\newremove{In our case, the result type is just \ls{either unit err}, thus the import is just the identity function.
However, before returning this as the result of the imported handler,}\newtext{
  After importing the result, but before returning it,
  we have to dynamically enforce the post-condition}---\newtext{\IE that} it wrote to the client.
This dynamic check prevents the previously presented \ls{adversarial_handler1}
attack to work
%\ch{``work'' as in mount its attack, right? should clarify that,
%  since ``work'' is generally seen as a positive not a negative thing :)}
because after being imported, instead of returning \ls{Inl ()},
it returns the error \ls{Contract_failure},
%\ch{Already the 3rd name for the same thing in
%  just a few sentences (err, error, now failure).  I've seen this kind of error
%  countless times in your writing and talks and it makes it unnecessarily
%  hard for readers/audience to follow.}
thus falling into the branch of the 
web server that checks whether the handler ran into an error and responds with HTTP error 400, ensuring 
that every request gets a response (\ref{spec:web_server_post}).

% \ch{Explaining all this in words in the last two paragraphs is kind of fine, but
%   I'm starting to wonder whether additionally showing some code for the exported
%   web server (potentially simplified by for instance unfolding away all the type
%   class magic for this particular example), would make it easier to follow all
%   this. Or is this code too complicated and/or anyway coming up later?}
% \ch{This time (after the rewrite) no longer had as big an urge to see code after
%   reading this. Now that I think about it, it would might have helped get less
%   lost in the first paragraph of this sub-section. But then, better writing could
%   help too.}\ch{Improved writing in the first paragraph.}

Since \sciostar{} gets a proof that enforcing this dynamic check implies the post-condition \ref{spec:handler_post_ctrs} (\autoref{sec:key-checks})
and since the intermediately-typed handler already satisfies post-condition \ref{spec:handler_post_acp},
it is easy
for the SMT solver
to finish the proof needed to type the handler with the strong type \ls{req_handler}.

%\ch{+ \fstar{} can prove (with some help?) that ...; {\bf in general focusing
%too much on the code and too little on the verification!}}
% \ch{+1; this seems
%   still true; for another example see the last phrase of this paragraph: this is
%   not just another boring implementation of higher-order contracts; these checks are
%   added because they help bridge the interface gap between verified and
%   unverified in a verified way; let's focus more on the novel, cool, exciting parts}
%this last dynamic check
%implies post-condition 4 and the intermediately-typed handler
%already satisfies post-condition 5,
% that the imported handler satisfies
% both of them, meaning
%it means it can be typed with the strong type \ls{req_handler}.
% This concludes the presentation of how the higher-order contracts add dynamic
% checks during compilation in this example. CH: space waster, no information conveyed

% \ch{Returning to my previous comment on ``let's focus more on the novel, cool,
%   exciting parts'', this subsection doesn't mention at all the coolest ideas
%   from \autoref{sec:contracts}. If this whole section is to deserve the name of
%   ``Key Ideas'', those key ideas should at least be briefly mentioned here, no?
%   This was about {\bf effectful checks}, which are a consequence of back-translating
%   effect-polymorphic target to source contexts. This adds a dependency on 2.3.}

\sciostar{} uses higher-order contracts not only during compilation but also during back-translation,
the dual of compilation. For the secure compilation proofs in \autoref{sec:rrhp}, we have to define
back-translation from a target context with a weak interface to a source context that is still
flag polymorphic, which involved designing our
higher-order contracts mechanism so that it can handle flag-based effect polymorphism.
For that we had to abstract away the enforcement of the checks from the higher-order contracts
and introduced a new notion of {\em effectful checks} (see \autoref{sec:exportable-importable}).

\subsection{Enforcing the access control policy by reference monitoring}
\label{sec:key-monitor}

In our example,
the assumption that we have to enforce using reference monitoring is
post-condition \ref{spec:handler_post_acp} of the handler, stating that it can only open, read, and close its
own files, and that only the web server is allowed to write during the execution (when
the handler calls \ls{send}).
Our traces are informative so that we can distinguish between the events produced
by the partial program (\EG web server) and those produced by the context (\EG handler).
This is because every event of the trace contains a bit of information of type \ls{caller}
that can have the value \ls{Prog} or \ls{Ctx}.
%a bit of information that exists on every event.\ch{This part about the operations might raise questions,
%  and I don't think it's crucial to mention it here (as opposed to Section 3.2).
%  Also the causality implied by ``because'' doesn't seem to be going in the right direction:
%  we added the extra argument in the operations because we wanted it on the trace.}
This bit % of information is used when writing specifications and it
allows us
to enforce a stronger specification on the untrusted context than on the partial program---\EG the
handler {\em cannot} directly write to any file descriptor, but the web server can
write to the clients.
% \ca{+to the clients / 1 more line}.\ch{server can write
%   to {\em what} file descriptors? to {\em every} file descriptor? to some (unspecified) one? (unclear)}

We use the fact that the monitor enforces an access control policy on the
context to give it a specification. However, using the policy directly to give
specification would say only half of the story because the policy characterizes
only what the context can do---\IE it would not
specify what the partial program does in case the context calls (back) the partial program. Because of this,
the \sciostar{} framework requires both an access control policy, denoted by $\Pi$ of type \ls{policy}, and
its specification, denoted by $\Sigma$ of type \ls{policy_spec}.

\begin{lstlisting}
type policy_spec = h:trace -> caller -> op:io_ops -> arg:mio_sig.args op -> Type0
type policy ($\Sigma$:policy_spec) = h:trace -> op:io_ops -> arg:mio_sig.args op -> r:bool{r ==> $\Sigma$ h Ctx op arg}
\end{lstlisting}
The access control policy is enforced {\em only} on the IO calls of the context,
% (thus it treats all calls as being the calls of the context), \tw{Maybe this 
% parenthesis is more confusing than anything.} CH:+1
while its specification characterizes the IO calls of the context {\em and} the IO calls
of the partial program during the execution of the context,
which is reflected by the extra \ls{caller} argument above.
The refinement on the returned value of the
policy makes sure that the policy implies the specification for the events of the context.

Here we give an example of an access control policy $\Pi$
and its specification $\Sigma$ that enforces the post-condition of the handler
(\ref{spec:handler_post_acp}).
The intuition is that the policy $\Pi$ allows the
handler only to open, read, and close its own files, while the specification $\Sigma$
says that the trace produced by the handler contains the events allowed by the policy plus
the writes done by the \ls{send} callback.

\vspace{-1em}\noindent
\begin{minipage}[t]{.5\textwidth}
\raggedright
\begin{mdframed}[backgroundcolor=black!5,hidealllines=true]
\begin{lstlisting}[frame=none]
val $\Sigma$ : policy_spec
let $\Sigma$ h caller op arg : Type0 =
match caller, op, arg with
| Ctx, Openfile, fnm -> fnm `in_folder` "/temp"
| Ctx, Read, fd -> is_opened_by_Ctx fd h
| Ctx, Close, fd -> is_opened_by_Ctx fd h
| Prog, Write, fd -> True
| _, _, _ -> False
\end{lstlisting}
\end{mdframed}
\end{minipage}% <---------------- Note the use of "%"
\begin{minipage}[t]{.5\textwidth}
\raggedleft
\begin{mdframed}[backgroundcolor=black!5,hidealllines=true]
\begin{lstlisting}[frame=none]
val $\Pi$ : policy $\Sigma$
let $\Pi$ h op arg : bool =
match op, arg with
| Openfile, fnm -> fnm `in_folder` "/temp"
| Read, fd -> is_opened_by_Ctx fd h
| Close, fd -> is_opened_by_Ctx fd h
| _, _ -> false
EOF
\end{lstlisting}
\end{mdframed}
\end{minipage}

% \ch{Starting around here the text becomes full of technical details, which are
% hard to follow. Moreover, it's unclear what the presented key ideas are; or
% what's cool, interesting, novel about what we're showing here. Why would one pay
% attention to these details in the first place?}\ca{A quick fix was to switch the
% order of the following two paragraphs.}%

We can encode such a policy as a post-condition by stating that
each event that happened during the execution satisfies $\Sigma$.
For that we use the following function \ls{enforced_locally} to
encode it as 
\ls{MIO a fl (fun _ -> True) (fun h lt -> enforced_locally}~$\Sigma$~\ls{h lt)},
which we shorten to \ls{MIO a fl True}~$\Sigma$ (from \autoref{sec:framework-overview}).
%\ch{Should remind the reader where
%  this was already used (otherwise bad flow):
%  for instance to define the intermediate interface for instance
%  many sections ago (2.2); also in 2.3.}
\begin{lstlisting}
let rec enforced_locally ($\Sigma$ : policy_spec) (h lt: trace) : Type0 =
  match lt with | [] -> $\top$ | e  ::  t -> let (| caller, op, arg, _ |) = destruct_event e in
                               $\Sigma$ h caller op arg /\ enforced_locally $\Sigma$ (e::h) t
\end{lstlisting}

The following three examples of adversarial request handlers all try to
violate the access control policy: 
\ls{handler3} tries to open a file outside of the \ls{/temp} folder, which is
the only folder authorized by the contract;
\ls{handler4} tries to write directly to the client, bypassing the
\ls{send}; \ls{handler5} tries to use IO operations outside of those
authorized by opening a socket.
All of them are prevented from executing by the reference monitor.
%
% \ch{Wondering: what was the point in
%   introducing that so early, if it's only used here? Couldn't the whole discussion
%   about the modeling of the target context from the end of 2.2 and whole of 2.3
%   come around here?}
\begin{lstlisting}
let adversarial_handler3 sec_io client req send = sec_io Openfile ("/etc/passwd",[O_RDWR],0x650)
let adversarial_handler4 sec_io client req send = sec_io Write (client,(Bytes.of_string "hello"))
let adversarial_handler5 sec_io client req send = sec_io Socket ()
\end{lstlisting}
\newtext{These examples are parametric in the IO library
(\ls{sec_io}) % \ch{of type \ls{io_lib ... }$\Sigma$ I guess?}
as described in \autoref{sec:key-context}, which
was hidden for simplicity in the previous examples.
The type of the IO library (\ls{io_lib}) is
defined using a single dependent function that takes as arguments
the operation (\EG \ls{Read}) and the arguments (\EG file descriptor) and
returns the result of the IO operation (\EG buffer).
The post-condition of the IO library guarantees what we mentioned earlier, that the
reference monitor enforces the access control policy and if an IO operation
is prevented from executing the trace does not change.}
\begin{lstlisting}
type io_lib (fl:erased tflag) ($\Sigma$:policy_spec) = (op : io_ops) -> (arg : mio_sig.args op) ->
  MIO (mio_sig.res op arg) fl $\top$ (ensures (fun h r lt -> enforced_locally $\Sigma$ h lt /\
                                               (match r with
                                                | Inr Contract_failure -> lt == []
                                                | r' -> lt == [convert_call_to_event Ctx op arg r'])))
\end{lstlisting}
\newtext{To generate the secure IO library,
we wrap the default IO library (denoted by \ls{call_io}) into
new functions that when invoked first enforce the access control policy $\Pi$,
by using the \ls{enforce_policy} function below.
We enforce the access control policy on each operation
by retrieving the monitor state using the \ls{get_mstate} operation,
and checking if the operation is allowed or not.
For the secure IO library, the caller is always set to be the context (denoted by \ls{Ctx}).}
\begin{lstlisting}
val enforce_policy : #$\Sigma$:policy_spec -> $\Pi$:policy $\Sigma$ -> io_lib $\Sigma$ AllOps
let enforce_policy $\Pi$ op arg = if $\Pi$ (get_mstate ()) op arg then call_io Ctx op arg else (Inr Contract_failure)
\end{lstlisting}

%\ch{The flow is not great here. Can we add an explicit connection between
%  \ls{io_lib} above and \ls{sec_io} below? I guess one is the type of the other?}

% \ifdraft
% \ch{I think verification should not just be a strap-on at the end of this
%   section, but integrated with the story. We already had a good ending to 2.3 for
%   instance about both soundness and RrHP, so I'm not sure whether having something
%   similar here is helpful. ifdrafted, since it looks odd with this title,
%   and then just soundness but no RrHP. Could also drop/fix the title and
%   better integrate with the story?}\ca{I agree that this strap-on is not great. What I
%   intended was a conclusion to this entire section.}%
% \paragraph{The formal global guarantees.} Static verification gives us the strong
% guarantee that the web server satisfies its post-condition---\IE that it responds to
% every request.
% %Since the web server has the initial control, its post-condition naturally
% %becomes the post-condition of the whole program.
% In \autoref{sec:soundness} we prove a soundness theorem that guarantees
% that the post-condition
% also holds after compilation because of the dynamic checks added by \sciostar{}.
% \ca{what can we say about RrHP?}
% \fi

%\ca{ends abrubtly, some kind of conclusion is needed}

\section{The \ls{MIO} monadic effect}
\label{sec:mio}

In this section we explain how we represent IO programs and their specifications
using a new monadic effect that we call \ls{MIO}
(for \emph{Monitored IO}).
%% in \fstar{}
% \ca{lets not talk about \fstar{} yet since the representation
%   can be used in other proof-assistants. maybe add a sentence later where we promise
%   to explain how this monadic effect is used in \fstar{} to represent computations}\er{commented \fstar{}}
We begin by introducing a general construction for obtaining monadic
effects % (\autoref{sec:mio-essence})
and then explain how to set up
each of the ingredients required for this construction:
computations% (\autoref{sec:computational-monad})
, specifications% (\autoref{sec:specification-monad})
, and how these last two relate %(\autoref{sec:monad-morphism})
(\autoref{sec:mio-essence}). We finally show how IO operations are
defined in \ls{MIO} (\autoref{sec:mio_ops}).
%  and give an example
% computation illustrating how to turn a pre-condition into a dynamic
% check in this effect (\autoref{sec:mio-examples}).

%% \ca{explain how the monadic effect can be defined in other languages too, and maybe
%%   say that \fstar{} has an advantage because it has native support for
%%   such monadic effects which helps with the SMT automation (see also
%%   TODO 02)}\er{Perhaps it's better to say this in the next subsection,
%%   saying that Dijkstra monads were introduced in Coq before as well?}\ca{Yes, sounds good!}\er{added it at the end, not convinced, see the comment}

\subsection{The Dijkstra monad behind the \ls{MIO} monadic effect}
\label{sec:mio-essence}
% was: the essence of ..., but that was before merging in subsections

%% \ch{Random details (move to 3.1?):}\er{moved, some changes in first phrase}%
We seek to write IO computations as terms of a type \ls{MIO alpha fl
pre post}. The arguments to \ls{MIO}
can be summarized as follows: $\alpha$ is the return type of
the computation, \ls{fl} constrains which IO operations
the computation can use (\autoref{sec:key-context}),
\ls{pre} is a pre-condition over the
trace
%% \ch{flow: new term, so far we called it trace}\er{added a parenthesis saying we refer to it as history, as we use it later}
of past IO events (which we sometimes refer to as \emph{history})
that must be satisfied to be able to call the
computation, and \ls{post} is a post-condition over the result and the
trace produced by the current computation.

To define our \ls{MIO} effect, we use a monad-like structure indexed
by specifications known as a Dijkstra monad~\cite{dm4free,dm4all}. For a
Dijkstra monad $\mathcal{D}$, a type $\mathcal{D}~\alpha~w$
classifies effectful computations returning values in $\alpha$ and
specified by $w : W~A$, where $W$ is called
a \emph{specification monad}.
%% a combination of Dijkstra
%% monads~\cite{dm4all} and layered effects~\cite{indexedeffects}.
% \ch{a regular
%   reader won't know what Dijkstra monads and layered effects are though;
%   added citations but that's not enough; at least Dijkstra monad needs
%   to be introduced informally (see DM4All or PDM4All how to do that;
%   and/or build on the ``Random details'' text from 3.0 above);
%   don't think we really need to introduce the term ``layered effects'' here}%
% \er{used some high-level description in DM4All, removed layered effects here}
%
%% A Dijkstra monad allows to index monadic computations by using values
%% in a monadic specification structure, relating both computations and
%% specifications via a translation.
%
We follow \citet{dm4all}, and construct a Dijkstra monad for \ls{MIO}
by picking three ingredients: i) a computational monad, ii) a
specification monad, and iii) a monad morphism from the former to the
latter; obtaining thus an effect in which computations are indexed by
specifications.
%% \tw{Is it not `computation monad'? Seeing as we use
%% `Specification monad'.}\er{I think that ``computational'' is how it's
%% usually called, e.g. in the abstract of DM4A}
%% \citet{dm4all} have shown that \emph{``any monad
%% morphism between a computational monad and a specification monad gives
%% rise to a Dijkstra monad''}. Following this, we pick three ingredients
%% for this construction: i) a computational monad, ii) a specification
%% monad, and iii) a monad morphism from the former to the latter; and we
%% obtain an effect in which computations are indexed by
%% specifications.\tw{Is it not `computation monad'? Seeing as we use
%% `Specification monad'.}\er{I think that ``computational'' is how it's
%% usually called, e.g. in the abstract of DM4A}
%
This effect captures the return type ($\alpha$) and specification
(\ls{pre} and \ls{post}). The flag index (\ls{fl})
% in \ls{MIO alpha fl pre post}, which is the flag
described in \autoref{sec:key-context}
and \autoref{sec:key-monitor} for controlling
what operations the computation can access, is added by refining the
effect with an extra index following \citet{indexedeffects}.
%
%% We refine this Dijkstra monad by defining a layered effect on top of
%% it. Specifically, we incorporate a flag that restricts the
%% operations that the computation can invoke.
% \ch{this flag was already introduced
%   in the intro and 2.2, so build on that, don't explain as if this is
%   a new thing (bad flow)}\er{Gave a explicit ref. to sec. 2.3,
%   referring to layered effects less explicitly}
%
In the following paragraphs, we provide a concise summary of the
ingredients above, and elaborate them in more detail in the subsequent
subsections.

% \ch{Neither here nor in 3.2 do we explain that a free monad is the most natural
%   choice for representing IO computations. It's also what's used in prior work
%   on IO verification, from FreeSpec~\cite{DBLP:conf/cpp/LetanR20} to DM4All. IO
%   is so connected in my mind to the free monad that I'm not even sure what other
%   representation one could use for IO. Instead of first explaining that our
%   focus is on {\em secondary} advantages, like the free monad being general and
%   extensible to {\em other} effects and whatnot, but those are {\em not} the
%   main reason we use a free monad for doing IO!}\er{did a little rephrase}%
%
For the computational monad, we use a free monad that is
parametric in the underlying primitive operations.
Free monads are particularly advantageous when it comes to
representing IO~\cite{LetanR20,dm4all,DBLP:conf/haskell/SwierstraA07},
%% \ch{Any other citations we can give on
%%   using free monads to represent IO?}\er{Added Beauty in the beast. Moggi already has ``interactive input'' and ``interactive output'' described as free monads in his first papers, but he doesn't make any remark saying that these are ``free monads''}
as its computations have an inherent tree structure.
In these trees, each node corresponds to a call to an
operation, encapsulating both the operation itself and its arguments, which act as \emph{output} of the program.
%% \ch{To me this is the O in IO}
Furthermore, each operation node has a corresponding child
node for each result, which act as \emph{input} for the program.
%% \ch{To me this is the I in IO, so I find using the word output here confusing!
%%   Should look at this from the program's POV, not the operation's POV.}\er{changed names and made explicit how they act}
%
For simplicity, we define operations for file management and for
socket communication, but these account only for one possible
instantiation of the monad. Our approach can be extended to other IO
operations as needed.\iflater\ch{could add that this could also be instrumentation
  (output) events, but first we would need an example showing that?}\fi{}
Moreover, the free monad could be used to implement other effects such
as exceptions, state, and non-determinism by extending the signature of effects
\cite{BauerP15}, and %% it is future work to integrate
we plan to integrate some of these effects in the future.
% \tw{`we plan to integrate
% ... in the future'?}\er{changed}
%% some of these effects in our account (\autoref{sec:conclusion}).
%% \tw{In our account? To me this sounds like just the text but I may be wrong.}
% \ch{Non-determinism is not on our future work list
%   though. Maybe if we wanted to support Concurrency / Multicore OCaml? Anyway, that
%   seems complicated and I have no immediate plans :)  Added ``some of'' before
%   ``these effects'' and changed from ``can'' to ``could''.}

The \ls{MIO} monadic effect uses a specification monad that captures the
behavior of a computation as a predicate transformer, \IE a function
that given a post-condition, returns a pre-condition that is strong
enough to guarantee the post-condition after execution of the
computation. As explained in \autoref{sec:intro}, in our setting a
pre-condition is a property of the current trace, and a post-condition
is a property of the result and the new trace.
%\ifsooner\tw{Didn't we say we were only
%using trace for the whole thing? What do you think Catalin?}\ch{When defining
% trace properties and program behaviors yes, but I think we are currently more
% lenient when we talk about specs, including in the intro.}\fi{}
When we write
specifications, we work directly with the pre- and post-condition,
and not with the predicate transformer. This is justified by the
existence of a translation from pre- and post-condition pairs to predicate
transformers~\cite{dm4all}.%% \ch{Don't think that equivalent is the best word
  %% here. For a start, AFAIK \citet{dm4all} show formally that wps are more
  %% expressive. Then in practice wps are better for inference, and pre- and post-
  %% conditions are better for humans.  So it's hard to claim that they are
  %% equivalent.}\er{looking at \citet{dm4all}, we have a Galois correspondence (``weak equivalence''), not sure
  %% how this translates to expressiveness, but removed the ``equivalent'' part}

%% The \ls{MIO} effect will allow writing specifications about the
%% computations written in the free monad. To give specifications to
%% computations, we will use properties defined over finite traces that
%% are ``computed'' by observation\er{Not sure about this observation
%% thing}. Each time an IO operation is called, an event (given by an
%% operation name, arguments and result) is appended at the beginning of
%% the trace.
%% %
%% The \ls{MIO} effect will restrict to only those computations that
%% behave according to the properties passed in the type.
%% %% The properties restrict the behavior of the
%% %% program to only a subset of all possible traces.  The traces exist
%% %% only at the specification level and are not actually computed at
%% %% runtime.
%% \tw{The phrase `the properties' looks weird to me.}\er{Rewritten something, but not sure if this is what the phrase wanted to convey. CA?}

Finally, we need to define a monad morphism from the computational monad
to the specification monad, so from a computation in the free monad to a predicate transformer.
% semantics
% written in the specification monad.
% \ch{I don't actually know what recipe this refers
%   to.  Even after reading the rest of the section, I still have no clue what
%   recipe this text refers to.}\tw{+1 The way this is stated it sounds like
%   a secret recipe. The sentence should at least hint at where this recipe is 
%   explained or what it is.}\er{removed the recipe thing}
%
% The steps outlined in this subsection
These steps allow us to build a Dijkstra monad for
statically verifying IO programs with respect to specifications, but we also
need a mechanism to support dynamic checks.
For this we include a new silent operation
called \ls{get_mstate} that returns % is meant to return
the reference monitor state, % at the point the computation is called,
which allows us to write dynamic checks.
Dijkstra monads can be implemented in other proof oriented programming
languages, not only in \fstar{}. For instance, \citet{dm4all} have
also implemented Dijkstra monads in Coq. Nevertheless, \fstar{} stands
out due to its built-in support for SMT automation and the ease of
defining new effects, making it a particularly convenient language
for working with Dijkstra monads.
% \er{not convinced on where to put this} -- seems fine to me

\paragraph{The computational monad.}
\label{sec:computational-monad}
For the \ls{MIO} effect, we chose a computational monad that
is parametric in the underlying primitive operations by using a free
monad \citep{BauerP15} that can accept any signature
of effects \citep{LetanRCH21}.
%
%% The operations are represented using an
%% interface \cite{BauerP15} that models the
%% available effects \cite{LetanRCH21} together with
%% their input and output types.
%
A signature is described by a type of operations together with
the arguments' %(or input)
and result type %(or output) -- input and output mean something else in a paper about IO
for each operation:
\begin{lstlisting}
type ops_sig (ops:Type) = { args : ops -> Type ; res : (op:ops) -> args op -> Type }
\end{lstlisting}

%% Free monads can be used to implement other effects such as exceptions,
%% state and non-determinism by extending the
%% interface \citep{BauerP15} and we plan to integrate
%% them in future work.\er{Repeated?}

We use the usual representation of the free monad, as trees
whose nodes are operation calls, and whose leaves are the values returned by a
computation.\ifanon\else\footnote{A special constructor of our free monad that is needed in
  \fstar{}~\cite{pdm4all} is omitted from our presentation.}\fi{}
%
%For our computations, we additionally allow
%an extra variant of nodes which contain a pre-condition. These (pure)
%pre-conditions %% (with type \ls{pure_pre} of pure pre-conditions)
%are
%used to capture extra logical constraints we enforce during
%computation\er{Is this better?}\ca{No. What does ``enforce during computation'' mean?
%  Also, we use ``enforce'' in the paper in a different way, thus I would avoid using
%  this word, especially that \ls{Require} is not relevant to the presented work. We should
%  simply say that it is necessary because of how \fstar{} effects work and it is not relevant.
%  Or even better, we should say that we omit a constructor that it is not relevant, the way we
%  did this in the HOPE submission}:
%
%
Our \ls{Call} constructor %% also
comes with an additional
parameter of the variant type \ls{caller} (\autoref{sec:key-monitor}) that we use to mark
whether the computation making the call is the partial program (\ls{Prog}) or the
context (\ls{Ctx}).\ch{This seems the right place to
% (1) remind the reader why we added this caller argument (it was already mentioned
%   in the contributions and section 2.2 that we want to make this distinction in the specs);
(2) explain why allowing the code to pass in this flag is not a problem; as far
  as I remember the reasoning goes like this:
  the context can't directly call these operations anyway,
  our verified operation wrappers set it appropriately to context,
  and the program is anyway statically verified and has no reason to want
  stricter rules imposed on it by the specs.}

\begin{lstlisting}
type free (ops:Type) (s:ops_sig ops) (a:Type) : Type =
| Call : caller -> (op:ops) -> (arg:s.args op) -> cont:(s.res op arg -> free ops s a) -> free ops s a
| Return : a -> free ops s a
\end{lstlisting}
%| Require : (pre:pure_pre) -> cont:((squash pre) -> free ops s a) -> free ops s a
%
The type constructor \ls{free} forms a monad independently of the
operations and signature, with \ls{free_return} and
\ls{free_bind} combinators we define in the usual way~\cite{BauerP15}.

For \ls{MIO}, the signature we are interested in is given by a type of operations:
%\ca{in the first paragraph of the section
%  we say that we
%  also support network communication, but here we only have
%  opening files. later, in the case study, we claim we extend this list.
%  the claim here has to be more nuanced.}\er{Wrote ... and something below ``The three dots...''}
\begin{lstlisting}
type mio_ops = | Openfile | Read | Write | Close | Socket | Setsockopt | Bind | SetNonblock | ... | GetMState
\end{lstlisting}
% \tw{I realise now that this is not valid F*, sorry. Do you prefer we add the
% first pipe even though it reads weird?}\er{Fixed. Also changed the \ls{Write} order.}%
The first operations
%% \ch{plus the ... I would expect, are those not \ls{io_ops} too?}\er{solved}
are natural IO operations
%% (opening, reading, writing, and closing files)
(captured by a sub-type \ls{io_ops}), while the last one is
the operation that allows performing dynamic verification by
retrieving the current trace.
The three dots indicate the existence of
more IO operations, such as those related to network communication
that we use in \autoref{sec:key-ideas}, but which we do not write here
for the sake of concision.%% \ch{There was more space on the line though, so to me
  %% this sounds like an unconvincing pretext :) Can we please add more operation
  %% names until we fill the line?  It doesn't cost us anything, while showing more realism.}\er{added other ops}
% This also overlaps with text from
%  3.1 (second paragraph) that you probably want to merge here (bad flow currently).}
%
%% We define a sub-type \ls{io_cmds} for capturing
%% only the operations related to IO.
%% , and another sub-type \ls{m_cmds}
%% that characterize \ls{GetMState}.
% \tw{How useful is that last one type with only one inhabitant?}\er{for now it's only useful for the signature construction below... but we could easily remove it}\tw{Ok, if it makes things easier later then we can keep it. I'll let you judge.}\er{removed it in the end}

The signature we give % to the IO operations
enables all IO operations to return errors,
as their result type % of these operations
is defined to be either a proper result or an error.
%\ca{I would
%avoid using throw since it is specific to exceptions}\tw{You could use
%`raise', or even just `return'.}\er{changed to return} %
For example, the \ls{Read} operation has as argument
type \ls{file_descr}, and as return type \ls{either buffer err}.
\iffull
The left case of \ls{either}
describes a successful read from the file descriptor, while the right
case uses the type \ls{err} % (of exceptions)
to signal an error.
\fi
The \ls{GetMState} operation has argument type \ls{unit} and result
type the type of the monitor state. For simplicity, we
assume here \ls{GetMState} to return the current trace
until \autoref{sec:abstract-state}, where we update it to return an
abstract state.
%% \ch{GetMState is now GetMState; still here is the first
%% time in this section where we would need to disclose that we will be
%% assuming a simplification until \autoref{sec:abstract-state}: we
%% assume that GetMState returns the current trace, instead of some
%% abstract state.}\er{added something}
%%  that
%% was thrown
%
%The signature for %% the 
%\ls{mio_ops} %% operations
%is called \ls{mio_sig}, and
%it is defined as the sum of the signatures \ls{m_sig} (of \ls{GetMState}) and \ls{io_sig} (of commands
%in \ls{io_ops}).\ca{this sentece is hard  to parse since it contains
%  a lot of references to code. is this sentence useful?}
%% \begin{lstlisting}
%% let mio_sig : op_sig cmds = add_sig m_sig io_sig
%% \end{lstlisting}
%% \ca{use \ls{mio} instead of \ls{iio} please, and maybe rename \ls{inst_sig} in \ls{m_sig}?}\er{used \ls{m_sig}, but not convinced still, maybe \ls{mon_sig}? anyway, it looks better that \ls{mio_sig} is the sum of \ls{m_sig} and \ls{io_sig}}
%
% \ch{Don't think the rest is really needed here. Is it used elsewhere? -- it is, sadly}%
The %% underlying 
computational monad % for the Dijkstra monad
is obtained %% simply
by
instantiating \ls{free} with the type of the operations (\ls{mio_ops})
and their signature (\ls{mio_sig}):
\begin{lstlisting}
type m a = free mio_ops mio_sig a
\end{lstlisting}

\paragraph{The specification monad.}
\label{sec:specification-monad}
The role of the specification monad is to give a logical account of
the semantics of the computation. We use a predicate transformer
semantics for the specification of computations, in which predicates
are properties written over a trace of past IO events, with traces
represented simply as lists of events.
% \ch{unclear why exactly; bad flow: didn't explain
%   here that the predicate transformer semantics has anything to do with the
%   traces (and can't assume the reader remembers it from 1 page ago in 3.1,
%   which seems like a good reason to repeat/move that here).}\er{rewritten some of it}
We chose a specification monad that
allows writing a pre-condition over the entire history of events,
while the post-condition is written over the result of the computation
and the events produced by it (we refer to it as local trace).
% \tw{We've already
%   described what the pre- and post- are about before.}\ch{The name
%   local trace was not introduced before, and it's used below, so let's keep at least that?
%   Also one needs to be a bit careful about one thing: the post-conditions
%   written by users also take the initial trace, while those used for predicates
%   transformers don't. Not sure we want to bring this up explicitly, but let's
%   be careful about it ourselves.}
%
These events are defined following the signature of the operations. For each
IO operation, we have an event constructor capturing the
operation name, caller, argument and result of the operation call:
% \ca{We represent events using an inductive type because it works better with the SMT automation.
%   We also tried to represent the events as dependent triples, but then one needed 
%   to apply tactics manually to type any IO code.}
% \ch{The contributions mention that this effect was ``engineered to take
%   advantage of SMT automation'', but then there's nothing explaining that in this
%   section, beyond the very last sentence about keeping the history backwards.
%   So yes, I think it would be good to add this kind of insights.}
% \er{added a phrase about SMT automation after the type}
\begin{lstlisting}
type event = | EOpenfile : (c:caller) -> a:mio_sig.args Openfile -> (r:mio_sig.res Openfile a) -> event | ...
\end{lstlisting}
%% \begin{lstlisting}
%% type event =
%%   | EOpenfile : (c:caller) -> a:io_sig.args Openfile -> (r:io_sig.res Openfile a) -> event
%%   | ERead     : (c:caller) -> a:io_sig.args Read     -> (r:io_sig.res Read a)     -> event
%%   | EWrite    : (c:caller) -> a:io_sig.args Write    -> (r:io_sig.res Write a)    -> event
%%   | EClose    : (c:caller) -> a:io_sig.args Close    -> (r:io_sig.res Close a)    -> event
%% \end{lstlisting}\er{TODO:explain \ls{caller}}\er{It's explained in free now}
%
For better SMT automation, we have opted to define events as a
variant type, with each operation having its own constructor.
This approach proved to be more effective in \fstar{}
compared to using dependent triples based on \ls{io_ops}.
Finally, \ls{GetMState} does not have an event associated.
%%  As% explained above, a trace is simply a list of events: \ls{type trace = list
%% event}.
%% \begin{lstlisting}
%% type trace = list event
%% \end{lstlisting}

We define a type constructor \ls{w} that captures transformers from
post-conditions to pre-conditions:
\begin{lstlisting}
type trace = list event
type w a = w_post:(lt:trace -> r:a -> Type0) -> w_pre:(h:trace -> Type0)
\end{lstlisting}
In this representation of specification, we interpret the history of
events in pre-conditions as given in reverse chronological order.
This is another factor contributing to successful SMT automation:
\newtext{Consider for instance the \ls{is_open} predicate
that checks whether some file descriptor \ls{fd} is open according to
the history. \ls{is_open} proceeds by looking at
events in the trace to see whether there is some \ls{EOpenfile _ (Inl
fd)} event that was not followed by an \ls{EClose fd _} event.
When such a predicate is checked, most often the full history is not known,
rather it is some abstract value \ls{h}. Therefore, when one
adds more events, the abstract history chunk is found at the end,
\ls{last_event :: h}, which is good because the SMT can reduce
\ls{is_open fd (last_event :: h)} to \ls{true} if \ls{last_event} is \ls{EOpenfile _ (Inl fd)}.}

%% let wp_pre = h:trace -> Type0
%% \citet{dm4free}\tw{Is it DM4All?}\er{used dm4free, CH or GM can confirm?}
%% \gm{the fact that WPs are continuation-like is much older than any of these,
%% is that what this is trying to say? (also added ``predicate'' to ``transformers'',
%% I was thinking monad transformers on a first read.}\er{thanks, removing the misattribution}
%% observed that
%
Predicate transformers can be naturally organized as continuation-like
monads~\cite{Jensen78}. Indeed, as in the case of \ls{free}, the type \ls{w} comes
equipped with a monadic structure given by combinators \ls{w_return} and \ls{w_bind}.
%% \begin{lstlisting}
%% let wp_return (x:'a) : wp 'a = fun p _ -> p [] x
%% let wp_bind (#a #b:Type) (w : wp a) (kw : a -> wp b) : wp b = ...
%% \end{lstlisting}
%\tw{Maybe we shouldn't call it \ls{hist} in the paper? It's not very explicit
%what it is and it could easily be confused with the history. How about W or wp?}
%\ca{+1}\er{Changed it to `wp`}
%
%% In addition to the monadic structure,
In addition, a specification monad also needs to come equipped with an
order between computations of the same type. For \ls{w}, we define
this order as follows:
%% val wp_ord (#a : Type) : wp a -> wp a -> Type0
\begin{lstlisting}[framesep=1pt]
let ([=) wp1 wp2 = forall h post. wp2 post h ==> wp1 post h
\end{lstlisting}
This order is a form of refinement~\cite{SwierstraB19} that compares
specifications by precision. % Monadic
Combinators \ls{w_return}
and \ls{w_bind} are monotone with respect to this order.
To form an ordered monad % it is also required that
we restrict to only
those predicate transformers that are monotonic---\IE which map
stronger post-conditions to stronger pre-conditions---which we enforce
by refining type \ls{w}.
%% \ch{This part is already about the monad morphism. I think it should be moved to
%%   the next sub-section. (currently flow is also bad)}\er{moved it below}%

\paragraph{The monad morphism and the Dijkstra monad.}
\label{sec:monad-morphism}
As the third ingredient for constructing a Dijkstra monad, we need a
monad morphism from the computational monad \ls{m} to the specification
monad \ls{w}. We call such a morphism \ls{theta}
% , as it is customary in literature on Dijkstra monads,
and it has type \ls{(#a:Type) -> m a -> w a}.
It is defined by recursion
% \tw{I don't understand}\er{changed to ``by
% recursion''}\tw{Still, the definition doesn't follow but is made by
% recursion? Or `It is defined by recursion'?}\er{not sure I get the
% difference here, but changed it to the latter, is it better?} 
% TW: for me yes :) Thanks!
on the monadic computation, and uses the following function
\begin{lstlisting}
let mio_wps (c:caller) (op:mio_ops) (arg:mio_sig.args op) : w (mio_sig.res op arg) = fun post h ->
  if GetMState? op then post [] h else (forall (r:mio_sig.res op arg). post [convert_call_to_event c op arg r] r)
\end{lstlisting}
that maps an operation and its arguments to a specification.
For \ls{GetMState}, this function captures the history, and
passes it as the result to the post-condition (in \autoref{sec:abstract-state} we
change this definition to work with an abstract state).
For IO operations, the
post-condition receives a single-event trace corresponding to the IO
operation performed, which is converted 
% with help of 
by the auxiliary
function \ls{convert_call_to_event}.

We proved that \ls{theta} is indeed a monad morphism~\cite{dm4all,relational700},
which means that it preserves
%% weakly\ch{preserves weakly??? no reader will understand this}
the \ls{return} and \ls{bind} of \ls{mio} and \ls{w}.
% \ch{Is it really only a {\em lax} morphism? Added the
% correct citation for that above~\cite{relational700}, but still
% surprised this would be only lax.}\er{I am not sure where I got that
% {\em lax} part... removed it}
%% Using it, %% together with \ls{mio_wps}, we can define 
The representation of our Dijkstra monad is then:
\begin{lstlisting}
type dm (a:Type) (wp:w a) = c:(m a){theta c $\mathrel{\sqsubseteq}$ wp}
\end{lstlisting}
The idea is that a computation is indexed by \ls{wp} if its image
by \ls{theta} is smaller than \ls{wp},
%% (with respect to the order in
%% the specification monad)
which means that the specification computed
by \ls{theta} is refined by the one given in the index.
%\ca{maybe we
% should quickly brag that we prove this in F*? well, in F* we show
% that is a lax morphism.}\er{added a comment above}

%\ca{the start of the subsection starts with Finally, and this paragraphs starts with The final}\ca{Since it is quite important, maybe we can make this a subsection too?}\er{changed the Finally, but also made the thing a subsection}

\paragraph{The layered effect.}
The final step for defining \ls{MIO} is to add
the flag that restricts the IO operations, which 
is needed to % essential when we
write \emph{flag-polymorphic} context types (\autoref{sec:key-context}).
We use a predicate \ls{satisfies : m 'a -> tflag -> bool}
%% \begin{lstlisting}
%% val satisfies : mio 'a -> tflag -> bool
%% \end{lstlisting}
that decides if the underlying computation uses
only the operations allowed by the flag. The representation of the
\ls{MIO} monadic effect is then defined as follows:
\begin{lstlisting}
type mio (a:Type) (fl:erased tflag) (wp:w a) = c:(dm a wp){c `satisfies` fl}
\end{lstlisting}
%% \IE computations that satisfy the specification given by \ls{wp} and
%% are restricted to use the operations described by \ls{fl}.

Given all the ingredients above, we use one of the \fstar{}
extension mechanisms~\cite{indexedeffects} to define the \ls{MIO} effect:
%
%\mw{A reader that is unfamiliar with F*'s "layered effects" won't understand what's going on here (without reading the cited paper). Maybe add a small description what this instruction does, e.g.:
%> The `effect` instruction of F* registers a 'layered effect' [42]. This effect is based on `mio`. From now on, we can write `f : X -> MIO ...`, and F* automatically applis the right monadic constructs when programming.}
\begin{lstlisting}
effect MIO (a:Type) (fl:erased tflag) (pre : trace -> Type0) (post : trace -> a -> trace -> Type0) = ...
\end{lstlisting}
By defining \ls{MIO}, we can write \fstar{} code
in a direct and ML-like applicative syntax (\autoref{fig:running_example})
that automatically elaborates into the monadic combinators.
More importantly, \fstar{} uses the specification monad to compute verification
conditions and discharge them using the SMT solver.
% \ch{Getting ML syntax is neat but not the main reason we define an \fstar{}
%   effect. Getting SMT automation is the main reason, and it's not even mentioned here!}

%% Strictly speaking, the effect
%% defined uses predicate transformers as indices (inhabitants of the
%% specification monad \ls{wp}). However, going from a predicate
%% transformer indexed effect to an effect that is indexed using pre- and
%% post-conditions is standard practice in \fstar{}:
%% \begin{lstlisting}
%% total reflectable effect MIOwp (a:Type) (fl:erased tflag) (wp : wp a) with { repr = mio ; ... }
%% effect MIO (a:Type) (fl:erased tflag) (pre : trace -> Type0) (post : trace -> a -> trace -> Type0) =
%%   MIOwp a fl (to_wp pre post)
%% \end{lstlisting}
%% \tw{Do we need so much information? It's not like people will understand.}\er{:-( commented it}
%% where \ls{to_wp} converts a pair of a pre and post-condition into a
%% weakest pre-condition transformer:
%% \begin{lstlisting}
%% let to_wp #a pre post : wp a = fun p h -> pre h /\ (forall lt r. post h r lt ==> p lt r)
%% \end{lstlisting}

\subsection{The \ls{MIO} effect operations}
\label{sec:mio_ops}

We have introduced the \ls{MIO} effect on top a free monad that
supports constructors such as \ls{Openfile}, \ls{Read}, \ls{Write},
and \ls{Close}. However, \ls{MIO} computations do not need to know
about these underlying constructors. Instead, we provide a function
that wraps the IO operations and provide a meaningful specification
derived from \ls{mio_wps}:
\begin{lstlisting}
let call_io (c : caller) (op : io_ops) (arg : mio_sig.args op) :
  MIO (mio_sig.res op arg) IOOps (requires (fun h -> True))
       (ensures (fun h (r:mio_sig.res op arg) lt -> lt == [convert_call_to_event caller op arg r])) =
  MIO?.reflect (Call c op arg Return)
\end{lstlisting}
%\er{not sure what to do with \ls{io_pre}, which is defined to be True. CA?}\ca{we can make \ls{io_pre} have a spec only when it is called by the partial
%  program, that would be cool. this means passing caller to \ls{io_pre}}\ca{actually, we do not do this when verifying the web server, so maybe just replace it with true?}\er{replaced it with true}
%
%\ca{please unfold \ls{mio_call} inside this definition and rename \ls{static_cmd} into \ls{call_io}}\er{Done}
This specification asserts that the computation will have as a local
trace a single event which corresponds to the performed operation. The
flag used is \ls{IOOps}, which indicates that the computation
performs (only) IO operations.
For writing code, we define synonyms for \ls{call_io} for each IO operation,
such as \ls{openfile}, \ls{read} and \ls{close}.
Similarly, we also define \ls{get_mstate}, a wrapper for the operation \ls{GetMState},
that takes no arguments and returns a \ls{trace} (for now, we change this in \autoref{sec:abstract-state}).
% This operation allows us to do reflection on the past;
% \tw{Reflection on the past might a bit obscure.}
% CH: GetMState was already introduced by now
%
The flag for its effect
is \ls{GetMStateOps}%% , which is the flag that marks that the only
%% effect that can be performed is \ls{GetMState}
. The specification
ensures that the result of calling \ls{get_mstate} is the history (\ls{r == h}), and
moreover, that the local trace is empty (no IO events happened, \ls{lt == []}).
%[framesep=1pt]
\begin{lstlisting}
let get_mstate () : MIO trace GetMStateOps (requires (fun h -> True)) (ensures (fun h r lt -> r == h /\ lt == [])) =
  MIO?.reflect (Call Prog GetMState () Return)
\end{lstlisting}%\er{used MIO?.reflect... not a terrible lie}%

%% The specifications of \ls{read},
%% \ls{write} and \ls{close} require the file descriptor given as an argument to have
%% been open before, by using the operation \ls{openfile}, and not closed
%% in the meantime. In Figure~\ref{fig:read} we present how the specification
%% of the operation
%% \ls{read} looks like.
%% \begin{figure}[h]
%% \begin{lstlisting}
%% val read : (fd:file_descr) -> IO string (requires (fun (h:trace) ->
%%   is_open fd h)) (ensures (fun msg lt -> lt = [Read fd msg]))
%% \end{lstlisting}
%% \caption{The \ls{read} operation with specification}
%% \label{fig:read}
%% \end{figure}
%% \tw{I guess the figure is gone now.}\er{Added.}

%\subsection{Static verification example}

%\ca{Examples on how the trace looks like, pieces of code, spec of primitives,
%  how \ls{GetMState} is used, etc.}
%\ca{Maybe these can be integrated already in the previous subsections.}
%\er{Here goes a very silly example...}

\section{Higher-order contracts}
\label{sec:contracts}

We gave a brief presentation of how higher-order contracts
work in \autoref{sec:key-contracts}, and now we give a more detailed picture of how the \ls{import} and \ls{export}
functions are implemented. The two functions are designed to be used
during the compilation of the partial program and during the back-translation of the
context, which are explained later in \autoref{sec:secure-compilation}, so the design was
influenced by the need to work with flag-based effect polymorphism and
by the need to define back-translation.
As we said, \ls{import} and \ls{export} are two dual functions that
% have the role to
convert between strongly-typed values and intermediately-typed values.
Because the \ls{import} and \ls{export} functions are implemented in \fstar{}
and the conversion happens between \fstar{} types directly, the correctness of
the type conversion is verified by \fstar{} typing.
% \ch{For a start, what do you mean by ``correctness of implementation''.
%   To me it sounds a lot like compiler correctness,
%   which boils down to a cancellation law on import/export,
%   on which we gave up. So we anyway need a better name!} CH: changed to type conversion

%\gm{
%Because we can access the trace of IO events from the monitor state,
%inside the contracts we only have to enforce the necessary dynamic checks because
%keeping the IO trace is abstracted away.}

\subsection{Intermediate types}
\label{sec:iterm-types}

To begin, we define what we mean by strong types and intermediate types.
We define \emph{intermediate types} as a subset of \fstar types,
by using the
%(methodless\ch{is that really an English word?},trivial)
trivial \ls{interm} type class from \autoref{fig:interm}.
% \tw{I still argue this is not empty, but trivial, o singleton class.}
% 
The \ls{interm} % type class 
instances represent the types of the target language
after the reference monitor was added to gave them a specification by
enforcing the access control policy (\autoref{sec:key-monitor}).
%
%\ca{\bf why is an interm arrow having $\Sigma$ as post-condition? This was explained I think in 2.2.
%  \ch{For now only in some paren, but yes} We have to
%  explain this because it is quite weird in the entire section that we talk about weak types, but they
%  have a post-condition.}\ch{I agree that in this section we should make it
%  clear from the start that export doesn't go all the way to the target
%  language, but it meets in the middle with the reference monitoring, and what
%  you call weak types here are {\em not} the types of the target language, but
%  that middle point.}
%  \ca{now we can explain this much earlier since we do not use weak anymore}
% CA: this was explained better in 2.2
%
% As we are considering an ML-like language,\ch{I don't think that our
%     intermediate language is ML-like, is it? simply typed? (I guess weakly typed is already taken)}
We provide instances for every base
(unrefined) type, as well as pairs, sums, options, and non-dependent functions.
%
%The \ls{weak} type class is indexed by a weak type (denoted by \ls{wtyp}),
%
The \ls{interm} type class is indexed by an operations flag (\ls{fl},
introduced in \autoref{sec:key-context}) and
the specification of an access control policy ($\Sigma$).
%The two indices guarantee that its instances are invariant with
%respect to them.
%
These indices ensure that higher-order functions are
invariant over the flag \ls{fl} and the policy $\Sigma$. This is realized
by only defining instances that respect this property.
%\gm{Why are fl and pi indices of the \textbf{class}? it's just a
%  predicate about the type and a trusted one at that, can we remove those
%  indices?}
%\tw{I think Cezar was saying that this was to make sure that when importing
%  or exporting you would have the same values for the flag and pi everywhere.
%  But maybe the instances can do that?}
%
We need intermediately-typed higher-order functions to be invariant over the flag
and the policy so that they are aligned with the weakly-typed higher-order functions
which are also invariant over those.
%because of the weak interfaces.
%The need to index the \ls{interm} type class by the flag \ls{fl} and the
%specification $\Sigma$ came from our usage of flag-based effect polymorphism.
%For example, for a type \ls{t}, the constraint of type
%\ls{fl:_ ->}~$\Sigma$~\ls{:_ -> interm t fl}~$\Sigma$ guarantees that type $t$
%is effect polymorphic,\ch{AFAIR up to this point in the paper our only use of flag-based
%  effect polymorphism was for target {\em contexts}, so {\em weak types}. So
%  we should first explain why flag-based effect polymorphism is at all a thing for
%  {\em intermediate types}. (Here or in \autoref{sec:key-ideas}, depending on
%  whether we think it's a (part of a / needed to explain a) key idea or not)}
%and we use this kind of constraints later when we model our secure
%compiler.\ch{where exactly? give more precise reference;
%  also can we get rid of ``our'' and ``model''?}

We do not give a definition for {\em strong} types.
We consider strong types the types from which we can export from and
we can import into.
This definition is open because we implement \ls{import} and \ls{export}
using type classes.
%
% Therefore, strong types includes only the types
We define instances for
% and this includes the types defined as
intermediate types, base refined types, pairs, sums, options,
and
%\remove{ (non-dependent)}\ch{this parenthesis seems exaggerated; AFAIR there is some
%  amount of dependency we don't do(?), and if so that can be explained somewhere,
%  but that definitely doesn't mean our functions are ``non-dependent''}
functions with pre- and post-conditions.%\ca{I don't think
%  this enumeration is clear, plus it does not explain 
%  that we support pairs of strong types, or sums of strong types, etc.}
%\ch{Seems pretty clear to me (otherwise what's unclear/missing?),
%  and your ``plus'' seems pretty obvious.}

\begin{figure}
\begin{lstlisting}
class interm (ityp:Type) (fl:tflag) ($\Sigma$:policy_spec) = {}
instance interm_unit fl $\Sigma$ : interm unit fl $\Sigma$ = {}
instance interm_file_descr fl $\Sigma$ : interm file_descr fl $\Sigma$ = {}
...
instance interm_err fl $\Sigma$ : interm err fl $\Sigma$ = {}
instance interm_pair fl $\Sigma$ t_1 {| interm t_1 fl $\Sigma$ |} t_2 {| interm t_2 fl $\Sigma$ |} : interm (t_1 * t_2) fl $\Sigma$ = {}
instance interm_either fl $\Sigma$ t_1 {| interm t_1 fl $\Sigma$ |} t_2 {| interm t_2 fl $\Sigma$ |} : interm (either t_1 t_2) fl $\Sigma$ = {}
instance interm_arrow fl $\Sigma$ t_1 {| interm t_1 fl $\Sigma$ |} t_2 {| interm t_2 fl $\Sigma$ |} : interm (t_1 -> MIO t_2 fl $\top$ $\Sigma$) fl $\Sigma$ = {}
\end{lstlisting}
\caption{The % \ls{interm} -- funny duplication
  type class representing intermediate types.
  %\ca{can we say
  %something more interesting?}\ch{not sure, but at least I removed the funny duplication}
}
\label{fig:interm}
\end{figure}

\subsection{Exportable and importable type classes}
\label{sec:exportable-importable}
The \ls{import} and \ls{export} functions are fields
% methods\ch{Never heard this word used for type classes (only for OOP
%   classes). Do people really use it to mean type class fields?}
of two type classes named
\ls{importable_to} and \ls{exportable_from}.
The classes are indexed by the \emph{strong} type \ls{styp} and
they have as a field the \emph{intermediate} type \ls{ityp}.
% \ch{
%   %Looking at the definition, only one is
%   %an index and the other one is a field!\ca{this is what the sentence is saying. tried to fix} And
%   it's funny that it's always the
%   \ls{styp} that's an index, so would better names for the two classes be
%   \ls{exportable_from} and \ls{importable_to}?}\ca{I did play a lot with the name
%   of this index and I think \ls{styp} and \ls{ityp} are the clearest and
%   the easiest to remember later because they stand for ``strong type'' and
%   ``intermediate type''}\ch{Your answer seems off-topic.  What I'm proposing is
%   better names for the current type classes. Can we rename them?}
%
Function \ls{export} is defined as taking a value of the strong type \ls{styp} and
returning a value of the intermediate type \ls{ityp}, and
function \ls{import} is defined as taking a value of the intermediate type \ls{ityp} and
returning either a value of the strong type \ls{styp} or an error in case that the
dynamic check failed.
The two type classes are additionally indexed by the operations flag \ls{fl} and the specification $\Sigma$;
these extra indices are required by the \ls{interm_ityp} constraint on \ls{ityp} to be an intermediate type.

\begin{lstlisting}
class exportable_from (styp : Type) (fl : erased tflag) ($\Sigma$ : policy_spec) (cks : checks) = {
  ityp : Type; interm_ityp : interm ityp fl $\Sigma$; export : eff_checks fl cks -> styp -> ityp; }
\end{lstlisting}
\vspace{-1mm} % hacky way to add some space between these two fragments, but less than a new line in the listing
\begin{lstlisting}
class importable_to (styp : Type) (fl : erased tflag) ($\Sigma$ : policy_spec) (cks : checks) = {
  ityp : Type; interm_ityp : interm ityp fl $\Sigma$; import : ityp -> eff_checks fl cks -> either styp err; }
\end{lstlisting}

One special case is functions that are flag-based effect polymorphic.
We want the wrapping of an effect polymorphic function to result
also in an effect polymorphic function.
This is necessary later in \autoref{sec:rrhp} when we define back-translation
from a target context to source context that are both flag-based effect polymorphic.
Our contracts dynamically enforce pre- and post-conditions about IO behavior,
for which we have to use the effectful \ls{get_mstate} operation---however, we {\em cannot} use it
directly without losing the effect polymorphism.
Our solution is to abstract away how the enforcement of the checks is done
by parameterizing the \ls{import} and \ls{export} functions over, what we call,
{\em effectful checks}.
The effectful checks are indexed by the same operations flag \ls{fl}
as for the wrapped function, so they can be used inside the contract
without affecting the effect polymorphism.
Effectful checks are obtained by merging the \ls{get_mstate} operation
with the checks.
%and to require it to be under the same flag as the imported/exported function.
%Therefore, we 
%merge calling the \ls{get_mstate} operation
%with the checks into what we call effectful checks,
%and we parameterize the \ls{import} and \ls{export} functions over these
%effectful checks.
%
The compilation framework automatically converts the checks into effectful checks
and passes them when calling the \ls{import} or the \ls{export} functions.
This approach forces us to index the type classes with the collection of checks
that have to be enforced.

The \ls{cks} collection contains dynamic checks to enforce
the pre- and post-conditions of the strong type \ls{styp}.
However, it does not include dynamic checks to enforce refinement types
because we do not need to use the \ls{get_mstate} operation to enforce them.
To learn about how we convert the refinement types into dynamic checks, we refer
the reader to the work of \citet{TanterT15} because we use the mechanism they
proposed, and most refinement types could anyway also be easily converted to
pre- and the post-conditions.
% \ch{It's a bit sad that such simple cases were cut for space, since they could
% help build intuition. If this is already written how about bringing it back to
% the appendix?}\ca{The cases presented were too simple}

\iflater
\gm{Do we anywhere require these two to be inverses or somehow related?
That is a natural question at this point.
I suppose it is not immediately needed since the contracts will check
whatever property is needed of the values.}\ca{No, we do not
require that, it is not needed.\ch{not needed for what we managed to already prove, but ...}
We plan to add `import (export (f)) = f` to prove compiler correctness.}
\fi

\subsection{Exporting and importing functions}
\label{sec:export-import-fun}
The situation is most interesting
% more complicated\ch{more complicated than what? we have no simple case earlier now}
for functions, where exporting and importing
make use of each other.
The basic idea is that in order to export a function \ls{f}, we create a
intermediately-typed
% \gm{hmm}\tw{?}\gm{just unsure on whether the wording sounded ok :)}
function that imports its argument, calls \ls{f},
and exports the result back to an intermediate type~\cite{FindlerF02}.
Essentially, the composition \ls{export}~$\cdot$~\ls{f}~$\cdot$~\ls{import}.
Dually, an intermediately-typed function \ls{g} is imported roughly as
\ls{import}~$\cdot$~\ls{g}~$\cdot$~\ls{export}.
However, as \ls{import} can fail (it returns a sum type),
we also require the codomain of each imported and exported
function to ``include'' errors.
Below is the instance for exporting simple functions, \IE
of type \ls{t_1 -> MIO fl (either t_2}~\ls{err) True}~$\Sigma$ without other pre- and post-conditions:

\begin{lstlisting}
instance exportable_smpl_arr (t_1 t_2:Type) (fl:erased tflag) ($\Sigma$:policy_spec) (cks:checks{EmptyNode? cks}) 
  {| d_1 : importable_to t_1 fl $\Sigma$ (left cks) |} {| d_2 : exportable_from t_2 fl $\Sigma$ (right cks) |}
: exportable_from (t_1 -> MIO (either t_2 err) fl $\top$ $\Sigma$) $\Sigma$ fl cks
= { ityp = d_1.ityp -> MIO (either d_2.ityp err) fl $\top$ $\Sigma$; interm_ityp = solve;
   export = (fun cks (f:(t_1 -> MIO (either t_2 err) fl $\top$ $\Sigma$)) (x:d_1.ityp) ->
      match x' <-- d_1.import x (left cks); f x' with       (** do-notation **)
      | Inr err -> Inr err | Inl y' -> Inl (d_2.export (right cks) y') ) }
\end{lstlisting}
The function is exported to the type \ls{i_1 -> MIO fl (either
i_2 err) True}~$\Sigma$, where \ls{i_1},\ls{i_2} are the intermediate types corresponding to
the strong types \ls{t_1},\ls{t_2} according to their instances. \fstar{}
automatically checks that the type of the function
is intermediate too (using the
\ls{interm_arrow} instance from \autoref{fig:interm}). Exporting produces a function
that imports its argument, applies \ls{f}, and exports the result with
\ls{Inl}. If either importing or the function itself fail (returning
\ls{Inr}), the error is returned instead.
 
We see here that the \ls{cks} collection has the structure of a
binary tree, defined using the inductive type \ls{checks}.
Since this function does not have a pre-/post-condition to enforce,
the root node must be an \ls{EmptyNode}, and its left and right
children
% branches of the tree
are used to % further
import \ls{t_1} and export \ls{t_2}.
For a function with pre- and post-conditions, of type \ls{a -> MIO b pre
post}, the tree takes the shape \ls{Node ck left right}, where \ls{ck}
is a dynamic check testing whether the post-condition \ls{post} indeed holds,
and \ls{left} and \ls{right} are sub-trees for \ls{a} and \ls{b} respectively.
For structured, non-arrow types such as tuples or sums, we also use
an \ls{EmptyNode left right} node, as there is no immediate pre- and
post-condition to enforce, but there may be some deeper within the type.
The \ls{Leaf} node is used for base types when there are no checks to perform.

When importing or exporting a function with pre- and
post-conditions, the \ls{cks} collection becomes important.
The \ls{cks} collection must contain a boolean predicate for each pre- and
post-condition (that needs enforcement).
Hence, the instances of \ls{importable_to} and \ls{exportable_from}
require a \ls{cks} collection and alongside a proof that each check
implies the actual (propositional) pre- or post-condition.
Crucially, the dynamic checks have the same structure as the post-condition
so that they can also distinguish between the history and the local trace of a
computation, and they can also distinguish the events of the partial program
from the events of the context. Since the types are higher-order,
the trace of a computation can contain events generated by both the partial
program and the context.%
%\gm{explain more, also should state that instances
%are not in control of the target, or this can break anyway}
%CA: there is no mention of target lang in this section
%
Therefore, we give the following type to dynamic checks:
\begin{lstlisting}
type ck_typ (t_1 t_2 : Type) = t_1 -> h:trace -> t_2 -> lt:trace -> bool
\end{lstlisting}
Given the type of the argument and return value of a function, a dynamic
check is a boolean predicate over the value of the argument, the history before
the call was made, the return value, and the local trace (\IE the events
generated during the activation of the function).
%
% Crucially, the result is in \ls{bool}, as this check must be computable.
%
%For the handler, the dynamic check in question is the one presented in 
%\autoref{sec:key-contracts}, \ls{check_handler_post}.

To enforce the dynamic checks, the compiler automatically generates the
corresponding effectful checks with the same tree structure. The effectful
checks are passed as an argument to the \ls{import} and \ls{export} functions.
An effectful check guarantees that it enforces the corresponding dynamic check.
An effectful check is done in two steps: a ``setup'' and an actual check.
First, for the setup, we must mark the point in history where the
function was initially called, in order to able to determine its
local trace when it returns.
This is accomplished by the effectful check because after the first
call when it captures the initial trace, it returns
a second function that enforces the dynamic check.

\paragraph{Enforcing the post-condition when importing.}
\label{sec:enforce-post}
We are now ready to import a function with an intermediate type into
one with a strong type that has a pre- and a post-condition.
The pre-condition does not need to be dynamically checked, because one can
always strengthen the pre-condition of a function using subtyping,
%,\ch{didn't get this explanation about ``wrapping''; do you actually mean subtyping?}
thus we only have to enforce the post-condition.

Here we show just a combinator that is used in an \ls{importable_to} instance
that focuses only on enforcing the post-condition.
The combinator takes an intermediately-typed function \ls{f}, a dynamic check \ls{ck}, and its effectful counterpart
\ls{eff_ck}, and returns a strongly-typed variant of \ls{f}.
The combinator first ``sets up'' the effectful check, obtaining the
\ls{do_ck}, then it calls \ls{f}, and then it enforces the post-condition
by running \ls{do_ck}, returning an error if the check fails.

\begin{lstlisting}
let enforce_post (t_1 t_2 : Type) (fl:erased tflag) ($\Sigma$:policy_spec)
  (ck : ck_typ t_1 (either t_2 err)) (eff_ck : eff_ck_typ fl ck) 
  (pre : t_1 -> trace -> Type0) (post : t_1 -> trace -> either t_2 err -> trace -> Type0)
  (c1_post : squash (forall x h r lt. pre x h /\ enforced_locally $\Sigma$ h lt /\ ck x h r lt ==> post x h r lt))
  (c2_post : squash (forall x h lt. pre x h /\ enforced_locally $\Sigma$ h lt ==> post x h (Inr Contract_failure) lt))
  (f : t_1 -> MIO (either t_2 err) fl $\top$ $\Sigma$)
: (x:t_1) -> MIO (either t_2 err) fl (pre x) (post x)
= fun x -> let do_ck = eff_ck x in
       let r : either t_2 err = f x in
       if do_ck r then r else Inr Contract_failure
\end{lstlisting}
To import a function, two constraints
have to hold between the pre- and post-condition, the dynamic check, and
the access control policy,
constraints that ensure that the post-condition can be
soundly enforced.
The first constraint, \ls{c1_post}, ensures that the policy together with the
dynamic check enforce the post-condition.
This constraint allows us to return the result when the
check succeeds by giving us that the post-condition holds.
The second constraint, \ls{c2_post}, ensures that the post-condition can be enforced
even if the dynamic check fails as long as the policy was enforced.
This constraint allows us to return \ls{Inr Contract_failure} when the
check fails after we called the function.

\paragraph{Enforcing the pre-condition when exporting.}
\label{sec:enforce-pre}
Exporting a strongly-typed function implies
converting its pre-condition into a dynamic check and involves
a mechanism similar to the one shown above for post-conditions.
The main difference is that checking pre-conditions does not require a
setup-check split, and can be done at once.
To denote a dynamic check for a pre-condition, we reuse the \ls{ck_typ}
type above, using \ls{unit} for the codomain and providing
an empty local trace to the checks.
We show below the \ls{enforce_pre} combinator that is used in
instances of \ls{exportable_from}:

\begin{lstlisting}
let enforce_pre t_1 t_2 fl $\Sigma$
  (ck : ck_typ t_1 unit) (eff_ck : eff_ck_typ fl ck)
  (pre : t_1 -> trace -> Type0) (post : t_1 -> trace -> either t_2 err -> trace -> Type0)
  (c_pre : squash (forall x h. ck x h () [] ==> pre h))
  (c_post : squash (forall h lt r. pre h /\ post h r lt ==> enforced_locally $\Sigma$ h lt))
  (f : (x:t_1 -> MIO (either t_2 err) fl (pre x) (post x)))
: (x:t_1 -> MIO (either t_2 err) fl True$\ \Sigma$)
= fun (x:t_1) ->  if eff_ck x () then f x else Inr Contract_failure
\end{lstlisting}
To export a function, also two constraints have to hold. 
The first constraint, \ls{c_pre}, ensures that the dynamic check enforces
the pre-condition. This constraint gives us that the pre-condition holds
after running the effectful check, allowing us to call the function.
The second constraint, \ls{c_post}, ensures that we can weaken the
post-condition to the access control policy by subtyping.

Effectful checks are the key technical novelty to define back-translation
between a target and a source context that are both flag-based effect
polymorphic (\autoref{sec:rrhp}).

\iflater
    \subsection{Non-interference theorem about the context}

    To formalize that a flag-polymorphic context can't directly
    call the IO operations + GetMState we [use parametricity? to] prove a
    noninterference property showing that a flag-polymorphic context cannot obtain
    more information about the trace than what our higher-order contracts+instrumentation\ch{reference monitor} reveals to it.
    \tw{Let us worry about this part after we have at least some hope that we can do
    it formally.}

    \begin{lstlisting}
    val ni :
    $\Sigma$ : policy_spec ->
    (** for any dynamic check **)
    rc : (rc_typ 'a 'b) ->
    (** the ctx is in a first-order setting. I don't think it matters **)
    ctx: (fl:erased tflag -> $\Sigma$:erased policy_spec -> io_lib fl $\Sigma$ -> typ_eff_rcs fl (make_rc_tree rc) -> unit -> MIO int fl (fun _ -> True) (fun _ _ _ -> True)) ->
    Lemma (
                                    (** one has to instantiate the ctx to be able to call beh **)
        let bh = beh_ctx #(fun _ -> True) (ctx AllOps $\Sigma$ (inst_io_cmds pi) (make_rcs_eff (make_rc_tree rc))) in
        forall h1 lt1 r1 h2 lt2 r2.
        (h1, (lt1, r1)) `pt_mem` bh /\
        (h2, (lt2, r2)) `pt_mem` bh ==>
        ni_traces $\Sigma$ rc r1 r2 h1 h2 [] lt1 lt2)
    \end{lstlisting}
\fi

\section{Dynamic Enforcement Over A Monitor State}
\label{sec:abstract-state}

\noindent
So far, we have described our specifications and dynamic checks
over IO traces. While traces are intuitive and very
expressive, making them a fine choice for specifications, they
are inadequate for an implementation as their size may grow without bound.
We would rather store only the information that is
\emph{needed} to implement the relevant dynamic checks. For example,
if all dynamic checks are about whether file
descriptors are open or not, the monitor could keep a set of
currently open file descriptors and efficiently probe this set
instead of inspecting a trace.
%
% The set must of course also be updated as file descriptors are created
% or closed, but this also can be done efficiently.
%
Importantly, independent of this runtime aspect we wish to still write
specifications over traces, to retain full generality and a single
\emph{lingua franca}.
\sciostar{} allows for such setups with full flexibility, allowing the
programmer to choose exactly what the runtime monitor state should be
and how the checks are performed, while still guaranteeing correctness
over traces. We explain the mechanism in the rest of this section.

\subsection{The monitor state}

The entire \sciostar{} development is in fact parametrized over a type
of \emph{monitor state}, which represents the
runtime information needed to implement the dynamic checks.
Most of the definitions shown previously actually have an extra argument
of type \ls{mstate}. We elided them for simplicity of the previous
sections, but will make them explicit here.

\begin{lstlisting}
type mstate = { typ : Type; abstracts : typ -> trace -> Type0 }
\end{lstlisting}

The \ls{mstate.typ} type must be instantiated by the programmer, who chooses it
according to the dynamic checks that the whole program
% \gm{By that I mean
% context+program, is it clear? Also is it clear that the choice of state
% is ``global'', for the full framework?}\tw{whole program?}
must perform. For instance, the type
\ls{list file_descr} would be a sensible choice if all one needs to check
is whether file descriptors are open.
The monitor state is exactly \ls{mstate.typ}.
% it is
% updated in-place\gm{technically untrue, we take snapshots for the setup/check,
% but at least, there's no unbounded growth}.
%
Program traces are never materialized: they only appear
in specifications.

% \footnotetext{Of course, the programmer could also choose to use full
% traces as the monitoring state, and implement the dynamic checks by
% inspecting the traces as exemplified in previous sections. We however
% do not expect this to be practical in real systems.}

The programmer must also relate the monitoring state to the trace of the
program by defining a predicate \ls{abstracts}.
This relation specifies how the monitor state abstractly models the
trace, and must be kept
invariably
true when applied to the current
monitoring state and the current (ghost, immaterial) trace.
For open file descriptors, one could use
\ls{(fun fds h -> forall fd. fd `mem` fds <==> is_open fd h)},
stating that a descriptor \ls{fd} is in the list if and only if
it is open according to the current trace \ls{h}.

So in fact,
the \ls{get_mstate} operation does not
return a trace, but a monitor state. Its post-condition also guarantees that the returned
state abstracts the current trace (which is still accessible in
specifications). The implementation of \ls{mio_wps} is adjusted similarly.
\begin{lstlisting}
let get_mstate (#mst:mstate) () :
  MIO mst.typ GetMStateOps (requires (fun h -> True)) (ensures (fun h s lt -> s `mst.abstracts` h /\ lt == [])) =
  MIO?.reflect (Call Prog GetMState () Return)
\end{lstlisting}%\er{used MIO?.reflect... not a terrible lie}%
\newtext{The recording of IO operations in the monitor state is abstracted away in
  \sciostar{}. The lower-level implementation of the state updates is
  explained in \autoref{sec:running-case-study}, where we discuss execution in OCaml.}

% \newtext{
% Since the reference monitor is abstracted away, it means that we
% assume that the monitor state is updated and abstracts the trace.
% %Wemake this assumption because the \ls{MIO} monadic effect does not
% %support the state effect, thus, we cannot allocate the monitor state.
% }\ch{\bf The justification in this last phrase seems very unconvincing to me. Can we just drop it?}%
% \ca{This is part to a change promised to reviewer C, and it was also
%   part of the rebuttal. I dropped the justification, but I'm not sure if this
%   one sentence is enough or how to improve it.}

\subsection{Enforcing the access control policy and the checks over the monitor state}

%t type\ch{According
%  to my suggestions above this should be called not ``state of the abstract type'',
%  but ``abstract state (of the correct type)''}
%that abstracts the trace.
% In particular in \ls{mio_wps} (\autoref{sec:monad-morphism}),
% for a monitor state \ls{mst},
% instead of returning \ls{h}, it returns any \ls{s} that abstracts \ls{h}:
%\ls{forall (s:mst.typ). s `mst.abstracts` h ==> p [] s}.
%\ch{Nontrivial to understand where this goes.
%  Can't we give the new definition, with the unchanged branch in ``...''?
%  I guess something like this:}
% \begin{lstlisting}
% let mio_wps (c:caller) (op:mio_ops) (arg:mio_sig.args op) : wp (mio_sig.res op arg) = fun post h ->
%   if GetMState? op then forall (s:mst.typ). s `mst.abstracts` h ==> p [] s else ...
% \end{lstlisting}

The access control policy and higher-order contracts must also work
over monitor states.
A policy, which must soundly enforce a higher-level \ls{policy_spec}
over traces, now takes a monitor state as argument, instead
of a trace. The policy must ensure that, when it returns true, the
policy specification is satisfied for \emph{any} trace that is abstracted
by the current monitor state.

% To work with a monitor state, we have to redefine more of the definitions
% from the previous sections.
% The monitor state is used both to enforce the access control policy and
% the higher-order contracts.
% Therefore, we redefine the \ls{policy} type to work with a monitor state:

\begin{lstlisting}
type policy (mst:mstate) ($\Sigma$:policy_spec) =  s_0:mst.typ -> 
    op:io_ops -> arg:io_sig.args op -> r:bool{r ==> (forall h. s_0 `mst.abstracts` h ==> $\Sigma$ h Ctx op arg)}
\end{lstlisting}
Coming back to the open files example, this means that a policy can only inspect the set of files to decide whether the operation should be
allowed, and the answer must be right for any trace that would lead to
the current open file set.

As for the dynamic checks (\ls{ck_typ}), they are defined over an initial and
final state of the function:
\begin{lstlisting}
type ck_typ (mst:mstate) (t_1 t_2 : Type) = t_1 -> s_0:mst.typ -> t_2 -> s_1:mst.typ -> bool
\end{lstlisting}
Relating these states to the trace is done in the pre-conditions
of the import and export functions from \autoref{sec:export-import-fun}.
%
% We also have to redefine two of our constraints so that they require
% that the checks over states imply the specifications over traces.
%
One is the first constraint required when exporting a function
(\ls{c_pre} from \autoref{sec:enforce-pre}) and the second one is
the first constraint required when importing a function
(\ls{c1_post} from \autoref{sec:enforce-post}):
\begin{lstlisting}
(c_pre : squash (forall x s_0 s_1 h. s_0 `mst.abstracts` h /\ s_1 `mst.abstracts` h /\ ck x s_0 () s_1 ==> pre x h))
(c1_post : squash (forall s_0 s_1 x h r lt . pre x h /\ enforced_locally $\Sigma$ h lt /\
       s_0 `mst.abstracts` h /\ s_1 `mst.abstracts` (rev lt @ h) /\ ck x s_0 r s_1 ==> post x h r lt))
\end{lstlisting}
%\gm{pretty unclear I think, those two lines are from separate functions.}
% \gm{also \ls{c_pre} does not make sense to me, why are there two states and only
% one trace? is this a typo? if so I can come back and attempt a new explanation}\ca{it is
%   not a typo. enforcing the pre-condition runs \ls{get_state} twice, but there is no guarantee
%   that \ls{get_state} returns the same state... the only thing we have is that it returns
%   two states that both abstract the same history}
\ca{The constraint \ls{c_pre} quantifies over two states because the contract that enforces
the pre-condition calls \ls{get_mstate} twice; this is an artifact of the fact that we have
only one type of contracts that is used to enforce pre- and also post-conditions.}

% The fact that the \ls{get_mstate} operation returns a state that abstracts the history
% makes these constraints as easy to solve as before for the SMT solver.
% \gm{I would just drop this, but if not, at least say something like that the proofs
% are still straightforward. Less mentions of F* and SMT would be good.}

% \subsection{Enforcing the checks on state}
% \gm{I'm not sure I can improve this currently, without
% editing elsewhere.
% In general I feel these points could have been made earlier
% over the traces, there's nothing "monitor state"-specific here, right?}
%   \ca{Generally, I agree, but I did not want to show types we do not use.
%   I think it is quite late to change this.}

\iflater
\gm{At a higher level I don't know (or remember) why the setup
returns a function instead of just a copy of the state/trace, \ls{enforce_post}
could call ck directly}
\ca{It is very important to first capture the initial state, run the function, and then
  test the post-condition. Maybe one can do the setup/capture differently, but not
  sure if one can call the ck directly. Note, one cannot call \ls{get_mstate} here.}
\fi

As explained in \autoref{sec:exportable-importable}, we have to
abstract away the enforcement of the checks by using effectful checks.
%
% An effectful check enforces a check by calling the \ls{get_mstate} operation.
% GM: ^ doesn't add much
%
They are represented by the type \ls{eff_ck_typ} below.
%\ch{You were/are talking   about \autoref{sec:exportable-importable} though?\ca{yes} (see the following big comment)}
Enforcing a dynamic check \ls{ck} is done in two steps: a ``setup'' and an actual check.
\iflater
\ch{Still talking about 
  \autoref{sec:exportable-importable} since you're just repeating stuff?
  (see the following big comment). \ca{the repetition is from 4.3, but 4.2 is correctly referenced in the beginning. not sure how to avoid the repetition}}
\fi
The setup is done by calling the effectful check,
which captures the current state \ls{s_0} and returns it together
with the second part of the effectful check.
The returned state is labeled with \ls{erased},
thus it cannot be used in the computation: it is only there
in order to specify the effectful check.
\iflater
\ch{So all the
  text so far is just repeating stuff from \autoref{sec:exportable-importable}?
  It's very unclear to me at what point in this text do we stop repeating stuff
  and switch and talk about the monitor state (the subject of this section)?
  \ca{the enforcement is on states, notice that this subsection does NOT exist anymore in 4, it was moved here}
  Please make this clear. The section heading also doesn't help in this respect,
  since this heading could also appear in section 4.\ca{changed heading}\ca{LE: heading was removed}}
\fi
The second part, of type \ls{eff_ck_typ_cont}, is called to enforce
the dynamic check \ls{ck} and guarantees that the property holds: \ls{(b <==> ck x s_0 y s_1)}.
% \tw{Weird sentence:
%   the check is called to enforce the check and guantees the check.}\ch{+1; probably the first
%   ``check'' should be ``effectful check''? Even with that this sentence is probably still weird.}
This function also returns an erased state, again only for specification purposes.
% enforced.
%
This gives us the ability to enforce the dynamic check on two different states---\EG
capture \ls{s_0} before calling a function and \ls{s_1} after the function returns---and
this enables us to enforce post-conditions.
Inside the \ls{enforce_pre} and \ls{enforce_post} combinators,
the specifications of the effectful check that
\ls{s_0} and \ls{s_1} abstract the histories, plus the constraints,
make it very easy to prove that the enforcement is done correctly.

\begin{lstlisting}
type eff_ck_typ_cont (mst:mstate) (fl:erased tflag) (t_1 t_2:Type) (ck:ck_typ mst t_1 t_2) (x:t_1) (s_0:mst.typ) =
  y:t_2 -> MIO (erased mst.typ * bool) mst fl True
       (ensures (fun h_1 (s_1, b) lt -> lt == [] /\ s_1 `mst.abstracts` h_1 /\ (b <==> ck x s_0 y s_1))) 
type eff_ck_typ (mst:mstate) (fl:erased tflag) (t_1 t_2:Type) (ck:ck_typ mst t_1 t_2) =
  x:t_1 -> MIO (s_0:erased mst.typ & eff_ck_typ_cont mst fl t_1 t_2 ck x s_0) mst fl True
             (ensures (fun h_0 (| s_0, _ |) lt -> s_0 `mst.abstracts` h_0 /\ lt == []))
\end{lstlisting}

Converting a pure dynamic check into an effectful check is straightforward:

\begin{lstlisting}
let make_check_eff mst (t_1 t_2:Type) (ck:ck_typ mst t_1 t_2) : eff_ck_typ mst GetMStateOps ck = fun (x:t_1) ->
  let s_0 = get_mstate () in
  let cont = (fun (y:t_2) -> let s_1 = get_mstate () in (hide s_1, ck x s_0 y s_1)) in (| hide s_0, cont |)
\end{lstlisting}

% \gm{attempt at a local closing of this subsection, it's not great.
% Consider removing this whole thing.}%
% Once this is all done, the traces are no longer needed at runtime, and
% the programmer is allowed to precisely define the runtime state of the
% monitor
% % remove{, how it is updated on new events,}
% and how the dynamic checks
% are implemented.
%When new properties need to be checked, this state
%can be augmented with more information (and shrink the
%set of traces a state can be related to) used to implement the new checks.
% Each part of this mechanism is fully verified, and does not compromise
%on the expressivity of traces, which remain the single language
%of specification.

\subsection{The web server's monitor state}
\label{sec:webserver-state}

%
% \ch{This jumps into details before explaining the high-level picture. Also bad flow,
%   since the text doesn't seem to have anything to do with the title of the subsection.
%   For a start tried to add comments in order to fix the title and text.}
% %
%   \gm{agree, I'm changing this to talk about the webserver in particular, and moved to
%   the end of S5}
%   \gm{did some edits, maybe too short now}
% \old{
% To initialize the monitor state in \sciostar{}, the user has to provide an
% inhabitant of type \ls{mstate}, and 
% %
% also the dynamic checks and the access control policy
% by instantiating the \ls{importable} and \ls{exportable} type classes (in case that \fstar{}
% cannot do this automatically). When doing this, the user also has to prove all the required
% constraints.
% %
% Naturally, we did this for our web server case study, thus, we're going to present here
% what state we chose for our example and how the dynamic checks look like.
% }

For our case study, we choose the following implementation of the monitor state:
(a) a list of file descriptors that were opened by the context,
(b) a single boolean flag that determines whether the current request was answered, and
(c) a list of the file descriptors that were written to.
We do not need to track \emph{all} of the open files, but only those
opened by the context. %, as they are the ones relevant to the checks.
This state contains enough information to
enforce the access control policy (\ref{spec:handler_post_acp}),
the pre-condition of \ls{send} (\ref{spec:send}),
and the first part of the post-condition of the handler (\ref{spec:handler_post_ctrs}), and
in \ls{abstracts} we can see a conjunction of all of these.
\begin{lstlisting}
let myst : mstate = {
  typ = { ctx_opened : list file_descr ; responded : bool ; written : list file_descr };
  abstracts = (fun s h -> (forall fd. fd `mem` s.ctx_opened <==> is_opened_by_Ctx fd h) /\
    (responded <==> ~(did_not_respond h)) /\ (forall fd. fd `mem` s.written <==> wrote_to fd h)); }
\end{lstlisting}
%
% The definition of \ls{abstracts} is straightforward.
%
% This suffices to implement specifications
% (3) the pre-condition of \ls{send}: we have not responded to the request;
% (4) checking that the request handler indeed wrote to the file descriptor;
% and (5) that the handler only reads from the files it opens.
%
%
%
% \tw{You mean conjunction? `the conjunction of all of them' I guess?}\ca{I want
% emphasis on `for each'}\tw{No but `conjuncture' means `a combination of
% events, a state of affairs'.} CH: it should clearly be conjunction
%\old{
%The access control policy looks at the monitor state to check if the handler reads only from files it opens,
%thus we have to keep track of what files the handler opens (\ls{ctx_opened}).
%We also keep track of whether the handler responded to the client it was passed using \ls{responded}.\tw{responded does not check the client, as in the spec.}
%Lastly, we maintain a list of files to which the program has written since it accepted each request.
%}
%
The access control policy is already defined in \autoref{sec:key-monitor} and to use it
with our chosen state we have to replace the \ls{is_opened_by_Ctx fd h} with
\ls{fd `mem` s.ctx_opened}. Below are the dynamic checks required to import the
handler, defined as a \ls{cks} collection of type \ls{checks} (\autoref{sec:export-import-fun}).
The first node is for the post-condition of the handler (\ref{spec:handler_post_ctrs}) and the next
one is for the pre-condition of \ls{send} (\ref{spec:send}).
%\gm{lost me here, afaics this type was not introduced, only mentioned briefly}
%
\begin{lstlisting}
let handler_cks : checks myst =
  Node (| file_descr, unit, (fun client s0 _ s1 -> not (client `mem` s0.written) && client `mem` s1.written) |)
    (Node (| Bytes.bytes, unit, (fun res s0 _ _ -> not (s0.responded) && valid_http_response res) |) Leaf Leaf)
    Leaf
\end{lstlisting}
When instantiating the \ls{importable_to} type class using these checks,
\fstar{} automatically proves that three of the four constraints hold,
asking for help only to prove that the first dynamic check implies the post-condition
of the handler.
Moreover, \sciostar{} guarantees that no other check is necessary.

\section{\sciostar{}: Formally Secure Compilation Framework}
\label{sec:secure-compilation}

\begin{figure}
\begin{mdframed}[backgroundcolor=black!5,hidealllines=true]
\begin{lstlisting}[xleftmargin=-6pt,framexleftmargin=0pt]
type interface^S = {
  ctype : erased tflag -> Type;
  $\Sigma$ : policy_spec;
  $\Pi$ : policy $\Sigma$;
  cks : checks;
  importable_ctype : fl:erased tflag -> importable_to (ctype fl) fl $\Sigma$ cks;
  $\psi$ : post_cond; }

type interface^T  =  {
  ctype : erased tflag -> policy_spec -> Type;
  $\Sigma$ : policy_spec;
  $\Pi$ : policy $\Sigma$;
  weak_ctype : fl:erased tflag -> interm (ctype fl $\Sigma$) fl $\Sigma$; }

let compile_interface (I$^S$:interface$^S$) : interface$^T$ = { (** denoted by $\cmp{I^S}$ **)
  $\Sigma$ = I$\overset{S}{.}$$\Sigma$;
  $\Pi$ = I$\overset{S}{.}$$\Pi$;
  ctype = (fun fl _ -> (I$\overset{S}{.}$importable_ctype fl).ityp);
  weak_ctype = (fun fl -> (I$\overset{S}{.}$importable_ctype fl).interm_ityp); }

type prog$^S$ I$^S$ = fl:erased tflag -> I$\overset{S}{.}$ctype (fl+IOOps) -> unit -> MIO int (fl+IOOps) $\top$ I$\overset{S}{.}\psi$
type ctx$^T$ I$^T$ = fl:erased tflag -> #$\Sigma'$:erased policy_spec -> io_lib fl $\Sigma'$ -> I$\overset{T}{.}$ctype fl $\Sigma'$

type prog$^T$ I$^T$ = I$\overset{T}{.}$ctype AllOps I$\overset{T}{.}\Sigma$ -> unit -> MIO int AllOps $\top$ $\top$
type whole$^T$ = unit -> MIO int AllOps $\top$ $\top$
let link$^T$ #I$^T$ (P:prog$^T$ I$^T$) (C:ctx$^T$ I$^T$) = P (C AllOps (enforce_policy call_io I$\overset{T}{.}$$\Pi$)) (** denoted by $C[P]$ **)

let compile_prog #I$^S$ (P:prog$^S$ I$^S$) : prog$^T$ (I$^S\cmparrow{}$) = (** denoted by $\cmp{P}$ **)
  fun (C : (I$^S\cmparrow{}$).ctype AllOps I$\overset{S}{.}$$\Sigma$) -> P AllOps (import C (make_checks_eff I$\overset{S}{.}$cks))

type ctx$^S$ I$^S$ = fl:erased tflag -> io_lib fl I$\overset{S}{.}$$\Sigma$ -> eff_checks I$\overset{S}{.}$cks -> I$\overset{S}{.}$ctype fl
type whole$^S$ = $\psi$ : post_cond & (unit -> MIO int AllOps $\top$ $\psi$)
let link$^S$ #I$^S$ (P:prog$^S$ I$^S$) (C:ctx$^S$ I$^S$) = (** denoted by $C[P]$ **)
  (| I$\overset{S}{.}\psi$, P AllOps (C AllOps (enforce_policy call_io I$\overset{S}{.}$$\Pi$) (make_checks_eff I$\overset{S}{.}$cks)) |)

let back_translate_ctx #I$^S$ (C:ctx$^T$ I$^S\cmparrow{}$) : ctx$^S$ I$^S$ = fun fl sec_io -> import (C fl sec_io) (** denoted by $C^T\bakarrow{}$ **)
\end{lstlisting}
\end{mdframed}
\caption{Secure compilation framework. Idealized mathematical notation.
%\ca{
%  is it still idealized?}\ch{I expect you know better how big the difference is
%  to actual code. I guess that we don't use any Greek symbols in the code
%  though? Superscripts? Downarrows? Uparrows? :)}\ca{indeed, we don't use them}
% \ca{There are too many redundant $^T$ and $^S$, better notation needed.}
%  CH: seems okay to me
%
% \ca{Not sure if this mix between symbols and \fstar code is ok. Trying to bridge the
% gap between the formalization in \fstar and the notation used later. It may be confusing
% and maybe using only symbols would be better?}
%  CH: seems okay to me
%
% \ca{I noticed that one could write the type of the interface as: (see comment)
% %type interface$^S$ $\Sigma$ ctrs = t:(erased tflag -> Type){forall fl. importable (t fl) $\Sigma$ ctrs fl}
% However, this made every other definition more verbose.}
}
\label{fig:model_sec_comp}
\end{figure}

\noindent
We now put all the pieces together and define the \sciostar{} secure compilation
framework (\IE languages, compiler, and linker in \autoref{sec:compilation-framework}),
give semantics to the source and target languages (\autoref{sec:beh}), and 
describe our machine-checked proofs that \sciostar{}
soundly enforces a global trace property (\autoref{sec:soundness})
and satisfies by construction a strong secure compilation criterion (\autoref{sec:rrhp}).
We present this in the setting where the partial program has initial
control and the context is the library, then also describe the dual setting when
the context has initial control and the program is the library (\autoref{sec:dual-setting}).

\subsection{Secure compilation framework}
\label{sec:compilation-framework}

We first use the definitions from the previous sections to
formally define our source and target languages, compiler, and linker---to which
we jointly refer as the \sciostar{} secure compilation framework.
We define our languages using shallow embeddings, which is a common
way to express DSLs in \fstar{}~\cite{lowstar,BhargavanBDHKSW21,steel},
and which also simplifies our proofs.
%
% We model our secure compilation framework in \fstar
% using a simple\ch{in what way simple?} shallow embedding, which is enough
% to capture the key ideas\ch{Unsure what this is trying it say and whether it's a positive thing.
%   This is not a toy that only illustrates the key ideas, is it? It's our actual compilation framework.}
% and it also simplifies our proofs.
%
% \tw{Dumping Figure~\ref{fig:model_sec_comp} on the reader with no explanation is a bit rough.}\ch{+1;
%   it's a monster and the reader will have no clue where to even start. Much better hand-holding needed,
%   otherwise this is not understandable at all, even for me. Think of the abstraction level
%   of your explanations in the PriSC talk, which was definitely not as gory as this figure.
%   For a start I moved the reference to the figure to the next paragraph,
%   where it seems you do guide the reader to specific parts of the figure.}
%
We represent the partial program as a function (since
%, like in other proof assistants\mw{"e.g. like in Coq"
%There are proof assistants, where modules are implemented shallowly. E.g. Isabale/HOL's "locales"},
in \fstar{} modules are not first-class objects and one
cannot reason about them).
%\ch{Unclear from the figure that context has to always be a function.}\ca{The context does not have to be always a function}
%
We now describe the setting where the partial program has the initial control,
thus we make the partial program get the context as argument
(after the context is instantiated with the flag, etc)---the setting
in which the context starts is outlined in \autoref{sec:dual-setting}.

The \sciostar{} framework is listed in \autoref{fig:model_sec_comp} and explained step by step below.
As explained before, the partial program and the context share a higher-order interface.
The difference between the source and the target language is that in
the source language, the partial program and the context share a {\em source} interface,
while in the target language they share a {\em target} interface.
The source interface (type \ls{interface}$^S$) is a dependent record type that contains
the type of the context (denoted by \ls{ctype}),
the access control policy $\Pi$ and its specification $\Sigma$ (both explained in
\autoref{sec:key-monitor}),
and the dynamic checks, \ls{cks}, that are enforced by the higher-order contracts 
(explained in \autoref{sec:contracts}).
The interface also contains a type class constraint \ls{importable_ctype}
ensuring that the guarantees the type of the context (\ls{ctype})
provides to the program
can be enforced using the policy $\Pi$ and the dynamic checks \ls{cks}.
The \ls{importable_ctype} constraint also ensures that \ls{ctype} is flag-polymorphic.
% parametric in the flag.

The target interface (type \ls{interface}$^T$) is another record type that contains
the type of the context (\ls{ctype}), but this time, the field \ls{weak_ctype} requires
\ls{ctype} to be a weak type (explained in \autoref{sec:key-context}).
The interface also includes the policy $\Pi$ and its
specification $\Sigma$, which are used during target linking.

The type of the partial source program (\ls{prog}$^S$)
and of the target context (\ls{ctx}$^T$) can also be found in \autoref{fig:model_sec_comp}.
The partial source program is a computation in the \ls{MIO} monadic effect
that can call IO operations (\ls{fl+IOOps}), but that, since it is statically verified, has no
need to call the \ls{get_mstate} operation, so it is otherwise parametric in the flag.
The type of the target context uses \ls{ctype} from the target interface \ls{I}$^T$
and the constraint \ls{weak_ctype} from \ls{I}$^T$ makes sure that \ls{ctype}
is a weak type parametric in the flag.
% \ch{invariant in the flag? What is that
%   supposed to mean? parametric? flag-polymorphic?}
%
Thus the target context {\em cannot} directly call the IO operations or the
operation \ls{get_mstate} because of flag-based effect polymorphism (as explained
in \autoref{sec:key-context}),
%\ch{is it \autoref{sec:key-context} now?}),
and has to instead use the secure IO operations
(of type \ls{io_lib fl}~$\Sigma'$, type explained in \autoref{sec:key-monitor}) to do any IO.

The compilation of a partial source program (\ls{compile_prog}) is defined
as explained in \autoref{sec:key-contracts}.
The result of the compilation is a partial target program (\ls{prog}$^T$)
that expects a target context that was instantiated so that it satisfies
specification \ls{I}$\overset{T}{.}\Sigma$.
The \ls{MIO} computation of the partial target program (type \ls{prog}$^T$)
has no pre- and
post-condition and is indexed by the flag \ls{AllOps} to be able to run
the dynamic checks that call the \ls{get_mstate} operation.

The target linking (\ls{link}$^T$) first instantiates the target context and
then it applies the partial program to the instantiated context.
The context is instantiated with the \ls{AllOps} flag (since the context calls
into the secure IO operations and wrapped callbacks that use \ls{get_mstate}) and
with the secure IO operations, thus obtaining an instantiated context that
satisfies specification \ls{I}$\overset{T}{.}\Sigma$ (as also explained in
\autoref{sec:key-monitor}).
The result of the target linking is a whole program (\ls{whole}$^T$), which is a
\ls{MIO} computation with no pre- and post-condition, that takes a unit
and returns an \ls{int}.
\iflater
\ch{Why do our program types have useless unit arguments? Warning: We're sending
  this to ICFP :)}\ca{tried to address this}\ch{You explained the 1/4 unit that
  makes sense, but I still see 3/4 spurious-looking occurrences of \ls{unit} in
  Figure 3. Are they not spurious?}\ca{I would say that the explanation covers
  2/4 cases (in the definitions of whole programs).\ch{Okay, 2/4, added now the missing
  parenthesis in wholeS that caused me to misparse it.}  And these two necessary units are
  forcing us to either have an \ls{fun () -> ...} inside linking or unit in the type
  of partial programs.}\ca{better?}\ch{No, still not. Why would necessary thunking
  propagate to unnecessary thunking? Instead of adding an unnecessary unit argument
  to prog just define linking better (hint: use a lambda).}
\ch{Anyway, the hope of my comment is not to add any lame excuses for inelegant
  code to the paper, but to try to get more elegant code.}
\fi

\subsection{Trace-producing semantics}
\label{sec:beh}
We define a trace-producing semantics to reason about whole programs.
For that, we use trace properties as predicates
over complete IO traces, which are traces of terminated executions
together with the final result of the whole
program (which is an \ls{int}, like in UNIX).

\begin{lstlisting}
type trace_property = trace * int -> Type0
\end{lstlisting}
%\ca{I don't like the notation here. The paper would have multiple
%  understandings of traces and trace properties.}\ch{Then don't improperly
%  use ``trace properties'' for other things? Are they properly used anywhere else?}\ca{by now,
%  we use ``trace properties'' starting this section. the keyword trace is defined in 3.3 as a list of events}

Whole programs are computations in the monadic effect \ls{MIO},
so to reason about them we have to reveal their monadic representation using the \ls{reify} construct of
\fstar{}~\cite{relational}.
Thus, we define the following trace-producing semantics function
\ls{beh},
%, which gives us the behavior of a whole program,
used in security theorems (\autoref{sec:soundness}):
\begin{lstlisting}
let beh (wt:whole^T) : trace_property = beh_m (reify wt)
\end{lstlisting}
\ls{beh} is defined for whole target programs, but it can also be used to 
define trace-producing semantics for whole source programs.
For concision, we write \ls{beh} for both target and source programs.

Having \ls{reify} reveal the underlying representation of the computation,
we define the semantics function \ls{beh_m} over the computational monad \ls{m}.
In its definition, we use the already presented monad morphism $\theta$
%of our Dijkstra monad
(\autoref{sec:monad-morphism}), which gives
a weakest-pre-condition semantics to \ls{m},
and which we adapt into a trace-producing semantics
by using first a backward predicate transformer and
then a pre-/post-condition transformer~\cite{dm4all}:
%We do the adaptation
%We define \ls{beh\_m} that gives a trace-producing semantics to a computation
%by taking the result of $\theta$
%\ch{applied to what? and anyway,
%  what does \ls{ws} even mean in the code?
%  it was never used before in this paper; and it has 2 different types}
%of type \ls{(trace -> int -> Type0) -> (trace -> Type0)},
%and applying to it the backward predicate transformer and then the pre-/post-condition
%transformer~\cite{dm4all} so that it obtains a semantics of type \ls{trace * int -> Type0}:
% \ch{I'm still completely lost by this way too technical explanation.  I don't
%   even understand what you mean by ``apply'' (is this talking about function
%   application? and if so what is applied to what?). And what's ``the
%   pre-/post-condition transformer''?  It's not possible for me to relate
%   anything you say in the text to anything I see in the code.}\ca{I don't know how to
%   explain the transformers here easily, and I don't know how to avoid them. I tried to
%   improve the relateness between the text and the code}
\ch{Still no clue what
  ``the backward predicate transformer'' and ``the pre-/post-condition transformer''
  mean and to what they correspond in your code.}
\begin{lstlisting}
let beh_m (wt : m int) : trace_property = fun (lt,res) -> forall p. theta wt p [] ==> p lt res
\end{lstlisting}
Because we use the monad morphism when defining \ls{beh_m},
we can lift the post-condition of a computation to a trace property---for
example, for a computation $w$ of type \ls{mio int AllOps True}$~\psi$ (\autoref{sec:monad-morphism}),
where $\psi$ is a post-condition, we can state that $\mathsf{beh\_m}(w) \subseteq \psi$,
where $\mathsf{tp} \subseteq \mathsf{post}$ is defined
as \ls{forall lt r. tp (lt, r) ==> post [] r lt}. This follows naturally from the
representation of the Dijkstra monad (explained in \autoref{sec:monad-morphism}) because the 
refinement in the type of the monad becomes an extrinsic property.

\subsection{Soundness of hybrid enforcement with respect to global trace property}
\label{sec:soundness}
\ch{Here is one point I found shady about this theorem:
{\bf Is there really no connection between \ls{I}$^S.\psi$ and $\Sigma$/$\Pi$?
    So you can enforce the laxest possible policy on the context and
    you still get a strong global policy? I would be very much be surprised!}}
    \ca{of course not, tried to explain at the end}
    \ch{In fact, will the source program be able to call the source context at
        all if there is no connection between \ls{I}$^S.\psi$ and $\Sigma$?\ca{No, it
        will not be able to call. I tried to explain the connection.}}\ch{Where
      exactly is the connection between \ls{I}$^S.\psi$ and $\Sigma$/$\Pi$ explained?}
\ca{The connection is explained in:
  "this property characterizes both the behavior of the web sever and of the request handler.",
  "the behavior of the partial program is verified, including on how and if it calls the context,",
  "the access control policy and the checks give the correct specification to the context because of the type constraint"
  There is no direct answer to your question, but I think the info is there.
  Not sure how to answer your question in text.}%
In our source language, one can verify
that the behavior of the partial program
together with the context satisfy a global trace property.
One can do this by encoding the global trace property into the post-condition
of who has the initial control, so for now the partial program
(the dual setting is explained in \autoref{sec:dual-setting}).
For example, we verified that the web server ``responds to every request'' by encoding
this property in the post-condition of the web server, and this property characterizes both
the behavior of the web sever and of the request handler. Since the web server has the
initial control, its post-condition naturally becomes the post-condition of the whole
program.
%\ca{maybe this should be explained in the key ideas section}

We prove in \fstar{} that the global trace property verified in the source also holds
in the target because of the dynamic checks added by \sciostar{}. Thus, we show
that the added higher-order contracts and the enforced access control policy soundly enforce
the global trace property.

When the partial program has the initial control, its post-condition
(denoted by \ls{I}$^S.\psi$ in \autoref{fig:model_sec_comp})
is the global trace property and we show that the property holds after compilation
and linking the compiled partial program against any target context.

\begin{theorem}[Soundness]
\label{thm:soundness}
$\forall I^S.\ \forall P:\mathsf{prog}^S\ I^S.\ \forall C^T:\mathsf{ctx}^T\ \cmp{I^S}.\ \mathsf{beh}(C^T[\cmp{P}])\subseteq I\overset{S}{.}\psi$
\end{theorem}

\begin{proof}[Proof sketch.]
We unfold the compilation and the target linking and obtain that
\[C^T[\cmp{P}]\ \mathsf{==}\ P\ \mathsf{AllOps}\ (\mathsf{import}\ (C^T\ \mathsf{AllOps}\ (\mathsf{enforce\_policy}\ \mathsf{call\_io}\ I\overset{T}{.}\Pi))\ (\mathsf{make\_checks\_eff}\ I\overset{S}{.}\mathsf{cks})).\]
%\ch{Why did we write this as $\approx$ and not as \fstar{}'s ``==''?  Wouldn't
%  it be an \fstar{} ``=='' if one filled the ...?  I think it is, so changed to
%  that. Generally, this whole text was a complete mess, super confusing, and way
%  too complicated. So I tried to fix it. Please double check!}\ch{Also, since we
%  had put this on a line, could we now please show a bit more of the ...? Below
%  we talk about instantiations that are not currently shown.}%
%
Before unfolding the term $C^T[\cmp{P}]$ has the weak type
$\mathsf{whole}^T=$~\ls{unit -> MIO int AllOps True} \ls{True},
% \ch{Why did we give
%   it this too weak type? Wasn't that stupid in the first place? I guess
%   it's so that our target language looks closer to what people would expect, right?}
but the term after unfolding is an instantiation of the source program $P$,
which has the stronger type \ls{unit -> MIO int AllOps True}$~I\overset{S}{.}\psi$.
%
% underlying computations are the same -- sounds too fancy
% The two terms are, however, the same. -- too trivial, it's just an unfolding, of course they are equal
%
From the term after unfolding, we can lift the
post-condition to a trace property (as explained in \autoref{sec:beh}) and have that
$\mathsf{beh}(P\ \mathsf{AllOps}\ (\mathsf{import}\ (C^T\ ...)\ ...)) \subseteq I\overset{S}{.}\psi$,
which implies that also $\mathsf{beh}(C^T[\cmp{P}])\subseteq I\overset{S}{.}\psi$.
% since the \ls{beh} function is specification agnostic and uses only the underlying computational monad.
% CH: fancy again, all F* functions have to respect equality, otherwise F* would be inconsistent
\end{proof}

The proof follows so quickly because of the intrinsic specifications that are part of the
types of the source partial program and of the imported context.
% \ch{what source context?
%   in this theorem you're linking with a {\em target} context ... which works
%   by importing it, which is like a back-translation ... worth explaining this}
Crucially, we lifted the
post-condition of the partial program to a extrinsic property about the whole program.
Less obvious is the role of the specification of the imported context
% \ch{now I'm still confused about what context you mean, before or after
% importing? reading on it seems you still mean source context}
in the proof:
the behavior of the partial program is verified, including on how and if it calls the
context, based on the specification of the context. Thus, it is enough to pass to the
partial program a context with the correct specifications---\IE with the correct type.
The specifications are given to
% \ch{confused; ``are given to'' doesn't seem the right word here;
%   do you just mean that ``the imported(?) context satisfies these specifications''?}
the target context by instantiating the context with
the secure IO operations which enforce the access control policy, and by the
dynamic checks enforced by the \ls{import} function which is verified to be correct.
We know that the access control policy and the checks give
the correct specification to the context because of the type constraint
($I\overset{S}{.}\mathsf{importable\_ctype}$) between them.
So the proof of soundness is done modularly by typing
and also heavily benefits from SMT automation in \fstar{}~\cite{dm4all,mumon}.

% \tw{Maybe we need to properly connect with the types to show it indeed follows
% from typing?}\ca{what do you mean?}\tw{I meant, that it might be good to show 
% the types to make it clear how the property follows from typing.}\ch{+1; Right!
%   The proof sketch above sounds like blah blah without pointing to the actual types}
% \ca{better or too syntactic?}\tw{I'm still a bit confused to be honest.
% What is \(\approx\) supposed to mean above?}\ca{I try to explain what \(\approx\) means
% in the next sentence with 'The difference between the left and right term ...'. How should I
% connect the two?}\tw{Well, but you say `we obtain' which sounds like it is a
% formal statement. Also the dots make it look light hand waving, it is formalised
% in \fstar so we should make sure the reader knows it's formally proven. It might
% be better to unfold both sides separately and then say they have the same 
% internal representation despite the differences? But how is that a proof by
% typing? I'm confused.}

\iflater
\subsection{Compiler correctness}
\ca{\citet{FindlerF02} have a compilation correctness theorem, is it related?}
\fi

\subsection{Robust Relational Hyperproperty Preservation (RrHP)}
\label{sec:rrhp}

We prove that \sciostar{} robustly preserves relational hyperproperties,
which is the strongest secure compilation criterion of \citet{AbateBGHPT19},
and in particular stronger than full
abstraction\ifallcites~\cite{MarcosSurvey}\fi, as proved by \citet{AbateBGHPT19} for a
determinate setting with IO that closely matches ours.
Relational hyperproperties~\cite{AbateBGHPT19} are a very broad class of
security properties that includes trace properties (such as safety and
liveness), hyperproperties~\cite{clarkson10hyp} (such as noninterference), and
also properties that relate the behaviors of multiple programs (such as trace equivalence).
Robust Relational Hyperproperty Preservation (RrHP) states that for any
relational hyperproperty that a collection of source programs robustly
satisfies---\IE the programs satisfy this relational hyperproperty when each of them is linked
with the same arbitrary source context---then the compilation of these programs will also
robustly satisfy the same relational hyperproperty with respect to an arbitrary
target context.
Intuitively, in order to achieve such a strong criterion, the various parts of
the secure compilation framework have to work together to provide enough protection
to the compiled programs so that a linked target context doesn't have more attack
power than a source context would have against the original source programs.

% \ch{The previous subsection
%   now also makes frequent references to source contexts, but doesn't define anything.
%   Is that more okay because there is no discussion about aligning things there?
%   Or is it because those source contexts are obtained by import and whatnot?}
% CA: I think this was in the meantime fixed by referring to them as imported contexts.

\newtext{
One challenge in stating RrHP in our setting is that, as explained in
\autoref{sec:key-ideas}, \sciostar{} allows the program and the context to
recover from contract and monitoring errors.
This is necessary for functionality, as it would be unacceptable for our web
server to crash whenever a dynamic check fails.
This does trade off some security though, because the results of these dynamic checks
can potentially be used by the target context to mount some attacks that rely on
the information they indirectly reveal about the program and its secrets.
Our RrHP theorem allows this information flow that improves functionality
because we align the attack capabilities of the source and the target contexts,\footnote{
  Aligning the attack capabilities of the source and target contexts is a common necessity
  in secure compilation~\cite{AbateBGHPT19, MarcosSurvey}.}
so that the recoverable errors caused by enforcement are possible in both contexts.}
%
% \ca{Catalin, can you help with citations?}\ch{I'll try} -- CH:added some generic ones
%
We achieve this by also explicitly passing the effectful checks to the source
context, which enables the source context to mount the same kind of attacks as
the target context.

To state RrHP, we present the definitions of source contexts 
and source linking from \autoref{fig:model_sec_comp}.
The definition of source contexts (\ls{ctx}$^S$) differs from the definition of target contexts (\ls{ctx}$^T$)
in two ways. First, the type of source contexts is indexed by a source interface (\ls{I}$^S$) instead of a target interface (\ls{I}$^T$).
%but it has
Second, a source context has an extra argument for the effectful checks that we
pass to it in order to align the attack capabilities of the source and target contexts.
We pass the effectful checks to the source context during source linking
(see \ls{link}$^S$ in \autoref{fig:model_sec_comp}),
in contrast to what happens at the target level where the compiled partial program is the
one that passes the effectful checks to the imported target context.
Finally, the source context also has to be parametric in the flag because it should not
have access to the \ls{get_mstate} operation, since otherwise it would be able to
distinguish between different source programs and their secrets just by looking
at the trace, which would mean that source programs could robustly satisfy only
trace properties, but not hyperproperties, making the RrHP theorem much weaker.

%
%\ca{Explain why source contexts are still parametric in the flag?}\ch{+1; IIRC
%  this was a very interesting story: otherwise source programs would not
%  robustly satisfy any non-trivial hyperproperties, right?
%  Also, the fact that the source context is parametric in the flag is
%  used to motivate non-trivial technical issues in the last paragraph of this subsection,
%  so we would better motivate why this parametricity is needed in the first place.}
%  \ca{thanks for the reminder}

We are now ready to state RrHP and we use a property-free characterization that
was proposed by \citet{AbateBGHPT19} and that is generally better tailored for proofs:

\begin{theorem}[Robust Relational Hyperproperty Preservation (RrHP)]
\label{thm:rrhp}
$$\forall I^S.\ \forall C^T:\mathsf{ctx}^T\ \cmp{I^S}.\ \exists C^S:\mathsf{ctx}^S\ I^S.\ \forall P:\mathsf{prog}^S\ I^S.\ \mathsf{beh}(C^T[\cmp{P}])	= \mathsf{beh}(C^S[P])$$
\end{theorem}

%\ca{We use \ls{beh} for both target and source, but \ls{beh} is defined only for
%  the whole target programs.\ca{fixed in 6.2}}

\begin{proof}[Proof sketch]
Looking at the quantifier structure, to prove this theorem
% secure compilation criterion
one has to create a source context by defining a back-translation
that only takes the target context as argument. In our case, we can define back-translation by
partially applying \ls{import}
%\ch{yey, this is very cool! you're becoming a functional programmer :)}
(see \ls{back_translate_ctx} in \autoref{fig:model_sec_comp}), making it
very similar to what compilation does to the target context.
% \ch{Only the part of it that applies to the context?
%   It would be really odd otherwise, since you compilation and back-translation go in the opposite directions.}
%
These definitions of compilation, back-translation, and source linking allow us
to prove that compiling the program and linking it with the context is
syntactically equal to back-translating the context and linking it with the program:
$\forall I^S.~\forall C^T{:}\mathsf{ctx}^T\ \cmp{I^S}.~\forall P{:}\mathsf{prog}^S\ I^S.~\ C^T[\cmp{P}] = \bak{C^T}[P]$.
This syntactic inversion law is proved by just unfolding the definitions
and makes the proof of the RrHP criterion immediate.
\end{proof}

While RrHP is the strongest secure compilation criterion of \citet{AbateBGHPT19},
and such criteria are generally challenging
to prove~\cite{MarcosSurvey, AbateBGHPT19, DevriesePPK17, NewBA16, JacobsDT22},
we have set things up so that it holds by construction,
so our proof is easy.
% which is not usually the case. -- CH: just repetition
%
This simplicity is possible because (1) our languages are shallowly embedded
in \fstar{}; (2) we use flag-based effect polymorphism to model the context;
and (3) we design our higher-order contracts mechanism so that we can define
%both compilation and back-translation so that they satisfy the syntactic inversion
%law above (\ls{compile_prog} calls \ls{import} on the context,
%which we directly use to define \ls{back_translate_ctx}).
both compilation and back-translation. \newtext{These elements allow us to
avoid a sophisticated logical relations argument~\cite{AbateBGHPT19,
 DevriesePPK17, NewBA16} by instead proving a
syntactic inversion law relating not only compilation
and back-translation, but also source and target linking.
Target linking is important here, because it adds the checks that
enforce the access control policy on the IO operations of the context.
}

\newtext{
Defining the back-translation function is often the most interesting part of
the proof of such secure compilation criteria\ifallcites~\cite{MarcosSurvey, AbateBGHPT19,
  DevriesePPK17, NewBA16, JacobsDT22}\fi.
Also in our setting it is non-trivial to define back-translation:
we had to redesign the higher-order contracts mechanism (as explained
in \autoref{sec:exportable-importable}) so that back-translation 
results in source contexts that are flag-based effect polymorphic.}

\subsection{Dual Setting: the context has initial control}
\label{sec:dual-setting}

In previous subsections we formally defined the \sciostar{} framework in the
setting when the partial program has the initial control and uses the context as a
library. In this subsection, we show that \sciostar{} and its security theorems
also work in the dual setting---\IE
when the context has initial control and uses the partial program as a library. 
In this presentation we focus on the main differences compared to the previous
setting (\autoref{fig:model_sec_comp}).
At the level of the interfaces, we replace the type of the context (\ls{ctype}) with
a type for the partial program (\ls{ptype}) and redefine the notions of partial program and
context so that the context takes the partial program as argument:
\begin{lstlisting}[lineskip=0pt]
type prog$^S$ I$^S$ = fl:erased tflag -> I$\overset{S}{.}$ptype (fl+IOOps)
type ctx$^T$ I$^T$ = fl:erased tflag -> $\Sigma'$:erased _ -> io_lib fl $\Sigma'$ -> I$\overset{T}{.}$ptype fl -> unit -> MIO int fl True$~\Sigma'$
\end{lstlisting}
The type of the partial program has to be exportable, since
we have to prepare the partial program by exporting it before passing it to the context.
%The constraint on the type of the partial program is
%\ls{fl:erased tflag -> exportable (ptype fl) fl}~$\Sigma$~\ls{cks}.\ch{Is there anything else
%  interesting to see here beyond what was said in words? If yes, then say it (\IE also in words).
%  If not, I would drop what looks like (unexplained) unnecessary details.}
The definition of compilation and back-translation change so that instead of
importing the context, the partial program is exported.

For this dual setting, we had to redefine the soundness theorem.
When the partial program has the initial control,
it is statically verified to satisfy its post-condition, thus
the soundness theorem guarantees that the resulting target whole program
also satisfies this post-condition.
When the unverified context has the initial control, however,
we only know that it is monitored, thus, the soundness theorem for this setting
guarantees instead that the resulting target whole program satisfies the
specification of the access control policy $\Sigma$.

\begin{theorem}[Soundness-Dual]
\label{thm:soundness-dual}
$\forall I^S.\ \forall P:\mathsf{prog}^S\ I^S.\ \forall C^T:\mathsf{ctx}^T\ \cmp{I^S}.\ \mathsf{beh}(C^T[\cmp{P}])\subseteq I\overset{S}{.}\Sigma$
\end{theorem}

We proved this soundness theorem using the same strategy as for \autoref{thm:soundness}.
We also proved the RrHP theorem in this dual setting using the same strategy as
for \autoref{thm:rrhp}, since the same syntactic inversion law also holds in this setting.
All proofs have been machine-checked in \fstar{}.
The artifact also contains instantiations of this dual setting with some examples (one presented in \autoref{sec:other-examples}).

\subsection{Syntactic representation of target contexts}
\label{sec:syntactic-contexts}

In \autoref{sec:compilation-framework} we defined target contexts as
flag-polymorphic functions taking $\Sigma'$ and \ls{io_lib} as arguments:

\begin{lstlisting}[lineskip=0pt]
type ctx$^T$ I$^T$ = fl:erased tflag -> #$\Sigma'$:erased policy_spec -> io_lib fl $\Sigma'$ -> I$\overset{T}{.}$ctype fl $\Sigma'$
\end{lstlisting}
These functions are good enough to represent general contexts that
make calls to the IO operations provided by the \ls{io_lib}
argument.
However, this definition of contexts might not be entirely transparent
about the expressive capabilities of the target contexts and the fact
that they are unverified. To make these more apparent, we also introduce a
syntactic representation of target contexts.
For this, we define a small deeply embedded language and a total translation
function taking an arbitrary syntactic expression in this language and
translating it to a function of type \ls{ctx}$^T$\ls{ I}$^T$ representing a
target context.
Our deeply embedded language is a simply-typed lambda calculus with primitive
types (such as \ls{bytes}, \ls{int}, etc.), and we wrote the adversarial
handlers of \autoref{sec:key-ideas} also as syntactic expressions in this
language, and then translated those expressions to obtain handlers that behave,
when executed, as the original ones do\ifanon\else{} (\autoref{sec:running-case-study})\fi.
We have reproved Theorems \ref{thm:soundness} and \ref{thm:rrhp} also when
the target context is a syntactic expression, in which case target linking and back-translation
do the extra step of first translating this expression into a context with a weak interface.
% \ch{Since this is still in the formal section this last statement might ask for
%   proof, which we don't want? Could add ``in practice'' or ``when executed'' (if
%   we tested them) in front of ``behave''? In any case, we can also wait and see
%   if this stays in the formal section based on whether we can reprove the theorems.}

\newtext{
\section{Running the case study in OCaml}
\label{sec:running-case-study}
%
%\et{space: maybe skip this and leave it as a readme to the artifact?}
%\ch{If space is tight let's make this an appendix, as discussed in our Monday call.
%  Also this will be expanded with the code on initializing and updating
%  the state, which should be moved here from 5.}
%
We run our main case study by extracting both the compiled web server and 
the handlers from our target language (a subset of \fstar{})
to OCaml using the standard extraction mechanism of \fstar{} and 
test that they execute as intended when linked with wrapped up versions of the 
realistic IO library of OCaml, which reads and writes to files and network 
sockets.
While securing and formally verifying this extra step going from our target
language to OCaml is out of scope for this work, this experiment still offers
empirical evidence that several attacks attempted by the adversarial request
handlers written in our target language are blocked by the dynamic checks of
either the higher-order contracts or the reference monitor, even at the OCaml level.
Therefore, the interesting compilation step we secure and verify here works as
expected also after further extraction to OCaml, and this should provide a good base
for the more challenging task of building a
larger secure compilation framework that protects verified programs against
arbitrary contexts written in a safe subset of OCaml (\autoref{sec:conclusion}).
The web server was linked against the adversarial handlers from \autoref{sec:key-ideas}
and a non-adversarial handler that responds to HTTP requests from a real browser.
%\ch{\bf TODO Part of what we promised the reviewers we should support
%the claims we added to 2.0: the web server can handle HTTP requests from a real browser!}
%\ca{Done!}

To run the web server, we implemented the recording of IO
operations in the monitor state in \fstar{}
using a variant of the state effect that extracts to native OCaml references.
This allowed us to verify that the individual updates of the monitor
state are done correctly.
Our implementation is parametric in
%While our formal development does not verify the implementation of the
%reference monitor, it does verify that the modeling of the traces via
%the particular monitor state outlined in \autoref{sec:webserver-state}
%is sound. That is, we provide
an initial state \ls{init_st} that abstracts the empty trace,
and a verified \ls{upd_st} function
that, given a monitor state and an event, produces a new state
that abstracts the new trace.
} % newtext
\begin{lstlisting}
type init_st (mst:mstate) = s:mst.typ{s `mst.abstracts` []}
type upd_st mst = s_0:mst.typ -> e:event -> s_1:mst.typ{forall h. s_0 `mst.abstracts` h ==> s_1 `mst.abstracts` e::h}
\end{lstlisting}

\iflater
\tw{I think it would be good to also have a non malicious handler that can 
sometimes perform (or attempt) illegal operations. For instance the request 
could be a file to read from and the handler would be not safe enough because
it wouldn't check that the file is indeed in tmp. This way we can illustrate how
a handler behaves when it sometimes works incorrectly but not always.}
\fi

% \tw{Notes from the meeting.}
% \tw{mechanisms are not erased by extraction, we don't get guarantees though
% except from attacks we could already mount in \fstar, it's more to show future
% work is realistic, but it still a long-term goal. The ocaml handlers have 
% equivalent in \fstar so that's why. Not secure against arbitrary contexts.}

\newtext{
\section{Other classes of examples}
\label{sec:other-examples}

%We tried to cover as many interesting aspects of \sciostar{} as we could using our
%big running example, the web server, but of course we did not manage
%to cover all of them. \tw{Sounds a bit like an answer to the reviewers.} 
In this section, we illustrate the applicability of \sciostar{}
by presenting three more examples:
\begin{inlist}
  \item the dual setting of \autoref{sec:dual-setting} where the
    context now has the initial control;
  \item a higher-order context that returns a callback;
  \item different instantiations of the monitor state, including a stateless monitor.
    % one where the monitor state is the whole trace of past IO events. -- CH: that's 2 not 3 below
\end{inlist}

\paragraph{1) Dual setting.} Our logging IO library prints to the console %metadata about
all the IO requests and guarantees that no IO requests from
the context can happen without being logged.
In this example, the context has the initial control and it gets the
partial program as argument, where the partial program is a logging
function that writes to the console its arguments.
To obtain our desired guarantees, we pick a specification that
states that each IO operation done by the context is preceded by a logging operation
done by the program, and each logging operation must be followed by an IO operation
done by the context. We encode this as the specification
of the access control policy:
}
\begin{lstlisting}
let $\Sigma$ h caller op arg : policy_spec = match caller, op with
  | Ctx, _ -> h <> [] /\ hd h == EWrite Prog (stdout, to_string op) (Inl ())
  | Prog, Write -> h == [] \/ get_caller (hd h) == Ctx
  | _ -> False
\end{lstlisting}
\newtext{
Therefore, the context is forced to call the program before being able to do
any IO operations.
We instantiated the \sciostar{} framework with this example, thus,
thanks to \autoref{thm:soundness-dual}, the resulting
target whole program satisfies the following specification: no IO request
from the context can happen without being logged. 
%This example shows
%the benefit of picking the access control policy, since it offers
%a stronger guarantee than the post-condition of the partial program,
%which specifies that a print to the console happens.
%CA: In the end, I don't think this is a correct point. The post-condition of the partial program
%    has nothing to do with the post-condition of the context.
%\ch{Didn't get this very last part.
%Tweak it starting from what and in what way? Stronger than what?}

\paragraph{2) Archiving library.} In this example, the partial program has
the initial control and uses an untrusted higher-order function (of type \ls{zip} below) to archive files.
The untrusted function takes as first argument the file descriptor of the archive,
writes the main header of the archive to it,
and then returns another function (of type \ls{zip_file} below) that the program has to call
to add files into the archive.
The specification of the archiving function guarantees that it only reads and writes to
the files opened by the partial program.
}
\begin{lstlisting}
type zip_file = fd:file_descr -> MIO (either unit err) mymst (requires (fun h -> is_open fd h))
                               (ensures (fun _ _ lt -> only_reads_and_writes_to_fds_opened_by_prog lt))
type zip = afd:file_descr -> MIO (either zip_file err) mymst (requires (fun h -> is_open afd h))
                               (ensures (fun _ _ lt -> only_reads_and_writes_to_fds_opened_by_prog lt))
\end{lstlisting}
\newtext{
This highlights that \sciostar{} fully supports higher-order contexts,
including ones returning functions.

\paragraph{3) The monitor state.}
Our artifact contains multiple examples that show different
instantiations of the monitor state, in addition to the one of the web server from 
\autoref{sec:webserver-state}. The archiving library from 2 has as state
the entire trace of past IO events. The IO logging library example has as monitor state
only the last event that happened. We also have an example where the
monitor is stateless.
% does not have a state (possible when the dynamic checks do not
% depend on what happened previously). -- CH: obvious to me
These different examples showcase that the monitor state can be chosen flexibly.
}

\section{Related work}
\label{sec:related-work}

\paragraph{Secure compilation.}
There are many approaches to secure compilation~\cite{MarcosSurvey,AbateBGHPT19},
but to our knowledge, 
the only ones to support secure compilation of {\em formally verified programs
  against unverified contexts} are those of \citet{AgtenJP15} and \citet{StrydonckPD21}.
% \ch{\bf I'm started to get worried
%   about this ``only ones'' claim because of work of Amal, William, etc
%   al,~\cite{AhmedMWA22, SchererNRA18, KoronkevichRAB22, Bowman21, Bowman18} who
%   often state secure interoperability as fully abstract compilation to another
%   statically typed language. And sometimes their source is dependently typed~\cite{Bowman18},
%   so one could claim that they are securely compiling verified programs
%   (and then link them with verified contexts)} CH: So I added unverified contexts above!
They protect programs
verified with separation logic against adversarial contexts using protected
module architectures~\cite{AgtenJP15, AgtenSJP12, PatrignaniASJCP15} or
linear capabilities~\cite{SkorstengaardDB21}.
While they % \citet{AgtenJP15} and \citet{StrydonckPD21}
focus on stateful code and prove full abstraction,
\sciostar{} focuses on code that can perform IO, and we establish RrHP with 
 machine-checked proofs in \fstar{}.

\iflater
\ca{At a first watch, \citet{StrydonckPD21} seems to enforce only easy post-conditions.
  They have this example of a context (see Figure 2)
  that has access to a pointer to an array, and the post-condition of the context specifies that the array should have the value \ls{[0]}.
  To enforce the post-condition, they wrap the context in a new function, where at the end, they use the
  guard statement to enforce the post-condition, where ``guard gets stuck during execution if its condition evaluates to false''.
  So, this is a very simple post-condition to enforce. The fact that the context does not modify anything else
  comes from linear capabilities.}
\fi

% \ch{\bf Secure compilation of programs verified with separation logic was done
% in a couple of papers by \citet{AgtenJP15} and \citet{StrydonckPD21}. They don't
% do IO and don't do higher-order, but still, in the secure compilation of
% verified programs space these are the most related works. Eric also mentioned
% this in his email, but no, we are not the first:
%   \et{I'm not aware of any "secure compilation result" for a hybrid verification approach}}
%

\paragraph{Static verification of IO.}
%Static verification of IO was studied by many people and we're inspired by them.
%However, we stress that none of these try to solve the same
%problem as us,\ch{which problem do you talk about? we do have section 3 that's about static verification,
%  and now also a contribution about that}
%however they can be a source of inspiration for further improvements to the way we do
%static verification.
There is a lot of work on statically verifying \emph{whole} IO programs~\cite{LetanR20,DBLP:journals/jsc/MalechaMW11,DBLP:conf/esop/Penninckx0P15,DBLP:journals/corr/abs-1901-10541,DBLP:conf/nfm/JacobsSPVPP11,AmanPohjolaRM19,FereePKOMH18,GueneauMKN17}.
Interaction trees~\cite{DBLP:journals/pacmpl/XiaZHHMPZ20} were used to
define a program logic using a monad morphism in the style of Dijkstra
monads~\cite{SilverZ21}
to verify non-terminating impure computations in Coq. A case study verifies an HTTP Key-Value Server~\cite{DBLP:conf/itp/ZhangHK0LXBMPZ21} that is part of the verified operating system CertiKOS~\cite{certikos}.
The web server is written in C and the trace properties are verified in Coq, requiring the manual application of tactics to prove verification goals. \sciostar{} simplifies this kind of use cases by taking advantage of SMT
automation, yet extending the \ls{MIO} monadic effect with non-termination is future work.
Finally, all this line of related work focuses on how to verify whole programs, and does not address the problem of secure compilation.

\paragraph{Dependent interoperability.}
Strong interfaces in \sciostar{} contain refinement types and pre- and post-conditions that
can depend on function arguments and results.
Converting refinement types into dynamic checks is inspired by
\citet{TanterT15}, who introduce a mechanism based on type classes.
% to convert a refinement into a runtime check.
We extend this idea to convert pre- and post-conditions, and to go beyond pure
functions, addressing new challenges. 
% we take this further and consider not only pure functions, but also IO functions,
% which introduces new challenges that we solve in this work.
%
%
We see \sciostar{} as a first step towards achieving secure compilation from
\fstar to OCaml (\autoref{sec:conclusion}).
% To this end, strong interfaces in \sciostar{} should be more expressive
% and support full dependent types.\ch{\bf I find this claim about not supporting
%   full dependent types highly problematic without further explanations what we support and what we don't.
%   I've removed such claims from other parts of the paper; so maybe we can save some space here too? :)}\ca{
%     We do mention what we support in 4.1. I do not see why this is highly problematic,
%     but I'm fine with removing it.}
The next steps could be build upon work on {\em dependent interoperability}~\cite{OseraSZ12,DagandTT18}.
%  introduce 
%  to soundly mediate
% between dependently-typed and simply-typed expressions.
% %
% \citet{} develop a Coq framework for dependent interoperability, 
% restricted to pure programs.
% proposed a dependent interoperability
% framework in which one can relate dependent and simple types. One
% can use those relations
% to weaken a dependently-typed interface
% or to strengthen a simply-typed interface of a pure and total program.
% Beyond the special treatment we give to the refinement types and the pre- and
% post-conditions, our interfaces do not support other dependent types.

% For our long term goal, to have a secure compiler between \fstar
% and a safe subset of OCaml, integrating their work into ours could
% make our strong interfaces more expressive by supporting dependent types because
% \fstar has support for dependent types while OCaml does not.

\iflater
    \paragraph{Correct compilation.}\ch{Don't think we should have a heading on correct
    compilation, since we don't prove that, it's super broad, and it's a lot less
    related than secure compilation. If there is anything actually relevant below
    please try to move it elsewhere (maybe new more specific heading).}
    \ch{In any case please explain clearly what the connection is.}

    In our paper, we focused only on the interface through which the partial program
    and the context communicate and we did not discuss a proper compilation step.

    \ca{Mention compcert too?}\ch{Not related!}

    \ca{CakeML: compilation of pure total programs \url{https://cakeml.org/icfp16.pdf}\ch{Not related!}
    and more recent that includes references, IO and exceptions \url{https://cakeml.org/ijcar18.pdf}}\ch{Maybe
    related to static IO verification? They seem have more related work in that space though (we cited it at HOPE)}
\fi

\paragraph{Higher-order contracts.} 
\citet{FindlerF02} pioneered higher-order contracts, now a standard feature of Racket~\cite[Chapter 8]{Racket}.
Several works have explored extensions to stateful contracts, \EG \citet{DisneyFM11}
propose temporal higher-order contracts, \citet{ScholliersTM15} propose computational contracts,
and \citet{TovP10} study stateful contracts for affine types.
Stateful contracts can be added at the boundary between
the partial program and the context, and they can be used to implement an IO
reference monitor.
%
% The latter is more clearly shown by 
In particular, \citet{MooreDFFC16} propose authorization
contracts for implementing access control monitors for software components.
%\citet{MooreDFFC16} propose an expressive framework for implementing access
%control monitors for software components based on authorization contracts, which
%are contracts that manage authority environments associating rights with an
%execution context.
%
%\ch{rough transition; unclear what the connection is until the end of the phrase}%
%However, as explained in \autoref{sec:intro} certain post-conditions of the context can be enforced using
%higher-order contracts at the boundary, but other post-conditions we have to turn into an access control policy enforced
%on each IO operation by a reference monitor, support we do not see explained in
%previous work.\ch{``we do not see explained'' sounds very weak}%
%\ca{I went with the soft wording because it is still not clear to me from Related work 
%  what they can do by having both HOC at the boundry and the HOC used for monitoring.
%  Can they enforce at the boundry contracts based on information gathered by the monitor?
%  I could not find a convincing explanations, but \citet{DisneyFM11} does
%  'module linking', which to my poor understanding/guessing seems to be
%  connecting the reference monitor to the contracts on the boundary.}\ch{Is such
%  a very weak claim worth making at all then? Isn't it watering down the
%  stronger argument below?}
%
The key novelty of \sciostar{}
% in this regard 
is to integrate
statically-verified code with untrusted code in a tool that is itself verified.
% Our setting is different than that of this related work
% %\ch{this related work?}
% because we use higher-order contracts in a hybrid setting where
% statically verified code is mixed with untrusted code, a setting
% not studied in previous work.% and which has many advantages,\ch{``has many
%  advantages'' sounds overclaimed, biased, and inflammatory. Do you think that
%  Racket has no advantages over what we do? And in any case trying to compare
%  settings seems as useful and meaningful like comparing apples to oranges to me.}
%\EG it allowed us to avoid redundant checks
%and more importantly, it allowed us to propose a solution itself formally verified.
%
%Moreover, we need a form of stateful higher-order contracts to record all
%prior IO events, but certain post-conditions of the context have to be turned
%into access control policies enforced on each IO operation by a reference monitor,
%creating a relation between the contracts and the monitor that was also not studied
%before.
%
%Our setting for secure compilation is in a way simpler,
%since we only have two components, the
%verified program and the untrusted context, and our access control policies are
%only enforced on the IO actions context~\ca{this is not because our secure compilation
%  is simpler}, in a relatively straightforward way.
It would be interesting to explore whether soft contract
verification~\citeFull{NguyenTH14}{NguyenGTH18, MoyNTH21} could be used in \sciostar{} to
% further reduce runtime overhead, by
eliminate dynamic checks that can be verified statically.
% , which tries to statically verify
% as many contracts as possible in order to reduce runtime overhead and
% discover errors earlier.
% %
% It would be interesting to see if ideas from soft contract verification
% could also be used for reducing the number of dynamic checks added on our
% untrusted contexts.

\iflater
\ca{\citet{NguyenGTH18} do soft contract verification for higher-order stateful program,
  but it is not clear if they tried to store information in the contract that 
  can be later used when enforcing the contract. aka keeping track of
  the IO operations that happened. Also, their technique requires the whole program
  and also they have to manually inspect false positives.}

\ch{I somehow have a vague memory that in Racket(?) they have ways to only
  put/check contracts on component boundaries? Should try to dig that up.}\ca{You
  mean that trusted code by-passes the contract while untrusted code calls the
  function with the contract?}\ch{I mean that calls within the same module don't get
  checked, while calls to other modules get checked.}
\fi

\iffalse
\ch{Since we are tight on space I think that we can drop the part about
  reference monitoring.  It's papers from the 70s and we are already citing them
  in the text. I don't see a big need to compare.}
\paragraph{Reference monitoring and runtime verification.}
There is extensive prior work on how to monitor program executions either by using reference monitoring or instrumentation.
%
Reference monitoring was introduced as a technique to detect and prevent unauthorized access to resources by enforcing access control policies in secure kernels~\cite{Anderson73,UNIX,Ames81}. In a secure kernel, the reference monitor is part of the kernel to make sure that all {\em untrusted processes} call the kernel through the monitor, while the {\em trusted processes} are allowed to by-pass the monitor.
%
In such a setting, one could define the partial program as a trusted process and the context as
an untrusted process and use inter-process communication.
%
\sciostar{} does not treat the partial program and the context as two separate processes, avoiding the need for inter-process communication.
%
Moreover, \sciostar{} does not require a secure kernel because the reference monitor is at the application level.
%
% A different approach to add dynamic checks to untrusted code is by via code instrumentation.
\fi

\paragraph{Runtime verification.}
There is extensive prior work on how to monitor program executions either by using runtime verification and instrumentation.
%\ifsooner\ch{do we have more citations we can give here (beyond mop)? Any survey for instance?}\fi{}
%CA: i looked for more, but nothing that is more relevant.
%
Java-MOP is a framework for Monitoring-Oriented Programming (MOP) that builds on Aspect-Oriented Programming (AOP) in Java~\cite{Chen2004,Chen2005,jin-meredith-lee-rosu-2012-icse}.
\newremove{The partial program and the context run inside the Java Virtual Machine.
Monitoring is obtained via
instrumentation of the whole program, synthesized from formal specifications.
MOP supports dynamic checks on module boundaries, but only for first-order interfaces.
There is no mechanized proof of soundness and no secure compilation criterion is considered.
}\newtext{
They focus on synthesizing the monitor from formal specifications (\EG in LTL) to instrument
the whole program. Automatically synthesizing the dynamic checks is an open
challenge in our setting. In this paper, we focus instead on highly
expressive specifications, on the programmability and efficiency of the checks,
and on obtaining formal security guarantees, by automatically verifying that the
checks are always enough to guarantee the specifications in the strong interface.}
%\ch{\bf Here is one
%  point where we could explain that they focus on synthesis, while we focus on something else. This
%  comparison seems anyway too one-sided / unfair, so it might be what triggered the reviewer.}%

% \ch{Added newline here, since otherwise we need a better transition}
\newtext{Our work currently only deals with a single policy. Prior work on
flow-based monitors---intensively studied in the AOP literature from both the
points of view of expressiveness~\cite{Tanter08,dfs04a} and
efficiency~\cite{MasuharaKD03,AvgustinovTM07,BoddenHL07}---could be
helpful to extend our work to support multiple, localized policies for
specific contexts and linking points.}
%
% In their framework, ``monitors are automatically synthesized from formal specifications and
% integrated at appropriate places in the program'' and they show how
%
% Moreover, the JVM with AOP enabled is a massive and complex
% project and for now there are no alternatives for the languages \fstar{} extracts to.
%
\begin{comment}  %ET: skipping because not central / space saving
In earlier iterations of our work, we attempted to model in \fstar{} instrumentation of whole programs
and we found it very difficult because we wanted to avoid adding redundant checks,
but for that we needed to discriminate between the verified and the unverified code and
the solutions we thought about were complex and ad-hoc.
%
We think our current solution is more straightforward since it uses
higher-order contracts and reference monitoring.
\end{comment}

%The topic of runtime verification has studied extensively how to synthesize monitors
%and also how to make them more efficient.
%
%Because in \fstar we model the monitor abstractly,
%there is some theoretical flexibility on how the monitor can be implemented during
%extraction.
%
%There are many works\ca{TODO:cite} that
%synthesize monitors as automata from formal specification.
%We chose to stick with traces because
%it is much easier to think in terms of traces and it also
%fits better with the problem of secure compilation.
%Extending our work with synthesizing an efficient monitor from the formal specification
%is subject to further research.

\iflater
\ch{[Moved from PriSC, but still unclear where this belong (if anywhere)]
  The big advantage of static verification over dynamic verification is
  getting static verification errors and much stronger guarantees once things
  verify, instead of runtime errors.}\ch{Looking at the PriSC reviews the relative
  advantages of static vs dynamic verification could be discussed here.
  One more advantage of static verification is that it doesn't add runtime
  overhead, while dynamic verification does.}
\fi

\paragraph{Gradual verification.}
\citet{BaderAT18} and \citet{WiseBWATS20}
propose gradual program verification to easily combine dynamic and static
verification in the same language at a very fine granularity, using variants of Hoare logic or separation logic with {\em imprecise} logical formulas. \sciostar{} is a hybrid verification and compilation framework with a coarser interoperability granularity (program vs. context), and targets the IO effect.
% A main difference is that our work 
% gives a model that combines dynamic and static verification in a source and a
% target language, at a much coarser granularity: program vs context.
% The other main difference is that we focus on code performing IO.
%
%Related to gradual typing and verification, \citet{JacobsTD21} have recently proposed the full abstraction of the embedding from
%the static to the gradual language as a criterion to characterize an adequate gradual language.\et{why is that cited? their work isn't about (gradual) verification but only about (gradual) typing.}
% \ch{Less related: There is also work by \citet{JacobsTD21} that proposed full
% abstraction as a helpful property gradual typing world. Probably not super
% related, but may want to still have a look\ca{does not seem related, I added
% the presentation to the bib file if you want to take a look too}. And there's
% also \citet{SiekTW21}.}

\iffull
\paragraph{Reasoning about robust safety.}
% One can reason about what safety properties a program
% satisfies robustly. Reasoning about such properties should be easier in \fstar than
% in a target language. We showed that our compilation framework satisfies robust relational
% hyperproperties preservation, therefore if one proves something\ch{?} is robustly satisfied in the source language,
% it will be robustly satisfied in the target language. This topic focuses on devising languages that
% robustly satisfy properties by construction. Most of them focus on memory-safety.\ch{Strange to
%   have a full paragraph of related work without any citations. Do we really need this heading?}
% \ch{Could we just move this to a better heading? It does seem related.}\ca{In their paper,
%  they say they do robust safety quite a lot.}
%
Interoperability between trusted and untrusted code has also been studied when
reasoning about robust safety~\cite{KupfermanV99, GordonJ04, SwaseyGD17}.\ca{explain how this is related
  with RrHP}\et{+1!}\et{also add ref for robust safety?}\ch{added standard refs}
\citet{SammlerGDL20} show that sandboxing enables reasoning about robust safety
when having a rich type system that contains the {\tt Any} type, a type inhabited by all values, and a higher-order contract mechanism going between types and {\tt Any}.
A main difference with \sciostar{} is that in their work all interactions between the trusted and the untrusted code happen through the {\tt Any} type. Also, they discuss only robust safety related to the
memory model, not trace properties.
\fi

\iflater
\ca{
\paragraph{Extraction/compilation from proof assistants}

The need to extract/compile code written and verified in proof assistants has multiple examples of compilers from Coq or F* to OCaml, C, etc. This show that exists a need, and
most of the work does not contain any secure compilation criteria.
\ca{cite Low*}\ch{Now we already cite these in the intro}

\citet{KoronkevichRAB22}
present a dependent-type-preserving translation to A-normal form (ANF)
for a subset of Coq. \tw{They compile to a language which doesn't have defined
operational semantics though: Coq extended with equality reflection.}
A different more practical approach is taken by Yannick et al.
that try to create a correct compiler for Coq programs that target
the untyped Malfunction language (one of the intermediate languages of the OCaml
compilation framework).
}
\fi

% \ifsooner
% \ca{Rustan's work on Dafny: \url{https://doi.org/10.1007/978-3-031-08166-8_6}
%   Specifying the Boundary Between Unverified and Verified Code.}
% \ch{Cited this when we talk about explicit assumptions (i.e. specs) in the intro}
% \fi

\section{Conclusions and Future Work}
\label{sec:conclusion}
%\et{for space, we might consider submitting without conclusion (not uncommon)}
%\ch{Future work is a good defense against ``why didn't you do X?''; Answer: ``well, X
%  was explicitly mentioned as future work, while explaining the involved challenges.
%  It's something that takes 1+ papers to do, not something we can hack up in an afternoon.
%  And yes, we did seriously think about it. Thank you for your review.'' :)}
%
\newremove{
In this paper we were the first to secure verified IO programs against unverified adversarial
code, which we achieved by building \sciostar{}, a secure compilation
framework that we implemented and formally verified in \fstar{}.}
\newtext{
Securely compiling verified code and linking it with adversarial unverified code
is a general open research problem for all proof-oriented programming languages
(Coq, Isabelle/HOL, Dafny, etc). In this paper we make substantial progress
towards solving this important problem by proposing a formally secure
compilation framework for protecting verified IO programs against unverified
code. In particular, we provide machine-checked proofs that our framework
soundly enforces a global trace property and, moreover, satisfies a secure
compilation criterion called RrHP, which is stronger than full abstraction.}

We see this as an important first step towards building a larger formally secure
compilation framework from \fstar{} to a safe subset of OCaml.
Protecting compiled programs against arbitrary OCaml contexts is challenging,
and would go beyond the state of the art in formally secure compilation, which
has so far never been achieved for such realistic languages, but we believe that
our current work can serve as a good base for that.
The side effects of OCaml will need to be exposed in the source language,
which requires a significant extension of the framework,
% will lead to a significant extension of the languages and specifications,
yet our free monad representation for MIO is general enough to allow expressing
more of the side effects of OCaml such as non-termination,
exceptions,\footnote{\newtext{Another interesting idea one could investigate is
    preventing the context from catching exceptions raised by our enforcement
    mechanism, so as to reduce the amount of information these errors can reveal
    to the context and thus to further strengthen our secure compilation theorem
    (\autoref{sec:rrhp}).}}
% CH: I think the original version was better, so reverted to that
% \footnote{\newtext{An interesting idea to investigate is
%     to further strengthen our secure compilation theorem (\autoref{sec:rrhp}) by
%     preventing the context from catching exceptions raised by the enforcement
%     mechanism and thus reducing the information revealed to the context.}}
and state.
The secure compilation proof techniques for shallowly embedded languages we
designed in this paper will have to be combined with proof techniques for deeply
embedded languages~\cite{MarcosSurvey,AbateBGHPT19}, and the security
enforcement and proof will have to be extended to include \fstar{}'s extraction
to OCaml, which is similar to that of Coq~\cite{Let2008}.
The hope is to be able to reuse the recent correctness proof of Coq
extraction~\cite{SozeauBFTW20, sozeau:hal-04077552} in a bigger secure
compilation proof~\cite{AbateABEFHLPST18,El-Korashy0PD0P21}, and our strategy from
\autoref{sec:syntactic-contexts} for defining the back-translation.

\iflater
\ch{Do we want to mention
  here that adding exceptions would also allow us to have a less clunky
  treatment of contract failures? See discussion about this and its connection
  to blame at the end of 4.0}\ch{Moved here:}
\gm{\bf There is no need for blame assignment: the context is the only one who can be at fault.}
\ca{Guido made this claim which I removed from the submission. I don't think we ever discussed about this,
  but it something that they always discuss in a higher-order contracts paper. What can we say about it?}
\ch{It's true that the context is always at fault, but I thought that blame is a
  bit more fine-grained than that, identifying the precise code/contract that caused
  the failure. One thing that's already clear: it will be a serious tradeoff between
  on the one side providing better error messages when things dynamically fail
  and allowing for recovery and on the other side leaking private information
  from the program to the context. Already making contract errors recoverable
  leaks more information to the context than fail-stopping would. We could study
  such interesting questions in the context of noninterference for the context,
  but first we would need to make that to work. Also, once we add exceptions
  will we be maybe able to only allow the program to catch contract failures,
  but not the context? Anyway, I agree we should think about something smart we
  can say about blame at the current moment ... future work? :)}
\fi

\newtext{
While our contributions are developed in the context of \fstar{}, we believe
that many ideas of the paper could be applied to solving the same general
problem also in other proof-oriented programming languages, especially those
based on dependent type theory.
% As mentioned in \autoref{sec:mio-essence},
Dijkstra monads have, in particular,
also been implemented in Coq~\cite{dm4all}, so in principle a direct port to Coq of
our whole \sciostar{} framework seems possible, just that the \ls{MIO} Dijkstra monad will not be
as usable in Coq without support for discharging the gathered verification
conditions with an SMT solver, which is a big advantage we get from using \fstar{}
instead. For Coq a more usable solution could probably be built more quickly by
targeting an interactive verification framework for IO programs, such as Ynot~\cite{DBLP:journals/jsc/MalechaMW11},
FreeSpec~\cite{LetanR20}, ITrees~\cite{DBLP:journals/pacmpl/XiaZHHMPZ20}, or Iris~\cite{DBLP:journals/corr/abs-1901-10541}.
While the precise details will vary
based on the choice of verification framework, some general ideas that we expect
to be portable to other frameworks are:
\begin{inlist}
\item the internal representation of \ls{MIO} as a free monad with an extra GetMState
  operation, which is directly compatible with how FreeSpec and ITrees represent
  computations, and which could provide a model to the axioms used by Ynot;
\item the formally verified combination of higher-order contracts and reference monitoring;
 % (we implement this with type classes, which are standard for proof
 %  assistants these days, and we expect one can still verify this combination
 %  using the chosen verification framework);
\item the secure compilation proof of RrHP from \autoref{sec:rrhp}, which should stay
  simple even in these frameworks, since all the main ingredients seem portable;
  % (shallow embeddings, higher-order contracts, and the syntactic identity law);
\item the soundness proof from \autoref{sec:soundness}, although this will likely become more manual.
\end{inlist}
}

\iffull
Another interesting line for future work would be to use parametricity to prove a
noninterference theorem formalizing that our flag-based effect polymorphic
context cannot directly call \ls{GetMState} or the IO operations.
There are, however, at least two challenges to overcome for achieving this:
(1)~our noninterference statement is significantly more complex than prior work
in this space~\cite{AlgehedB19,AlgehedBH21}; and
(2)~the \ls{erased} type, its interaction with the primitive \ls{Ghost} effect
of \fstar{}, and their parametricity properties would first need to be better
formally understood.
\fi

% \tw{Mention parametricity for flag-based effect polymorphism as future work.
% Challenging because of ghost effect we need to properly understand.}


\appendix

\section*{Data Availability Statement}
This paper comes with an artifact in \fstar
that contains a formalization of the contributions above.
% \footnote{\url{}} \ch{No URL; they expect such artifacts to be uploaded on submission.}
% IO computations are implemented using the
% Dijkstra Monad~\cite{fstar-pldi13,dm4all}
% of \citet{fstar-hope22}.
The artifact contains the \sciostar{} framework, the mechanized proofs of sound
enforcement of a global trace property and of RrHP, as well as a few examples.
The artifact is available on Zenodo \cite{artifact_zenodo} and on Github \cite{artifact_github}.

\ifanon\else
\begin{acks}
We thank the POPL 2024, ICFP 2023, PriSC 2023, and ICFP SRC
2020 referees for their helpful reviews.
This work was in part supported
by the \grantsponsor{1}{European Research Council}{https://erc.europa.eu/}
under \ifcamera\else ERC\fi{} Starting Grant SECOMP (\grantnum{1}{715753}),
by the German  Federal Ministry of Education and Research BMBF (grant 16KISK038, project 6GEM),
and by the Deutsche Forschungsgemeinschaft (DFG\ifcamera\else, German Research Foundation\fi)
as part of the Excellence Strategy of the German Federal and State Governments
-- EXC 2092 CASA - 390781972.
E.R. was supported by the Estonian Research Council starting grant PSG749.
\end{acks}
\fi

\ifcamera
\bibliographystyle{ACM-Reference-Format}
\else
\bibliographystyle{abbrvnaturl}
\clearpage
\fi
\bibliography{fstar}

\end{document}